\documentclass[12pt]{article}

\usepackage{graphicx} 
 \usepackage{setspace}
 \usepackage{amssymb}
\usepackage{amsmath}
 \usepackage{bm}
 \usepackage{graphicx}
 \usepackage{mathrsfs}
  \usepackage{fancyhdr}
   \usepackage[colorlinks, linkcolor=blue]{hyperref}
    \usepackage{amsthm}
     \usepackage{booktabs}
 \usepackage{makecell}
 \usepackage{float}
 \usepackage{footmisc}
 \usepackage{tikz}
 \usetikzlibrary{external}
 \usepackage{appendix}
\usepackage{pgfplots}
\pgfplotsset{compat=1.18}

\usepackage{listings}
\usepackage{tocloft}
 \usepackage{wasysym}
 \usepackage{ragged2e}
\PassOptionsToPackage{backref=page}{hyperref}
 \usepackage{comment}
 \usepackage[ruled, vlined]{algorithm2e}

\usepackage{natbib}
\usepackage[backref=page]{hyperref} 

\usepackage{amsfonts}
\usepackage{setspace}
\usepackage{amsmath,amsthm}
\usepackage{txfonts}
\usepackage{amssymb}
\usepackage{graphicx}
\usepackage{epsfig}
\usepackage{psfrag,pstricks}
\usepackage{caption}
\usepackage{soul}
\usepackage{apptools}
\usepackage{ wasysym }
\usepackage{xfrac}
\usepackage{titling}
\usepackage{booktabs}
\usepackage{dcolumn}
\usepackage{multirow}
\usepackage{dsfont}

\hypersetup{
    citecolor=blue,
    colorlinks=true,
    linkcolor=blue,
    filecolor=magenta,      
    urlcolor=blue,
}

 \newtheorem{proposition}{Proposition}
 \newtheorem{example}{Example}
 \newtheorem{remark}{Remark}
  \newtheorem{theorem}{Theorem}
  \newtheorem{definition}{Definition}
  \newtheorem{lemma}{Lemma}
  
  \newtheorem{assumption}{Assumption}
 
  \newtheorem{corollary}{Corollary}

    \DeclareMathOperator*{\argmax}{argmax}
        \DeclareMathOperator*{\argmin}{argmin}
        \DeclareMathOperator*{\sign}{sign}
        
        \newcounter{savefootnote}
\newcounter{symfootnote}
\newcommand{\symfootnote}[1]{%
   \setcounter{savefootnote}{\value{footnote}}%
   \setcounter{footnote}{\value{symfootnote}}%
   \ifnum\value{footnote}>8\setcounter{footnote}{0}\fi%
   \let\oldthefootnote=\thefootnote%
   \renewcommand{\thefootnote}{\fnsymbol{footnote}}%
   \footnote{#1}%
   \let\thefootnote=\oldthefootnote%
   \setcounter{symfootnote}{\value{footnote}}%
   \setcounter{footnote}{\value{savefootnote}}%
}

\usepackage{microtype}

\begin{document}

\begin{spacing}{1.2}

\begin{titlepage}
    \centering
    \vspace*{1cm}
    
    \Large
    \textbf{A Constructive Characterization of Optimal Bundling} \symfootnote{This paper was previously circulated as ``Optimal Bundling and Dominance.'' I am indebted to Laura Doval, Navin Kartik, and Qingmin Liu for their continued guidance and support. I am grateful to Elliot Lipnowski, Jacopo Perego, Kai Hao Yang, and Yangfan Zhou for valuable discussions. I also thank Yeon-Koo Che, Zihao Li,  Ellen Muir, Axel Niemeyer, David Pearce, Matt Weinberg, and participants at conferences and seminars for helpful comments. }
    
    \vspace{0.5cm}
    \large
    Zhiming Feng \symfootnote{Department of Economics, Columbia University. Email: \href{zf2295@columbia.edu}{zf2295@columbia.edu}.}

    \vspace{0.5cm}
    \large
    \today
    
    \vspace{1.5cm}

    \normalsize
    \textbf{Abstract}
    \vspace{0.5cm}

\justifying
This paper studies a monopolist selling multiple goods to a consumer with one-dimensional private types. I provide a sufficient condition---single-crossing differences of virtual values together with monotonic differences of valuations---under which the monopolist's problem is equivalent to finding the upper envelope of the marginal revenue curves. This approach guarantees that the optimal mechanism is deterministic and can be implemented via a menu of bundles. I further characterize this upper envelope using a dominance notion. This characterization yields a constructive algorithm that computes the unique optimal menu by iteratively eliminating dominated bundles. As my main application, I use this framework to introduce and provide sufficient conditions for the optimality of tree bundling---a common but previously unmodeled sales strategy where the optimal menu contains a ``root'' bundle but features distinct upgrade paths. This structure captures prevalent sales practices across industries---from automobile manufacturers offering base models with customizable upgrade packages, to software companies allowing modular feature additions.
    \vfill

\end{titlepage}
\pagenumbering{gobble}
{\small
\tableofcontents
}

\newpage
\pagenumbering{arabic}
\setcounter{page}{1}
\section{Introduction}

The optimal bundling problem---designing a menu of bundles and prices to maximize profit---is a fundamental question for multi-product firms. Despite decades of research, the problem remains remarkably intractable. Early work focused on multidimensional type settings with additive values (\cite{mcafee1987}, \cite{manelli2006}, \cite{manelli2007}), where the complexity of characterizing optimal mechanisms led researchers to seek special cases and approximations. Relatively general results are either obtained by characterizing properties of and certifying certain optimal mechanisms (\cite{daskalakis2017}) or by focusing on finite types (\cite{bergemann2021}, \cite{kwak2026separate}). More recent analyses have restricted attention to one-dimensional type spaces but allowed for non-additive values (\cite{ghili2023characterization}, \cite{yang2023nested}), trading off the dimensionality of private information for richer substitution and complementarity patterns across products.

Even with consumer heterogeneity restricted to a single dimension, the seller's optimization over $2^n$ possible bundles when selling goods $\{1,...,n\}$ remains a formidable mechanism design challenge. The difficulty lies not just in the number of possible bundles, but in ensuring that the resulting allocation satisfies incentive compatibility constraints. Consequently, the literature has often imposed additional structure to make the problem tractable, restricting attention to specific menu structures like pure bundling (offering only the grand bundle) or nested bundling (offering a sequence of increasingly comprehensive packages).

This paper develops a constructive solution to the one-dimensional type, non-additive value problem without focusing on specific structures on bundles. My approach begins with a key and well-understood insight: rather than directly tackling the constrained optimization problem, I analyze the relaxed, constraint-free problem of maximizing expected virtual surplus. I provide a sufficient condition (Assumption \ref{a0})---single-crossing differences of virtual value functions together with monotonic differences of valuations---under which the solution to this relaxed problem is guaranteed to be implementable and thus solves the original constrained problem. This condition, which jointly restricts consumer valuations and the type distribution, effectively reduces the complex mechanism design problem to the much simpler task of finding the \emph{upper envelope} of virtual value functions.

The connection between relaxed and constrained problems builds on classical insights from mechanism design (\cite{myerson1981optimal}), but the non-additive setting introduces new challenges. Unlike standard single-good problems where monotonicity of virtual values suffices, or multi-good problems with additive values where bundling can be analyzed component-wise, the interaction between goods requires careful analysis of how virtual value and valuation differences evolve across types. Assumption \ref{a0} provides exactly the structure needed to ensure that the pointwise maximization approach yields an implementable mechanism.

While the solution to the relaxed problem is conceptually simple---pointwise maximization of virtual values---this does not immediately reveal the structure of the optimal menu or offer a practical way to compute it. To bridge this gap, I develop a constructive characterization using certain notions of dominance and best response. Drawing inspiration from \cite{pearce1984}, I establish a key equivalence: if virtual value functions satisfy \emph{single-crossing differences in convex environments} (SCD$^\star$) as defined by \cite{kartik2023} (Assumption \ref{a1}), then a bundle is a ``never strict best response'' if and only if it is ``very weakly dominated'' by a (possibly stochastic) alternative. The classical arguments for such equivalences apply to payoffs that are linear in the underlying state; here the virtual value is a nonlinear function of the type, and the single-crossing structure substitutes for linearity.

This dominance-based characterization yields a practical algorithm that computes the unique minimal optimal menu by iteratively eliminating dominated bundles. The algorithm first removes bundles dominated by pure alternatives, then identifies and eliminates those dominated only by stochastic combinations. Unlike existing characterizations that merely describe properties of optimal menus, this approach constructively generates the solution. The algorithm runs in polynomial time in the number of bundles, though this remains exponential in the number of goods---a complexity that appears inherent to the problem without further structural assumptions.

The economic content of Assumptions \ref{a0} and \ref{a1} merits discussion, which I provide in Section \ref{sec:discussionofassump}. Both are implied by a single stronger condition---monotonic differences of virtual values over stochastic bundles, \`{a} la \textit{monotonic differences in convex environments} (MD$^\star$) as defined by \cite{kartik2023}---whose characterization leads to a multiplicative-additive form of the valuation: each bundle is summarized by a quality index and a baseline term, and each type by a marginal value for quality. Under this form, the assumptions hold whenever the marginal value for quality, adjusted by the standard information-rent term, is monotone in the type---Myerson's regularity condition, applied to the taste for quality rather than to the type itself. I give two examples of settings under which the condition holds: a concave generalization of \cite{mussa1978}, and a specification following \cite{HopkinsKornienko2004} in which willingness to pay depends on the consumer's rank in the type distribution rather than on the type itself---the latter requiring no regularity of the distribution.

Equipped with this constructive methodology, I analyze the optimality of common but previously unformalized sales strategies. My main application introduces and characterizes \emph{tree bundling}. A tree menu is organized around a root bundle---a base product contained in every other bundle on the menu---from which the remaining bundles branch out along different upgrade paths. Consider a computer retailer selling a base laptop, the laptop with a graphics upgrade, the laptop with extended storage, and a fully loaded machine: every option builds on the base laptop, but the graphics and storage upgrades lie on different branches, since neither contains the other. A tree menu thus lets consumers who agree on the base product part ways on how to upgrade it. Nested bundling---a single chain of increasingly comprehensive packages, as in \cite{yang2023nested}---is the parallel case in which the menu has exactly one branch.

I provide two types of results for tree bundling. First, I establish conditions under which the optimal menu is organized around a root: when a root bundle has the highest \textit{sold-alone quantity} and provides higher value than any bundle not containing it to the highest type, the minimal optimal menu is either a tree or a nested menu (Proposition \ref{p33}).\footnote{The sold-alone quantity of a bundle (formally defined in Definition \ref{d4}) is the quantity that a profit-maximizing monopolist would sell if she were selling the bundle alone.} Second, I characterize when the menu genuinely branches. In Theorem \ref{t33}, I provide a condition on the \emph{ratios} of virtual values under which every possible upgrade path from the root bundle to the grand bundle appears in the optimal menu; a worked example (Example \ref{ex:tree}) with a graphical illustration of the upper envelope accompanies the result. This is the ``build-your-own'' model ubiquitous in industries from automobiles to computers, where customers start with a base model and assemble their preferred combination of upgrades. Theorem \ref{t34} provides an alternative characterization through a ``least-favorite good''---one whose removal from the grand bundle, or addition to the root, matters least to the top type---which generates two distinct upgrade paths. These results formalize the intuition that tree bundling is optimal when different customer segments value different upgrade dimensions, yet all agree on a common base product.

Beyond tree bundling, the framework provides a unified lens for analyzing a range of classic bundling structures. It recovers and extends results for pure and nested bundling --- in particular, I show that the results of \cite{ghili2023characterization} and \cite{yang2023nested} can be recovered under a common framework (Theorem \ref{corollary: alternative}): it suffices that the conditions of Theorem \ref{t1} hold on the restricted bundle class each paper targets---the pure-bundling menu and the nested bundles, respectively---rather than on all bundles. I also show that how conditions based on sold-alone quantities, such as the union quantity condition, ensure the optimality of these menus using my approach. For the case of additive values, I provide a complete characterization, demonstrating that the optimal menu is a specific nested structure determined by a strict ordering of virtual value ratios for singleton goods. These applications demonstrate that the dominance-based characterization is not just a solution for a specific case but a general and constructive methodology for a broad class of optimal bundling problems.

\subsection{Related Literature}

\textbf{Optimal Bundling and Product Lines:} The most closely related papers are \cite{ghili2023characterization} and \cite{yang2023nested}, which characterize optimal bundling through the upper envelope of marginal revenue curves. While both papers make significant advances, they focus on verifying when specific menu structures---pure bundling and nested bundling, respectively---emerge as optimal.

These papers employ two types of restrictions. First, they impose structural conditions on preferences: monotonicity (and monotonic differences) of valuations and quasi-concavity (and quasi-concave differences) of revenue curves, which translate to single-crossing (and single-crossing differences) properties of virtual value functions. Importantly, these conditions need only hold for restricted bundle classes---e.g., nested bundles in \cite{yang2023nested}---since each paper targets a particular menu structure. Second, they require specific shapes for the upper envelope of marginal revenue curves. \cite{ghili2023characterization} shows pure bundling is optimal when the grand bundle's marginal revenue dominates all other non-empty bundles (equivalently, has the highest sold-alone quantity). \cite{yang2023nested} establishes that nested bundling emerges when the ``nesting condition'' ensures the upper envelope is generated by nested bundles.

The upper-envelope approach itself has a long history in the product-line literature. \cite{mussa1978} solve the monopolist's quality-choice problem, \cite{johnson2003, johnson2006} develop the ``upgrades'' approach for multi-product lines, and \cite{AndersonCelik2015} characterize the optimal product line as the upper envelope of marginal revenue curves under single-crossing of inverse demands and of marginal revenues---conditions that correspond to the two parts of Assumption \ref{a0}. The difference lies in where the conditions are imposed. In those papers, the products come with a one-dimensional order---qualities on the real line, or a chain of nested packages---and the monotonicity conditions are imposed along that order, so the order in which products appear on the optimal menu is part of the model's input. Here the conditions are imposed on all pairs of the $2^n$ bundles; the order in which bundles appear on the upper envelope is determined by the conditions rather than given, and it need not respect set inclusion. Whether it does---whether the optimal menu is a tree, a chain, or a single bundle---is what the applications characterize.

Rather than verifying optimality of predetermined menu structures, I ask: when can the optimal mechanism be characterized through upper envelope methods, and how can we constructively identify which bundles constitute this upper envelope? This requires conditions on all bundles, not just special subsets, which consequently gives a complete characterization: I identify sufficient conditions for upper envelope methods to solve the bundling problem regardless of the structure of the bundles that constitute the upper envelope, and provide an algorithm that computes the optimal menu for any problem satisfying these conditions. Rather than checking whether a specific bundling structure is optimal, the constructive nature of my approach allows me to systematically compute the optimal menu. The structural results in the applications are then read off this computation: conditions for tree, nested, or pure bundling describe when the elimination procedure outputs a menu of that shape, rather than serving as verification of a menu conjectured in advance. 

\textbf{Dominance and Best Response:}
The connection between dominance concepts and best response behavior has deep roots in game theory. \cite{pearce1984} establishes a fundamental result: when opponents' strategies may be correlated, a strategy is a never best response if and only if it is strictly dominated. 
\cite{weinstein2020} extends this framework to best-reply sets, generalizing the equivalence to never strict best responses and very weak dominance. Alternative proof techniques have been developed, including the duality approach of \cite{myerson1991} and the separating hyperplane theorem approach of \cite{fudenberg1991}, the latter recently extended by \cite{cheng2023do} to weaker dominance concepts. However, all these approaches require players to have von Neumann-Morgenstern utility functions---a linearity assumption that is stronger than the SCD$^\star$ condition I employ.

My contribution reconciles these game-theoretic insights within a mechanism design framework. I prove that under SCD$^\star$, bundles that are never strict best responses are precisely those that are very weakly dominated (Proposition \ref{p1}), adapting Pearce's logic to the bundling context without requiring linearity. Moreover, while \cite{bernheim1984}, \cite{pearce1984}, and \cite{fudenberg1991} characterize the existence of rationalizable strategies through best responses, their arguments do not extend to strict best responses. I establish the existence of strict best response bundles under SCD$^\star$ (Corollary \ref{c1}) and provide a counterexample (Example \ref{e6}) showing this existence can fail without the condition. This existence result is crucial for ensuring the minimal optimal menu is non-empty and well-defined.

\textbf{Single-Crossing Differences in Convex Environments:} \cite{kartik2023} introduces MD$^\star$, providing a framework ideally suited to the bundling problem. A convex environment is sufficiently rich that for any two actions and any convex combination of their utility functions, there exists a third action yielding exactly that convex combination as its utility function. In the bundling context, this naturally accommodates preferences over lotteries---that is, stochastic bundles. Importantly, SCD$^\star$ and MD$^\star$ do not presume any exogenous order on the choice space, aligning perfectly with my approach of not imposing a priori structure on the set of bundles. Moreover, \cite{kartik2023} provides elegant characterizations showing when functions exhibit SCD$^\star$ and MD$^\star$. I leverage these characterizations to identify which combinations of valuation functions and type distributions yield virtual value functions with the requisite structure.

\section{Setup}
\label{s332}

There are $n$ different goods $\{1,...,n\}$. A monopolist sells bundles $b \in \mathcal{B}:= 2^{\{1,...,n\}}$ to a consumer whose type $t$ is distributed on $T:=[\underline{t},\bar{t}] \subseteq \mathbb{R}$ according to distribution $F$ with a continuous, positive density $f$. The consumer's value function is $v: \Delta(\mathcal{B}) \times T \rightarrow \mathbb{R}$, in which $v$ is continuously differentiable in $t$. For any stochastic bundle $a \in \Delta(\mathcal{B})$, $v(a,t)= \mathbb{E}_{b \sim a}[v(b,t)]$. The consumer's preferences are over $(a,p)$, in which $p$ is the price paid to the monopolist, and are quasi-linear in $p$, i.e., type $t$ consumer's utility of getting $a$ and paying $p$ is $v(a,t)-p$. 

The monopolist, who has zero production costs, wants to maximize her expected profits over all mechanisms. By the revelation principle, I only need to focus on direct mechanisms, which are given by, in a slight abuse of notation, a measurable mapping $(a,p): T \rightarrow \Delta(\mathcal{B}) \times \mathbb{R}$, where for each $t$, $a(t)$ is a stochastic bundle assigned to type $t$ consumer, and $p(t)$ is a payment from type $t$ consumer to the monopolist. The mechanism is incentive compatible (IC) if \begin{equation} \label {eq1} v(a(t),t)-p(t) \geq v(a(t'),t)-p(t'), \forall t,t' \in T. \tag{IC} \end{equation}
The mechanism is individually rational (IR) if  \begin{equation} v(a(t),t)-p(t) \geq v(\emptyset,t):=0, \forall t \in T. \tag{IR} \end{equation}
The monopolist's problem is \begin{equation} \label{eq3} \begin{split} & \max \limits_{(a,p)} \mathbb{E}[p(t)], \\ \text{s.t. }  & (a,p) \text{ is IC and IR. }\end{split} \end{equation} Solutions to problem (\ref{eq3}) are called (the monopolist's) optimal mechanisms. A \textit{menu} $B \in 2^\mathcal{B}$ is optimal if there exists an optimal mechanism $(a,t)$ such that for any $t, a(t) \in B$. This implies that an optimal menu exists only when the optimal mechanism is deterministic. An optimal menu $B$ is a \textit{minimal optimal menu} if any menu $B' \subsetneq B$ is not optimal. 

The virtual value function is \begin{equation*} \label{eq4} \phi(b,t):= v(b,t)-\dfrac{1-F(t)}{f(t)}v_t(b,t), \end{equation*}
where $v_t(b,t):=\dfrac{\partial v(b,t)}{\partial t}$. For ease of exposition, I assume that no two different bundles have the same virtual value functions, i.e., for any $b \neq b', \text{ there exists } t \in T$ such that $\phi(b,t) \neq \phi(b',t)$. I also assume that $\max \limits_{b \neq \emptyset} \phi(b,\underline{t})<0=\phi(\emptyset, \underline{t})$. This condition ensures that no non-empty bundle has a non-negative virtual value across the entire type space. By the envelope theorem, the monopolist's problem can then be written as
\begin{equation} \label{eq5} \begin{split} \max \limits_{(a,p)} &\mathbb{E}[\sum \limits_{b \in \mathcal{B}} a_b(t) \phi(b,t)]-[v(a(\underline{t}),\underline{t})-p(\underline{t})], 
\\ &\text{s.t. }  (a,p) \text{ is IC and IR,} \end{split} \tag{MP} \end{equation} 
where $a_b(t)$ is the weight that the stochastic bundle $a(t)$ assigns to the deterministic bundle $b$.\footnote{The reason why I can write $\phi(a,t)$ as $\sum \limits_{b \in \mathcal{B}} a_b(t)\phi(b,t)$ is because $v(a,t)= \mathbb{E}_{b \sim a}[v(b,t)]=\sum \limits_{b \in \mathcal{B}} a_b(t)v(b,t)$.}

\section{Solving the Problem}
\label{sec:solving}
\subsection{The Upper Envelope of Virtual Value Functions}
\label{s331}
This subsection aims to demonstrate when taking the upper envelope of virtual value functions, or more standardly, doing pointwise maximization of virtual values, is equivalent to solving the monopolist's problem. I first introduce the following assumption, which requires the virtual value functions to satisfy single-crossing differences and the value functions to satisfy monotonic (either increasing or decreasing) differences:
\begin{assumption}
\label{a0}
For any $b, b' \in \mathcal{B}$: (i) $\phi(b,t)-\phi(b',t)$ is single-crossing in $t$; (ii) $v(b,t)-v(b',t)$ is monotonic in $t$.
\end{assumption}

Taking the upper envelopes of virtual value functions, or figuring out the component bundles of upper envelopes is almost the same as figuring out minimal optimal menus for the following unconstrained problem:
\begin{equation}
\label{eq6}
    \mathbb{E}[\max \limits_{b \in \mathcal{B}} \phi(b,t)],
    \tag{UE}
\end{equation}
where a menu of problem \eqref{eq6} is a set of bundles obtained through picking a bundle from $\argmax \limits_{b \in \mathcal{B}} \phi(b,t)$ for each $t$. 

The following theorem demonstrates the connection between problem \eqref{eq5} and problem \eqref{eq6}. 

\begin{theorem}
    \label{t1}
    Suppose Assumption \ref{a0} holds. The set of minimal optimal menus of problem \eqref{eq5} is the set of minimal optimal menus of problem \eqref{eq6}. 
\end{theorem}

The substance of the theorem is that every minimal optimal menu of problem \eqref{eq6} is a minimal optimal menu of problem \eqref{eq5}. Such a menu attains the value of the relaxed problem, since $\mathbb{E}[\sum_{b \in \mathcal{B}}a_b(t)\phi(b,t)]-[v(a(\underline{t}),\underline{t})-p(\underline{t})] \leq \mathbb{E}[\max_{b \in \mathcal{B}}\phi(b,t)]$ for any $(a,p)$ that is IR; what has to be established is feasibility---a price scheme under which the allocation induced by the menu is incentive compatible and individually rational.

The argument proceeds in three steps. First, part (i) of Assumption \ref{a0} gives the allocation an interval structure: each bundle in a minimal optimal menu of \eqref{eq6} is optimal over a unique interval of types, and the intervals quasi-partition the type space, overlapping only at their boundaries. Second, any two bundles optimal over adjacent intervals swap the ranking of their virtual values between $\underline{t}$ and $\bar{t}$, and for such pairs part (ii) of Assumption \ref{a0} pins down the direction of the valuation difference:

\begin{lemma}
    \label{l1}
    Suppose part (ii) of Assumption \ref{a0} holds. If $b, b' \in \mathcal{B}$ satisfy $\phi(b, \underline{t}) < \phi(b', \underline{t})$ and $\phi(b, \bar{t}) > \phi(b', \bar{t})$, then $v(b,t)-v(b',t)$ is increasing in $t$.
\end{lemma}

Third, with increasing valuation differences across adjacent intervals in hand, the standard price scheme---the one leaving each marginal type indifferent between adjacent bundles---makes the allocation incentive compatible. Individual rationality follows because the lowest type receives the empty bundle at a price of zero, and by IC every higher type obtains at least that outside utility.\footnote{Note that the set of valuation functions on $\mathcal{B}$, $\{v(\cdot, t): t \in T\}$, is not necessarily convex since $v(b, \cdot)$ need not be linear in $t$, so the feasibility of deterministic allocations in \cite{saksandyu2005} cannot be directly applied.}

Part (ii) of Assumption \ref{a0} ranks bundles: for every pair, the direction in which $v(b,t)-v(b',t)$ moves with $t$ orders $b$ against $b'$, and the resulting ranking is complete. The product-line and bundling literatures work with the same kind of order---\cite{mussa1978}, \cite{ghili2023characterization}, and \cite{yang2023nested} likewise impose monotonicity conditions on ordered pairs of products---but there the order is given: qualities on the real line, or a chain under set inclusion, with the conditions imposed along it. Here the conditions apply to all pairs at once, and the ranking they induce need not coincide with set inclusion or any other order available before the analysis. 

Moreover, this global ordering is not necessary for the upper-envelope methodology. Section \ref{s663} demonstrates an alternative version of Theorem \ref{t1}: Assumption \ref{a0} only needs to be imposed on all bundles in some $\mathcal{B}'  \subseteq \mathcal{B}$, and, under an additional condition, the set of minimal optimal menus of \eqref{eq5} still coincides with the set of minimal optimal menus of \eqref{eq6}, with each minimal optimal menu a subset of $\mathcal{B}'$.

\subsection{Characterizing and Computing the Upper Envelope}
\label{s22}

Theorem \ref{t1} reduces the monopolist's problem to finding the upper envelope of virtual value functions but does not specify structure of the upper envelope. This section further deals with it in two steps. First, it provides a theoretical characterization of the upper envelope using the concepts of dominance and best response. Second, it develops an algorithm that operationalizes this characterization, allowing for the systematic computation of the menu. 

I begin with two notions of dominance and best response induced by $\phi$ that I study in this paper. I bring up these two notions because the upper envelope of the virtual value functions consist of strict best response bundles that I define below. While the strict best responses may be hard to characterize directly, it is more tractable to characterize the bundles excluded from the upper envelope—the never strict best responses, or equivalently, once the equivalence below is established, the very weakly dominated bundles.

\begin{definition}
\label{def1}
    A stochastic bundle $a \in \Delta(\mathcal{B})$ is \textit{very weakly dominated} if there exists $a' \in \Delta(\mathcal{B}), a' \neq a$ such that for any $t \in T$, $\phi(a,t) \leq \phi(a',t).$
\end{definition}

\begin{definition}
\label{def2}
    A bundle $b \in \mathcal{B}$ is a \textit{never strict best response} if for any $t \in T$, there exists $b' \in \mathcal{B}, b' \neq b$ such that $\phi(b,t) \leq \phi(b',t)$.\footnote{This is equivalent to say: for any $t \in T$, there exists $a' \in \Delta(\mathcal{B}), a' \neq b$ such that $\phi(b,t) \leq \phi(a',t)$.} Correspondingly, a bundle $b \in \mathcal{B}$ is a \textit{strict best response} if there exists $t' \in T$ such that $\phi(b,t')>\phi(\tilde{b},t')$ for any $\tilde{b} \neq b$.
\end{definition}

These two definitions resemble the notions of dominated strategies and best response in game theory. Note that $b'$ in Definition \ref{def2} could be type-dependent, as in the standard ``never best response'' definitions.  

The following assumption helps establish the equivalence between very weakly dominated deterministic bundles and never strict best response bundles. Moreover, this assumption also ensures the uniqueness of the minimal optimal menu and the fact that it is indeed the set of strict best responses. The assumption requires $\phi$ to have a certain single-crossing difference property. 

\begin{assumption}
\label{a1}
    For any $a,a' \in \Delta(\mathcal{B}), \phi(a,t)-\phi(a',t)$ is single-crossing in $t$.\footnote{A function $\tilde{f}: T \rightarrow \mathbb{R}$ is single-crossing in $t$ if $\sign [\tilde{f}(t)]$ is monotonic in $t$. $\sign: \mathbb{R} \rightarrow \{-1,0,1\}$ is defined as: $ \sign x =  -1 \text{ if } x<0; =0 \text{ if } x=0; =1 \text{ if }x>0 $.} 
\end{assumption}

Note that Assumption \ref{a1} is imposed on any pairs of \emph{stochastic} bundles, whereas Assumption \ref{a0} is only imposed on any pairs of \emph{deterministic} bundles. The virtual value function $\phi$ satisfying Assumption \ref{a1} is the same as $\phi$ having SCD$^{\star}$, or \emph{single-crossing differences in convex environments}, defined in \cite{kartik2023}.\footnote{This is because the set of functions $\{\phi(a,\cdot): T \rightarrow \mathbb{R}\}_{a \in \Delta(\mathcal{B})}$ is convex: for any $a,a' \in \Delta(\mathcal{B}), t \in T, \lambda \in [0,1], \lambda \phi(a,t)+(1-\lambda) \phi(a',t)= \phi(\lambda a+(1-\lambda) a',t)$, which implies $\lambda \phi(a,\cdot)+(1-\lambda)\phi(a',\cdot)=\phi(\lambda a+(1-\lambda) a',\cdot) \in \{\phi(a,\cdot): T \rightarrow \mathbb{R}\}_{a \in \Delta(\mathcal{B})}$.}
\begin{proposition}
\label{p1}
    Suppose Assumption \ref{a1} holds. $b \in \mathcal{B}$ is a never strict best response if and only if $b$ is very weakly dominated.
\end{proposition}

The basic idea of the argument is familiar from the duality of never-best responses and strict dominance (\cite{pearce1984}), but there are two crucial differences. First, $\phi$ is not necessarily linear in $t$. Second, under strict dominance ``$b$ is strictly dominated by $a$'' implies $b \neq a$, whereas under very weak dominance this implication is lost. As demonstrated in the proof, Assumption \ref{a1} plays a role in showing that there does exist $a \neq b$ such that $\phi(a,t) \geq \phi(b,t), \forall t$.

\begin{remark}
Example \ref{e2} in Online Appendix \ref{s622} will show that, without Assumption \ref{a1}, a never strict best response bundle may not be very weakly dominated, which is exactly the case when there exists a never strict best response $b$ such that $\phi(a,t) \geq \phi(b,t), \forall t$ only holds for $a=b$. 
\end{remark}


The existence of strict best response bundles can be regarded as a corollary of Proposition \ref{p1}, and will be introduced in Online Appendix \ref{sec:existence}. 

With Proposition \ref{p1} in hand, I am ready to introduce the characterization of the minimal optimal menus for problem \eqref{eq6}.

\begin{proposition}
\label{p4}
    Suppose Assumption \ref{a1} holds. $\{b_1,...,b_m\} \in 2^{\mathcal{B}}$ is a minimal optimal menu for problem \eqref{eq6} if and only if the following hold:
    \begin{enumerate}
        \item Any very weakly dominated bundle, in particular, any $b \in \mathcal{B} \backslash \{b_1,...,b_m\}$ is very weakly dominated by a convex combination of $b_1,...,b_m$.
        \item For any $i \in \{1,...,m\}, b_i$ is not very weakly dominated. \label{cp2}
    \end{enumerate}
    In addition, the minimal optimal menu for problem \eqref{eq6} is unique. 
\end{proposition}

The two conditions say that a minimal optimal menu has the flavor of a von Neumann--Morgenstern stable set (\cite{vonneumann1944}) of the very weakly dominance relation---externally stable in that every excluded bundle is dominated by a convex combination of menu bundles, internally stable in that no menu bundle is dominated. It is worth mentioning that if $\{b_1,...b_m\}$ does not satisfy condition \ref{cp2} above, then $\{b_1,...,b_m\}$ is an optimal but not necessarily minimal optimal menu for problem \eqref{eq6}. This condition, along with Proposition \ref{p1}, ensures that getting rid of everything very weakly dominated is exactly getting rid of everything that is not a strict best response. 

Proposition \ref{p4} says that a set of bundles is a minimal optimal menu if and only if it is exactly the set of not very weakly dominated bundles.\footnote{This result resembles Theorem 3 of \cite{cheng2023do}, which basically says that the set of strict best response is the minimal optimal menu when never strict best response is equivalent to very weakly dominance.} The uniqueness result directly follows from this observation, as the set of not very weakly dominated bundles is unique. 

\begin{remark}
\label{rmk:nsbr}
Characterizing minimal optimal menus requires \textit{never strict best response}, not merely never best response: a bundle that is optimal at some type but never \textit{uniquely} optimal is redundant, yet ``never best response'' fails to exclude it. For instance, if $b''$ ties $b$ and $b'$ at a single type but is never the sole maximizer, then it could be that $\{b,b'\}$ and $\{b,b',b''\}$ are both optimal menus, so $b''$ should be dropped when discussing minimal optimality---which only the strict notion does.\footnote{The strict notion raises one caveat: if $\phi(b,t)=\phi(b',t)$ for all $t$, then $b$ and $b'$ very weakly dominate each other, yet there could be cases where neither should be removed from the upper envelope. The assumption in Section \ref{s332} that distinct bundles induce distinct virtual value functions rules this out.} 
\end{remark}

\begin{remark}
    Examples \ref{e5} and \ref{e6} in Online Appendix \ref{s622} will show that both the properties of minimal optimal menus specified in Proposition \ref{p4} and the uniqueness of the minimal optimal menu cannot hold if Assumption \ref{a1} is not satisfied. 
\end{remark}

In addition to building up the equivalence between never strict best response and very weakly dominance, Assumption \ref{a1} also plays a role in offering a simple condition that can be used to tell when bundle $b$ is very weakly dominated by (stochastic) bundle $a$: if $b$'s virtual value is no larger than $a$'s virtual value at $\underline{t}$ and $\bar{t}$, then $b$'s virtual value is no larger than $a$'s virtual value for any $t$, as $\sign[\phi(a,t)-\phi(b,t)]$ is monotonic in $t$. Formally,

\begin{corollary}
\label{c2}
    Suppose Assumption \ref{a1} holds. $b \in \mathcal{B}$ is a never strict best response if and only if there exists $a \in \Delta(\mathcal{B})$ such that $\phi(a,\bar{t}) \geq \phi(b,\bar{t})$ and $\phi(a,\underline{t}) \geq \phi(b,\underline{t})$. 
\end{corollary}

Proposition \ref{p4} and Corollary \ref{c2} offer the potential to come up with an algorithm that can generate the minimal optimal menu, in particular, an algorithm that can (iteratively) getting rid of all very weakly dominated bundles. To have a good grasp of the algorithm, I first introduce the following lemma: 

\begin{lemma} \label{l2} 
Suppose Assumption \ref{a1} holds. For any $b, b', b'' \in \mathcal{B}$ with
\begin{equation*}  
\phi(b,\underline{t}) < \phi(b',\underline{t}) < \phi(b'',\underline{t}) \text{ and } \phi(b,\bar{t}) > \phi(b',\bar{t}) > \phi(b'',\bar{t}),
\end{equation*}
$b'$ is very weakly dominated by a convex combination of $b$ and $b''$ if and only if 
\begin{equation} \label{eq66}
\dfrac{\phi(b,\bar{t}) - \phi(b',\bar{t})}{\phi(b',\bar{t}) - \phi(b'',\bar{t})} \geq \dfrac{\phi(b,\underline{t}) - \phi(b',\underline{t})}{\phi(b',\underline{t}) - \phi(b'',\underline{t})}.
\end{equation}
\end{lemma}

\paragraph{Interpretation via marginal revenue curves.} Define $MR(b,q)$ as the standard marginal revenue curve for bundle $b$ at quantitiy $q \in [0,1]$. Since $\phi(b, t) = MR(b, 1-F(t))$ for any $b$ and $t$, the hypothesis of Lemma \ref{l2} can be restated as
\begin{equation*}
MR(b, 1) < MR(b', 1) < MR(b'', 1) \text{ and } MR(b, 0) > MR(b', 0) > MR(b'', 0),
\end{equation*}
and the inequality \eqref{eq66} becomes
\begin{equation*}
\dfrac{MR(b, 0) - MR(b', 0)}{MR(b', 0) - MR(b'', 0)} \geq \dfrac{MR(b, 1) - MR(b', 1)}{MR(b', 1) - MR(b'', 1)}.
\end{equation*}
Ranked by marginal revenue, $b''$ is most attractive at $q = 0$ (the top of the demand curve, low quantities) while $b$ is most attractive at $q = 1$ (the bottom, high quantities). Bundle $b'$ lies strictly between them at both endpoints.

The lemma asks when this intermediate position suffices to make $b'$ part of the optimal menu. The answer is: not always. Inequality \eqref{eq66} is a \emph{chord test}---it requires that the chord connecting $\phi(b'', \cdot)$ and $\phi(b, \cdot)$ at the two endpoints passes weakly above $\phi(b', \cdot)$. When this holds, a randomization between $b$ and $b''$ replicates at least as much marginal revenue as $b'$ at both endpoints, so $b'$ is redundant. 

I proceed to demonstrate the algorithm for generating the minimal optimal menu for problem \eqref{eq6}: 

\begin{algorithm}[H]

\KwIn{A set of bundles $\mathcal{B}$; the value of $\phi$ at $\bar{t}$ and $\underline{t}$: $\phi(\cdot, \bar{t}); \phi(\cdot, \underline{t})$.\footnotemark} 
\KwOut{The minimal optimal menu for problem \eqref{eq6}.}

\textbf{Step 1:} Remove from $\mathcal{B}$ all $b$ for which there exists 
$b' \neq b$ such that 
$\phi(b',\bar{t}) \ge \phi(b, \bar{t})$ and $\phi(b', \underline{t}) \ge \phi(b, \underline{t})$.\\

\Repeat{\text{no bundles are removed during Step 2}}{
  \textbf{Step 2:} Denote the remaining bundles as 
  $\{b'_1, b'_2, \dots, b'_l\}$ such that 
  $\phi(b'_1, \underline{t}) > \phi(b'_2, \underline{t}) > \dots > \phi(b'_l, \underline{t}); \phi(b'_1, \bar{t}) < \phi(b'_2, \bar{t}) < \dots < \phi(b'_l, \bar{t})$. 
  
  \For{$i \leftarrow 2$ \KwTo $l-1$}{
    \If{\(
      \frac{\phi(b'_{i+1}, \bar{t}) - \phi(b'_i, \bar{t})}%
           {\phi(b'_i, \bar{t}) - \phi(b'_{i-1}, \bar{t})}
      \;\;\ge\;\;
      \frac{\phi(b'_{i+1}, \underline{t}) - \phi(b'_i, \underline{t})}%
           {\phi(b'_i, \underline{t}) - \phi(b'_{i-1}, \underline{t})}
      \)}
    {
      Remove $b'_i$ from $\mathcal{B}$\;
    }
  }
}

\caption{Generating the Minimal Optimal Menu}
\label{al1}
\end{algorithm}
\footnotetext{The input does not need to be the set of functions $\{\phi(b,\cdot)\}_{b \in \mathcal{B}}$.}
\normalsize
Note that a direct implementation of Algorithm \ref{al1} runs in time polynomial in $|\mathcal{B}|$. When there are $n$ goods, $|\mathcal{B}|=2^n$, which is exponential in $n$. This complexity is typical in multi-good mechanism design: enumerating all bundles is often unavoidable unless further structural assumptions reduce the search space, so an exponential dependence on $n$ is not surprising.

Lemma \ref{l2} is key in Step 2 in the algorithm, which can be regarded as the deletion of very weakly dominated bundles that are not very weakly dominated by pure bundles but are very weakly dominated by stochastic bundles. The output of the algorithm is characterized by the following theorem:

\begin{theorem}
\label{t2}
    Suppose Assumptions \ref{a0} and \ref{a1} hold. The output of Algorithm \ref{al1} is the minimal optimal menu for the monopolist's problem (problem \eqref{eq5}). 
\end{theorem}

With Theorem \ref{t1}, it suffices to show that, under Assumption \ref{a1}, the output of Algorithm \ref{al1} is the minimal optimal menu for problem \eqref{eq6}. The key step of the proof is to show that, given a set of bundles $\{b_1,...,b_m\}$, if $b_i$ is not very weakly dominated by any convex combination of any pairs of bundles $b_j$ and $b_k$, then $b_i$ is not very weakly dominated by any convex combination of $b_1,...,b_{i-1},b_{i+1},...,b_m$. This observation does not hold generally, and depends on Assumption \ref{a1}. 

\subsection{Discussions of Assumptions}
\label{sec:discussionofassump}
In this subsection, I discuss what $v$ and $F$, and more broadly what environments, satisfy Assumptions \ref{a0} and \ref{a1}. 
Readers primarily interested in applications may proceed directly to the tree bundling application (Section \ref{tree}), where Assumptions \ref{a0} and \ref{a1} are maintained throughout.

\label{s23}

Assumptions \ref{a0} and \ref{a1} provide analytical structure directly on the virtual value function $\phi$. To understand the economic foundation of this structure, I proceed to explore the primitive conditions on the consumer's valuation $v$ and the type distribution $F$ that ensure these assumptions hold. I first propose a stronger assumption than the combination of Assumptions \ref{a0} and \ref{a1}:

\begin{assumption}
    \label{a2}
    For any $a,a' \in \Delta(\mathcal{B}), \phi(a,t)-\phi(a',t)$ is monotonic in $t$. 
\end{assumption}

The virtual value function $\phi$ satisfying Assumption \ref{a2} is $\phi$ having MD$^{\star}$, or \emph{monotonic differences in convex environments}, defined in \cite{kartik2023}. Again, Assumption \ref{a2} imposes the condition on any pair of stochastic bundles, and it is clear that if $\phi$ satisfies Assumption \ref{a2}, then $\phi$ satisfies Assumption \ref{a1}, but the converse does not hold. Nevertheless, since both Assumptions \ref{a1} and \ref{a2} are conditions on any pair of stochastic bundles, there is indeed a condition that can bridge the gap between the two. Detailed discussions can be found in Proposition \ref{p133} in Online Appendix \ref{sec:a1a2}.

The relationship between Assumptions \ref{a2} and \ref{a0}, in particular, part (ii) of Assumption \ref{a0} follows from the observation below:

\begin{proposition}
    \label{prop:mdphiv}
    If for any $b,b' \in \mathcal{B}, \phi(b,t)-\phi(b',t)$ is monotonic in $t$, then for any $b,b' \in \mathcal{B}, v(b,t)-v(b',t)$ is monotonic in $t$. 
\end{proposition}

Hence, Assumption \ref{a2} is indeed stronger than the combination of Assumptions \ref{a0} and \ref{a1}, and any $(v,F)$ pair that satisfies Assumption \ref{a2} also satisfies Assumptions \ref{a0} and \ref{a1}. 

As a benchmark, the canonical settings of \citet{myerson1981optimal} and \cite{mussa1978} satisfy Assumption \ref{a2}---and hence Assumptions \ref{a0} and \ref{a1}---whenever $F$ is regular. In the bundling notation, the valuation function corresponds to $v(b,t) = g(b)\,t$ for some $g: \Delta(\mathcal{B}) \to \mathbb{R}$, so that for any $a, a' \in \Delta(\mathcal{B})$,
\[
\phi(a,t) - \phi(a',t) = [g(a) - g(a')]\left[t - \frac{1-F(t)}{f(t)}\right],
\]
which is monotonic in $t$ if the function $t - \frac{1-F(t)}{f(t)}$ is increasing, i.e., when $F$ is regular.

Going further, by \cite{kartik2023}, $\phi$ satisfying Assumption \ref{a2} is equivalent to the following: for any $a \in \Delta(\mathcal{B}),$ there exist $g_1,g_2: \Delta(\mathcal{B}) \rightarrow \mathbb{R}$ and $h_1,h_2: T \rightarrow \mathbb{R}$ such that:
\begin{equation} \label{eq9} \phi(a,t)=g_1(a)h_1(t)+g_2(a)+h_2(t), \end{equation}
where $h_1$ is monotonic. Consider the right hand side of (\ref{eq9}) in terms of deterministic bundles and plug it into the left hand side of (\ref{eq4}):
\begin{equation}  \label{eq12} v(b,t)-\dfrac{1-F(t)}{f(t)}v_t(b,t)=g_1(b)h_1(t)+g_2(b)+h_2(t).  \end{equation}
Given $b$ and $F$, \eqref{eq12} can be regarded as a linear ODE of order one, for which the solution is:
\begin{equation*}
    \begin{split}
        v(b,t) &=\dfrac{\displaystyle{\int} f(t)[g_1(b)h_1(t)+g_2(b)+h_2(t)]dt \bigg|_{t=\bar{t}}  -\displaystyle{\int} f(t)[g_1(b)h_1(t)+g_2(b)+h_2(t)]dt}{1-F(t)} , t\in [\underline{t},t),  \\  & v(b, \bar{t}) =g_1(b)h_1(\bar{t})+g_2(b)+h_2(\bar{t}),
    \end{split}
\end{equation*}
for any $b \in \mathcal{B}$, in which the boundary condition is pinned down by the continuity of $v$ at $\bar{t}$. From the above expression, if for any $i \in \{1,2\}, a \in \Delta(\mathcal{B}), g_i(a)=\mathbb{E}_{b \sim a} [g_i(b)]$, i.e., the expected utility property, then $v(a,t)=\mathbb{E}_{b \sim a} [v(b,t)]$ is satisfied for all $a \in \Delta(\mathcal{B})$ and the solution form persists. Hence, it suffices to focus on the valuation functions for deterministic bundles. To further simplify, for any $b,t$, \begin{equation} \label{eq:multiadd} v(b,t)=g_1(b)x(t)+g_2(b)+y(t), \tag{quality-index} \end{equation}
where $x,y: T \rightarrow \mathbb{R}$.\footnote{A foundation of this simplification is the quantile transformation \`{a} la \cite{toikka2011}. See Online Appendix \ref{sec:quantile} for details.}  The name reflects the reading the form imposes: each bundle is summarized by a one-dimensional quality index $g_1(b)$ and a type-independent baseline $g_2(b)$, and each type by a marginal value $x(t)$ for that index---so the entire interaction between bundles and types runs through a quality dimension.\footnote{``Quality'' is a ``memory aid'': neither monotonicity of $x$ nor any relation between $g_1$ and set inclusion is imposed. The specification of \cite{HopkinsKornienko2004} below, where $x$ depends on rank, illustrates the reach beyond vertical differentiation.}\textit{Under this specification}, the following observation holds:
\begin{proposition}
\label{prop:equivalence}
    Suppose $v$ satisfies \eqref{eq:multiadd}. If for any $i \in \{1,2\}, g_i$ satisfies the expected utility property, then:
 \[ \hat{x}(t):=x(t)-\frac{1-F(t)}{f(t)}x'(t) \text{ is monotonic in } t \Rightarrow \text{ Assumption \ref{a2} holds }. \]

\end{proposition}

The sufficient condition has a clean economic reading: $x(t)$ is the type's marginal value for the single ``quality index'' $g_1$, and $\hat{x}(t)$ is its virtual counterpart; the entire apparatus of Section \ref{sec:solving}---the upper-envelope characterization (Theorem \ref{t1}), the equivalence between never-strict-best-response and very weak dominance (Proposition \ref{p1}), and the iterative elimination algorithm (Algorithm \ref{al1})---is guaranteed by a single condition: Myersonian regularity of the virtual taste $\hat{x}$. 

After reducing to the monotonicity of $\hat{x}$, apart from the aforementioned benchmark settings, two economically natural pairings of $(x,F)$ illustrate the reach of the condition.

The first is a concave generalization of \cite{mussa1978}: the consumer's taste $x$ for the bundle characteristic $g_1$ increases in her type but at a diminishing rate, so higher types value quality more, yet the taste gap between adjacent types narrows toward the top---a natural saturation pattern for most quality dimensions. When tastes exhibit this shape ($x' \geq 0$, $x'' \leq 0$) and $F$ is regular,
\[ \frac{d\hat{x}(t)}{dt}=\underbrace{x'(t)\,\frac{d}{dt}\Big[t-\frac{1-F(t)}{f(t)}\Big]}_{\geq 0 \text{ by regularity}}\;\underbrace{-\,\frac{1-F(t)}{f(t)}\,x''(t)}_{\geq 0 \text{ by concavity}} \geq 0, \]
so $\hat{x}$ is monotone: the information-rent correction, which scales with the slope of tastes, shrinks exactly where tastes flatten, and never overturns the underlying order. The mirror case ($x' \leq 0$, $x'' \geq 0$) describes heterogeneous sensitivity to a ``bad'' characteristic of the bundles---say, complexity or maintenance burden---with higher types less bothered by it.

The second pairing makes willingness to pay depend on the consumer's \emph{relative position} rather than the cardinal distance between types. Following \cite{HopkinsKornienko2004}, let the taste be determined by percentile rank, $x(t)=\tilde{q}(F(t))$, with $\tilde{q}$ increasing and concave, so that improvements in rank raise the taste for $g_1$ at a diminishing rate. This specification is natural when $t$ represents income or socioeconomic status and $g_1(b)$ measures the visibility or status intensity of bundle $b$.\footnote{This is related to conspicuous consumption and wealth signaling, e.g., \cite{bagwell1996veblen}.} Here,
\[ \frac{d\hat{x}(t)}{dt}=f(t)\Big[2\tilde{q}'(F(t))-(1-F(t))\,\tilde{q}''(F(t))\Big] \geq 0 \quad \text{for \emph{any} } F, \]
so monotonicity holds without regularity: because tastes are pinned to ranks, redistributing types across the line rescales the picture without reordering it, and only the shape of $\tilde{q}$ matters.

\section{Tree Bundling}
\label{tree}

This section presents the main application of this paper: \textit{tree bundling}. In particular, I introduce ``\textit{tree menu}'' and the conditions for the optimality of tree bundling. I start with the formal definitions of nested menus and tree menus:

\begin{definition}
    \label{d35} $\{b_1,...,b_m\} \in 2^\mathcal{B}$ is a nested menu if there exists a one-to-one mapping $\rho: \{1,...,m\} \rightarrow \{1,...,m\}$ such that $b_{\rho(1)} \subseteq ... \subseteq b_{\rho(m)}$.\footnote{In words, a menu is nested if any two bundles inside the menu are ordered by set inclusion.} 
\end{definition}

\begin{definition}
     \label{d33}
    $\{b_1,...,b_m\} \in 2^\mathcal{B}$ is a tree menu if:
    \begin{enumerate}
        \item $\text{ there exists }i \in \{1,...,m\}$ such that $b_i \neq \emptyset, \text{ for any } b_j \neq \emptyset, b_i \subseteq b_j$, and 
        \item $\text{ there exists } k,l \in \{1,...,m\}$ such that $b_k \nsubseteq b_l$ and $b_l \nsubseteq b_k$. 
    \end{enumerate}
\end{definition}

The first condition of Definition \ref{d33} characterizes the role of the ``\textit{root bundle}'': every bundle in the menu has to build on it. Note that in nested menus (Definition \ref{d35}), there is also a ``root bundle''--- the existence is guaranteed by the finiteness of the menu. In the language of product upgrading, the root bundle is the simplest, least upgraded one and the foundation for all upgrades. The uniqueness of the root bundle is straightforward: If there exists $b_i \neq b_i'$ such that both $b_i$ and $b_i'$ are root bundles, then $b_i \subseteq b_i', b_i' \subseteq b_i \Rightarrow b_i=b_i'$, a contradiction. 

The second condition of Definition \ref{d33} is what distinguishes a tree menu from a nested menu: It demonstrates the existence of different upgrade paths in a tree menu, whereas there is only one but not many upgrade paths in a nested menu. 

Here are some simple examples of tree and not-tree menus when selling three goods $\{1,2,3\}$: $\{\{1\}, \{1,2\}, \{1,3\}\}$ is a tree menu: $\{1\}$ is the root bundle, and there are two different upgrade paths: $\{1\} \rightarrow \{1,2\}$ and $\{1\} \rightarrow \{1,3\}$. In contrast, $\{\{1\},\{2\}, \{1,2\}\}$ is not a tree menu, as there does not exist a root bundle; $\{\{1\}, \{1,2\}, \{1,2,3\}\}$ is a nested but not a tree menu, as there is only one upgrade path. 

\subsection{Optimality of Tree Bundling}

This subsection demonstrates the optimality of tree bundling. I start with the definition of a bundle's \textit{sold-alone quantity}. 

\begin{definition}
\label{d4}
    The sold-alone quantity $Q(b) \in [0,1]$ of bundle $b \in \mathcal{B}$ is $1-F(t_b)$, in which $t_b \in T$ satisfies:
    \begin{equation*}
        \phi(b,t_b)=MR(b, 1-F(t_b))=0.
    \end{equation*}
\end{definition}

In words, bundle $b$'s sold-alone quantity is exactly the quantity sold by a monopolist (that maximizes her revenue) when she is only selling bundle $b$. The existence of the sold-alone quantity for any non-empty bundle $b$ is guaranteed by the assumption that $\max \limits_{b \neq \emptyset} \phi(b, \underline{t})<0$ and $\phi(b, \bar{t})=v(b, \bar{t})>0, \forall b \neq \emptyset$, and the continuity of $\phi$; meanwhile, if the single-crossing difference in Assumption \ref{a0} is strengthen to the strict version, then the uniqueness is guaranteed: for any $b, \phi(b,t)-\phi(\emptyset,t)=\phi(b,t)$ is strictly single-crossing in $t$.

As a first step of showing the optimality of tree bundling, I propose conditions such that any minimal optimal menu must have either a nested or a tree structure.

\begin{proposition}
\label{p33}
    Suppose Assumptions \ref{a0} and \ref{a1} hold. The minimal optimal menu is either a tree menu or a nested menu if there exists $b^{root} \in \mathcal{B}$ such that:
    \begin{enumerate}
        \item $\{b^{root}\}=\argmax \limits_{\tilde{b}} Q(\tilde{b})$, and  \label{cd1}
        \item $\text{ for any } b$ such that $b \not \supset b^{root}, v(b,\bar{t})<v(b^{root},\bar{t})$. \label{cd2}
    \end{enumerate}
\end{proposition}

Corresponding to the notation, $b^{root}$ is the root of the tree menu, or plays the role of ``$b_i$'' in Definition \ref{d33}. The first condition of Proposition \ref{p33} requires the root bundle to be the one with the largest sold-alone quantity. This reflects the real-life observation that the base of any product upgrade line is usually the most widely accepted product that can be sold to the largest set of types of consumers when selling alone. The second condition of Proposition \ref{p33} says that if a bundle does not contain the root bundle, then the valuation of it to the type $\bar{t}$ consumer is less than the valuation of the root bundle. This condition is natural in the sense that even though different consumer prefers different upgrades, all of them agree that the upgrades have to be built on the root. Any upgrade will lose its value without the support of the base, like a car with high-end leather and audio but without wheels cannot be sold. 

Taken together, these conditions ensure that any bundle that does not contain $b^{root}$ is very weakly dominated by $b^{root}$. Meanwhile, the root bundle is not very weakly dominated because it has the largest sold-alone quantity. Hence, the minimal optimal menu consists of the root bundle and bundles that contain it, which makes it either a tree or a nested menu.

An alternative approach can also ensure that bundles not containing the root bundle are very weakly dominated. This approach relies on two conditions: i) The assumption that the top-type consumer has a higher valuation for larger bundles, and ii) A ``partial union quantity condition'', which applies only to the root bundle and bundles that do not contain it.\footnote{See \cite{yang2023nested} or (\ref{eq114}) below for the definition of ``union quantity condition''.}

The following proposition formalizes this alternative characterization:

\begin{proposition}
\label{p34}
    Suppose Assumptions \ref{a0}, \ref{a1}, and the assumption that for any $b \subsetneq b', v(b, \bar{t})<v(b', \bar{t})$ hold. The minimal optimal menu is either a tree menu or a nested menu if there exists $b^{root} \in \mathcal{B}$ such that:
        \begin{enumerate}
        \item $\{b^{root}\}=\argmax \limits_{\tilde{b}} Q(\tilde{b})$, and 
        \item $\text{ for any } b$ such that $b \not \supset b^{root}, Q(b \cup b^{root}) \geq \min \{Q(b), Q(b^{root})\}$. \label{cd52}
    \end{enumerate}
   
\end{proposition}

This characterization is similar to the one in \cite{yang2023nested}, but the conditions do not give the optimality of nested bundling just because condition \ref{cd52} only holds for the root bundle and bundles that do not include the root, but not for \emph{any} two bundles. 

I proceed to discuss when the minimal optimal menu must be a tree menu but not a nested menu. I first demonstrate a result that makes \emph{every} bundle ``between'' the root bundle and the grand bundle, i.e., every bundle that contains the root bundle (and is contained by the grand bundle, which is trivially true) is included in the minimal optimal menu. Similar to Proposition \ref{p34}, I also need an additional assumption that depicts the monotonicity of type $\bar{t}$ consumer's valuation among certain nested bundles.

\begin{theorem}
\label{t33}
    Suppose conditions \ref{cd1} and \ref{cd2} of Proposition \ref{p33}, Assumptions \ref{a0} and \ref{a1}, and the assumption that for any $b_1, b_2$ with $b^{root} \subseteq b_1 \subsetneq b_2 \subseteq b^*$, $v(b_1,\bar{t})<v(b_2,\bar{t})$ hold, where $b^*$ is the grand bundle. The minimal optimal menu is a tree menu if:
    \begin{enumerate}
        \item $|\{\tilde{b}: b^{root} \subsetneq \tilde{b} \subsetneq b^*\}| \geq 2$, and
        \item $ \text{ there exists } \mathfrak{f}: [0,1] \rightarrow [0,1]$ that is increasing and strictly convex such that \begin{equation*} \text{ for any } \tilde{b} \text{ such that } b^{root} \subsetneq \tilde{b} \subsetneq b^*, \dfrac{\phi(\tilde{b},\underline{t})-\phi(b^{root},\underline{t})}{\phi(b^*,\underline{t})-\phi(b^{root},\underline{t})}=\mathfrak{f}(\dfrac{\phi(\tilde{b},\bar{t})-\phi(b^{root},\bar{t})}{\phi(b^*,\bar{t})-\phi(b^{root},\bar{t})}). \end{equation*} 
    \end{enumerate}

\end{theorem}

The two conditions have clean readings. The first says $b^*$ contains at least two goods beyond $b^{root}$, so at least two upgrade paths exist. The second says that, across all intermediate bundles, the normalized virtual values at $\underline{t}$ and $\bar{t}$ are linked by a single increasing, strictly convex map $\mathfrak{f}$.

These two properties of $\mathfrak{f}$ are exactly what keep every intermediate bundle on the menu. \emph{Increasing} rules out one intermediate bundle pointwise dominating another; \emph{strict convexity} rules out dominance by lotteries---via Lemma \ref{l2}, it ensures no intermediate bundle fails the chord test against its neighbors, so Algorithm \ref{al1} never removes one. Both properties are tight: if $\mathfrak{f}$ is linear ($\mathfrak{f}(x)=x$), every intermediate bundle is very weakly dominated by a convex combination of $b^{root}$ and $b^*$, collapsing the menu to $\{b^{root}, b^*\}$---not even a tree. 

See below a concrete example of Theorem \ref{t33}:
\begin{example}
\label{ex:tree}
A monopolist sells a computer (good ``$1$''), a lighting kit (good ``$2$''), and a performance upgrade (good ``$3$''). $t$ is uniformly distributed on $[0,1]$ and the valuation takes the \eqref{eq:multiadd} form of Section \ref{sec:discussionofassump}: $v(b,t)=g_1(b)\,x(t)+g_2(b)$ with $x(t)=-t^2+2t$, so that $\phi(b,t)=g_1(b)\,\hat{x}(t)+g_2(b)$ with $\hat{x}(t)=-3t^2+6t-2$ increasing on $[0,1]$, and Assumptions \ref{a0} and \ref{a1} hold by Proposition \ref{prop:equivalence}. The $g_1(b)$ and $g_2(b)$ for $b \in 2^{\{1,2,3\}}$ are
\begin{center}
\begin{tabular}{lcccccccc}
$b$ & $\emptyset$ & $\{1\}$ & $\{2\}$ & $\{3\}$ & $\{2,3\}$ & $\{1,2\}$ & $\{1,3\}$ & $\{1,2,3\}$ \\
\hline
$g_1(b)$ & $0$ & $1$ & $0.8$ & $0.9$ & $1.5$ & $2$ & $3$ & $5$ \\
$g_2(b)$ & $0$ & $1$ & $0.4$ & $0.5$ & $0.4$ & $1.5$ & $1.8$ & $2$ \\
\end{tabular}
\end{center}
Accessories are worth little without the computer, and the computer alone has the strictly largest sold-alone quantity, so the conditions of Proposition \ref{p33} hold with $b^{root}=\{1\}$; the conditions of Theorem \ref{t33} hold as well. Running Algorithm \ref{al1}, Step 1 removes the three bundles without the computer, and no bundle is removed in Step 2. The minimal optimal menu is
\[
\{\emptyset,\ \{1\},\ \{1,2\},\ \{1,3\},\ \{1,2,3\}\},
\]
a tree with root $\{1\}$ and two upgrade paths: add the lighting kit, or add the performance upgrade. Figure \ref{fig:tree-envelope} plots the virtual value functions; each menu bundle occupies its own segment of the upper envelope. Notably, the two branches serve adjacent segments of the middle market even though neither bundle contains the other: the allocation is ordered by virtual values, not by set inclusion. 
\end{example}

\begin{figure}[H]
    \centering
    \includegraphics[width=1\linewidth]{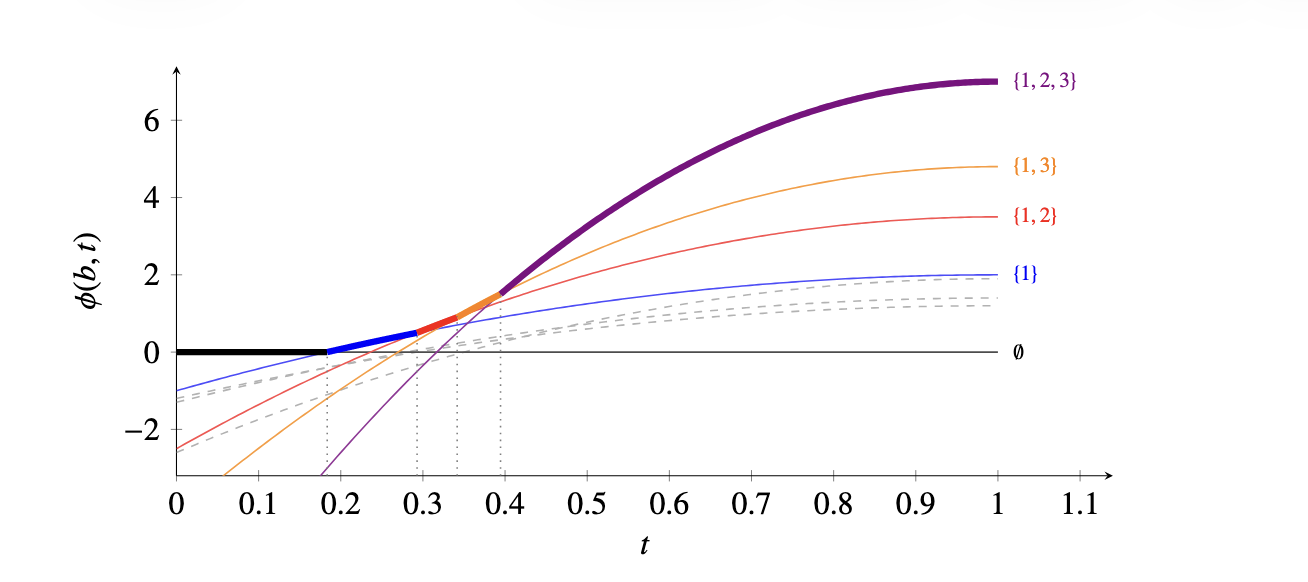}
    \caption{Virtual value functions in Example \ref{ex:tree}. The upper envelope is drawn thick, colored by the bundle attaining it on each segment; dotted verticals mark the marginal types. Gray dashed curves are bundles removed in Step 1 of Algorithm \ref{al1}.}
    \label{fig:tree-envelope}
\end{figure}

Theorem \ref{t33} delivers a tree by placing every intermediate bundle on the menu, which requires all of them to line up on a single convex curve. A tree, however, needs much less: it suffices that the menu contain two set-incomparable bundles. The next result exploits this observation. It identifies one good whose presence or absence generates the two branches, and imposes conditions only on the two bundles that differ from the corners by exactly that good.

Say that good $i$ is a \emph{least-favorite good} if
\[
\{i\}=\argmax_j v(b^* \backslash \{j\}, \bar{t}) \quad \text{and} \quad \{i\}= \argmin_{j,\, b^{root} \cup \{j\} \neq b^{root}} v(b^{root} \cup \{j\}, \bar{t}),
\]
that is, removing $i$ from the grand bundle causes the smallest drop in the top type's value, and adding $i$ to the root causes the smallest gain---like keyboard lighting on a computer or a decorative decal on a car.

\begin{theorem}
\label{t34}
Suppose conditions \ref{cd1} and \ref{cd2} of Proposition \ref{p33}, Assumptions \ref{a0} and \ref{a1}, and the assumption that for any $b_1, b_2$ with $b^{root} \subseteq b_1 \subsetneq b_2 \subseteq b^*$, $v(b_1,\bar{t})<v(b_2,\bar{t})$ hold. If there exists a least-favorite good $i$ and the branching condition (formally introduced in Appendix \ref{app:t34}) holds, then the minimal optimal menu is a tree menu. In particular, it contains $\{b^{root}, b^{root} \cup \{i\}, b^* \backslash \{i\}, b^*\}$. 
\end{theorem}

The branching condition concerns four bundles: the root $b^{root}$, the grand bundle $b^*$, and the two bundles that differ from them by exactly the least-favorite good---$b^{root} \cup \{i\}$, the cheapest upgrade from the root, and $b^* \setminus \{i\}$, the largest bundle short of everything. The condition ensures that all four survive the elimination of Algorithm \ref{al1}. Its first part applies the chord test of Lemma \ref{l2} among the four themselves: neither $b^{root} \cup \{i\}$ nor $b^* \setminus \{i\}$ is very weakly dominated by a lottery over the other three. Its second part restricts the endpoint virtual values of every other intermediate bundle $\tilde{b}$, so that no lottery involving $\tilde{b}$ can dominate $b^{root} \cup \{i\}$ or $b^* \setminus \{i\}$ either. The least-favorite property is what reduces this second part to a single, checkable restriction: it places $b^{root}\cup\{i\}$ and $b^*\setminus\{i\}$ at the two extremes of the intermediate bundles' values at $\bar{t}$, so the only lotteries that could threaten them are pinned down---one involving the root, the other involving the grand bundle---and one region for $\tilde{b}$'s endpoint values rules out both at once. The contrast with Theorem \ref{t33} is in how much the conditions demand: Theorem \ref{t33} keeps \emph{every} intermediate bundle on the menu, which requires all of them to stand in a precise relation to one another, while Theorem \ref{t34} keeps these two but not necessarily the rest, which requires conditions on the pair and merely enough discipline on the rest not to eliminate them.

Since $b^{root} \cup \{i\}$ and $b^* \setminus \{i\}$ are set-incomparable, the minimal optimal menu, which contains $\{b^{root}, b^{root} \cup \{i\}, b^* \setminus \{i\}, b^*\}$, cannot be nested---it is a tree. Economically, the two upgrade paths $b^{root} \rightarrow b^{root} \cup \{i\}$ and $b^{root} \rightarrow b^* \setminus \{i\}$ are differentiated enough in their marginal-revenue profiles that neither collapses into the other or into the base or grand bundle. This matches common practice: tech firms routinely sell $\{\text{``computer''},\; \text{``computer''}\!+\!\text{``keyboard lighting''},\; \text{``premium computer''}\!-\!\text{``keyboard lighting''},\; \text{``premium computer''}\}$.\footnote{For example, Alienware Aurora R16: \href{https://www.dell.com/en-us/shop/cty/pdp/spd/alienware-aurora-r16-desktop/useahbtsr16igtgy}{https://www.dell.com/en-us/shop/cty/pdp/spd/alienware-aurora-r16-desktop/useahbtsr16igtgy}.} An numerical example of Theorem \ref{t34} can be found in Online Appendix \ref{f4}

\section{Other Applications}
\label{sec:other}
I apply the framework to characterize optimal bundling with additive values (Section \ref{s62}) and establish results about the optimality of pure and nested bundling (Section \ref{s663}).
\subsection{Additive Values}
\label{s62}
This section considers one of the most natural cases in which the consumer's valuation function is additive. Rigorously, for a generic $b = \{i_1,...,i_j\}\in \mathcal{B} $, in which $i_1,...,i_j \in \{1,...,n\}, v(b,t)=\sum \limits_{k=1}^j v(\{i_k\},t), \forall t \in T$. 

The following theorem directly pins down the minimal optimal menu when values are additive:
\begin{theorem}
\label{t4}
    Suppose Assumptions \ref{a0} and \ref{a1} hold. If $(\dfrac{\phi(b,\bar{t})}{\phi(b,\underline{t})})_{|b|=1}$ are strictly ordered, i.e., \begin{equation*} \dfrac{\phi(b_{j_1},\bar{t})}{\phi(b_{j_1},\underline{t})} < \dfrac{\phi(b_{j_2},\bar{t})}{\phi(b_{j_2},\underline{t})}< ... < \dfrac{\phi(b_{j_n},\bar{t})}{\phi(b_{j_n},\underline{t})}, \end{equation*}
    where $\{b_{j_1},...,b_{j_n}\}=\{\{1\},...,\{n\}\}$, then when the consumer's valuation function is additive, the minimal optimal menu is a nested menu. In particular, it is \begin{equation*} \{\emptyset, \{b_{j_1}\}, \bigcup_{i=1}^2 \{b_{j_i}\}, ...,  \bigcup_{i=1}^{n-1}\{b_{j_i}\}, b^*\}.\footnote{It is not hard to see that $b^*=\{1,...,n\}=\bigcup_{i=1}^{n}\{b_{j_i}\}$.} \end{equation*} 
\end{theorem}

The additional condition on the strict order the $\phi$ ratios is arguably not demanding, as these ratios are naturally ordered, though may not strictly. The strictness is used to show that every bundle in the menu is not very weakly dominated. If the strictness assumption is not satisfied, then certain bundles in the menu will be very weakly dominated, making the menu not minimal optimal. In particular, if $\dfrac{\phi(b_{j_{k'}},\bar{t})}{\phi(b_{j_{k'}},\underline{t})}=\dfrac{\phi(b_{j_{k'+1}},\bar{t})}{\phi(b_{j_{k'+1}},\underline{t})}$ for $k' \in \{1,...,k-1\}$, then $\bigcup \limits_{i=1}^{k'}\{b_{j_i}\}$ is very weakly dominated by a convex combination of $\bigcup \limits_{i=1}^{k'-1}\{b_{j_i}\}$ and $\bigcup \limits_{i=1}^{k'+1}\{b_{j_i}\}$, a direct result of Lemma \ref{l2}. 


It is worth pointing out that, though the minimal optimal menu is nested, the nesting condition in \cite{yang2023nested} might not be satisfied: consider bundle $b_{j_1}$, for any $b \supsetneq b_{j_1}$, it is true that $\phi(b,\underline{t}) < \phi(b_{j_1}, \underline{t})$, but this observation does not rule out the possibility that $Q(b_{j_1}) > Q(b)$, in which $Q$ is the sold-alone quantity defined in Definition \ref{d4}. In other words, it is possible that $b_{j_1}$ is not dominated in the sense of \cite{yang2023nested}. Similarly, $b_{j_2}$ may also not be dominated in the sense of \cite{yang2023nested}. Nevertheless, $b_{j_1}$ and $b_{j_2}$ are not nested, a possible violation of the nesting condition.\footnote{As a side note, consider a setting where both Assumption \ref{a1} and all of \cite{yang2023nested}'s assumptions hold. If $b$ dominates $b'$ according to Definition \ref{def1}, then $b$ may not dominate $b'$ in the sense of \cite{yang2023nested}, as $b' \subseteq b$ may not hold. On the other hand, if $b$ dominates $b'$ in the sense of \cite{yang2023nested}, then $b'$ is either dominated by $b$ or is dominated by a convex combination of $b$ and $\emptyset$ in the sense of Definition \ref{def1}. }

\subsection{Optimality of Pure Bundling and Nested Bundling}
\label{s663}

This section presents results regarding the optimality of pure and nested bundling using conditions similar to the ones in \cite{ghili2023characterization} and \cite{yang2023nested}, with reconciliations into this paper's environment to better demonstrate how the ``taking and computing upper envelope'' approach shown above can be used to replicate their results. 

I first show a relatively general theorem, which is inspired by Theorem \ref{t1}, that captures what \cite{ghili2023characterization} and \cite{yang2023nested} aim to discover. 

\begin{theorem}
    \label{corollary: alternative}
    If there exists $\mathcal{B}' \subseteq \mathcal{B}$ such that 
    \begin{enumerate}
        \item for any $b,b' \in \mathcal{B}',$ (i) $\phi(b,t)-\phi(b',t)$ is single-crossing in $t$; (ii) $v(b,t)-v(b',t)$ is monotonic in $t$, and
        \item for any $t, \argmax \limits_{b \in \mathcal{B}} \phi(b,t) \subseteq \mathcal{B}'$,
    \end{enumerate}
    then the set of minimal optimal menus of problem \eqref{eq5} is the set of minimal optimal menus of problem \eqref{eq6}. In particular, every minimal optimal menu is a subset of $\mathcal{B}'$. 
\end{theorem}

The proof is straightforward from previous results: any minimal optimal menu of \eqref{eq6} does not contain bundles not included in $\mathcal{B}'$, so the fact that the set of minimal optimal menus of \eqref{eq6} is a subset of the set of minimal optimal menus of \eqref{eq5} naturally follows from Lemma \ref{l1} and the arguments in Theorem \ref{t1}. The other direction, namely the set of minimal optimal menus of \eqref{eq5} is a subset of the set of minimal optimal menus of \eqref{eq6}, follows from Theorem \ref{t1} by an identical argument.

In \cite{ghili2023characterization}, $\mathcal{B}'$ is the pure bundling menu $\{\emptyset, b^*\}$, whereas in \cite{yang2023nested}, $\mathcal{B}'$ is a set of nested bundles. It is worth mentioning that both \cite{ghili2023characterization} and \cite{yang2023nested} make key assumptions on quasi-concave differences of revenue functions. I aim to translate this assumption into single-crossing difference of virtual value functions in accordance with Assumption \ref{a0}. I first introduce the following proposition.

\begin{proposition}
\label{lemma:quasiconcave}
    Consider $z: T \rightarrow \mathbb{R}$ continuously differentiable. $z$ is strictly quasi-concave on $T$ if and only if there does not exist $t, t' \in T$ such that $t<t', \sign[z'(t)]<0, \sign[z'(t')]>0$, and $z'$ is not identically zero on any interval of $T$. 
\end{proposition}

Heuristically, apart from knife-edge cases, the conditions for $z'$ in Proposition \ref{lemma:quasiconcave} can be characterized as $z'$ is strictly single-crossing from above.\footnote{That is, $\sign[z']$ is decreasing, and $|\{t: z'(t)=0\}| \leq 1$.} \cite{yang2023nested}'s main assumptions are: for any $b_1 \subseteq b_2, v(b_2,t)-v(b_1,t)$ is strictly increasing in $t$ and $\pi(b_2,t)-\pi(b_1,t)$ is strictly quasi-concave in $t$, in which for any bundle $b$, $\pi(b,\cdot)$ is the revenue function when only selling $b$. By Proposition \ref{lemma:quasiconcave}, these assumptions can almost be translated (rigorously, a sufficient condition of) into Assumption \ref{a0} imposed on pairs of nested bundles. The similar story goes for \cite{ghili2023characterization}. Combining with their conditions on the sold-alone
quantity that guarantees the second condition in Theorem \ref{corollary: alternative} holds for their particular $\mathcal{B}'$, they obtain their results in the same spirit as Theorem \ref{corollary: alternative}.

I proceed to present results with more specific conditions. \cite{ghili2023characterization} says that pure bundling is optimal if and only if the sold-alone quantity of $b^*$ is the largest among the sold-alone quantities of all bundles. Following the same conditions, under Assumptions \ref{a0} and \ref{a1}, I am able to obtain the same result, with the help of an additional but natural assumption that makes $b^*$ the (strictly) most valuable bundle for type $\bar{t}$ consumer:

\begin{theorem}
    \label{t5}
    Suppose Assumptions \ref{a0}, \ref{a1}, and the assumption that $\{b^*\} = \argmax \limits_{b \in \mathcal{B}} v(b,\bar{t})$ hold. If $Q(b^*) \geq \max \limits_{b \neq b^*} Q(b)$, then $\{\emptyset, b^*\}$ is the minimal optimal menu. 
\end{theorem}

It is worth mentioning that the additional assumptions in Theorem \ref{t5} (and later in Theorem \ref{t6}) only focuses on the valuation of type $\bar{t}$ consumer, which is different from \cite{ghili2023characterization} and \cite{yang2023nested} where the corresponding assumptions have to hold for all types. 

The proof basically shows that any $b \notin \{\emptyset, b^*\}$ is very weakly dominated. It is not hard to see that, with the help of the additional assumption, $b$ that satisfies $\phi(b, \underline{t}) \leq \phi(b^*, \underline{t})$ is very weakly dominated by $b^*$. As for $b$ such that $\phi(b, \underline{t}) > \phi(b^*, \underline{t})$, it is very weakly dominated by a convex combination of $\emptyset$ and $b^*$ by the condition that $b^*$ has the largest sold-alone quantity. 

\cite{yang2023nested} provides a sufficient condition for the optimality of nested menus: the \emph{nesting condition}. Furthermore, a sufficient condition for the nesting condition is introduced by \cite{yang2023nested} as the \emph{union quantity condition}:  \begin{equation*} \label{eq114} \text{ for any } b_1,b_2 \in \mathcal{B}, Q(b_1 \cup b_2) \geq \min \{Q(b_1), Q(b_2)\}. \end{equation*} Similar to how I obtain Theorem \ref{t5}, following the \emph{union quantity condition} and with the help of an additional assumption that specifies type $\bar{t}$ consumer has a higher valuation on larger bundles, I am able to show the optimality of nested menus. Moreover, I am also able to pin down a concrete optimal menu like Proposition 9 of \cite{yang2023nested}. 

\begin{theorem}
    \label{t6}
     Suppose Assumptions \ref{a0}, \ref{a1}, and the assumption that $\text{ for any } b \subseteq b', v(b',\bar{t}) \geq v(b,\bar{t})$ hold. If (\ref{eq114}) holds, then the minimal optimal menu is nested. In addition, $\{\emptyset, b_1^\star, b_1^\star \cup b_2^\star,...,b^*\}$ is an optimal menu, in which $\{b_1^\star, ..., b_{2^n}^\star\}=\mathcal{B}$, and $b_i^\star$ is totally ordered by the sold-alone quantity: for any $i \leq j, Q(b_i^\star) \geq Q(b_j^\star)$.\footnote{It is not hard to see that $b_1^\star \cup b_2^\star \cup ... \cup b_{2^n}^\star=b^*$, as for any $i, b_i^\star \subseteq b^*$. } 
\end{theorem}

The arguments of the proof are similar to the ones in the proof of Theorem \ref{t5}, which shows that for any pair of non-nested bundles $b_1$ and $b_2$, if $Q(b_1) \leq Q(b_2)$, then $b_1$ is either very weakly dominated by $b_1 \cup b_2$ or very weakly dominated by a convex combination of $\emptyset$ and $b_1 \cup b_2$, so that there cannot exist a pair of non-nested bundles in the minimal optimal menu. Similar arguments for the second part of the theorem: the proposed menu is optimal because every bundle not included is very weakly dominated. However, while this construction guarantees the proposed menu is optimal, it does not guarantee it is the \emph{minimal} optimal menu. Some bundles within this menu could be weakly dominated by others, making them redundant without further assumptions.\footnote{One way to see this is that the order of the sold-alone quantity of the bundles inside the menu is unknown.}

\section{Conclusion}

This paper develops a constructive approach to solving the optimal bundling problem for a multiproduct monopolist facing a consumer with a one-dimensional type. The central insight is that under single-crossing differences of virtual values together with monotonic differences of valuations (Assumption \ref{a0}), the monopolist's problem simplifies to finding the upper envelope of the virtual value functions. I provide a complete characterization of this upper envelope using notions of dominance and best response. I show that the minimal optimal menu is unique and consists precisely of bundles that are not very weakly dominated, a condition I prove is equivalent to being a strict best response for some consumer type. This equivalence allows me to develop a practical algorithm that computes the optimal menu by iteratively eliminating dominated bundles.

My framework provides a tool for analyzing common, yet complex, sales strategies. I apply it to introduce and formalize tree bundling, where a menu is built around a common ``root'' bundle with multiple upgrade paths. I derive conditions under which such a structure is optimal and, in some cases, when it is optimal to offer all possible upgrades from the root—providing a theoretical foundation for the ``build-your-own'' models offered by many firms. The methodology is also applied to provide new insights on classic problems, including bundling with additive values and the optimality of pure and nested bundling.

A promising direction for future research is extending this framework to a multidimensional type space. A significant challenge in such settings, as discussed by \cite{rochet1998}, is the direct extension of the Myersonian approach, which is a prerequisite for a direct application of my upper-envelope approach. Although a well-defined virtual value function can either be obtained by picking a specific path to compute the line integral (\cite{armstrong1996}), or be obtained through the duality approach (\cite{cai2016}, \cite{bergemann2021}), the implementability of the former is hard to capture, whereas the latter only works well for finite types or distributional robustness problems. (\cite{carroll2017}, \cite{haghpanah2021when})

However, my characterization based on dominance may prove more portable, as both the equivalence between never strict best response and very weakly dominance (Proposition \ref{p1}) and the sufficiency to check the virtual values at $\underline{t}$ and $\bar{t}$ for dominance (Corollary \ref{c2}) still holds in the multidimensional type setting. An interesting question is how my approach connects to alternative dominance notions like the transfer dominance of \cite{manelli2007}. Developing these connections could provide a new and tractable path for analyzing complex multidimensional screening problems.

\bibliographystyle{apalike}
\bibliography{dominance}

\appendix

\addtocontents{toc}{\protect\setcounter{tocdepth}{1}}

\section{Proofs for Taking the Upper Envelope}

\subsection{Proof of Lemma \ref{l1}}

\begin{proof}
    By Assumption \ref{a0}, $v(b,t)-v(b',t)$ is either increasing or decreasing in $t$. Suppose it is decreasing in $t$, i.e., $v_t(b,t)-v_t(b',t) \leq 0$ for all $t \in T$, then $\phi(b, \bar{t})>\phi(b',\bar{t})$ implies that $v(b,t)-v(b',t)>0$ for all $t \in T$. Meanwhile, by the continuity of $\phi$ and $\phi(b, \underline{t})<\phi(b', \underline{t})$, there exists $t^* \in (\underline{t}, \bar{t})$ such that $\phi(b,t^*)=\phi(b',t^*)$, which further implies that
    \[ v(b,t^*)-v(b',t^*)=\dfrac{1-F(t^*)}{f(t^*)}[v_t(b,t^*)-v_t(b',t^*)] \leq 0,\]
    a contradiction to $v(b,t)-v(b',t)>0$ for all $t \in T$. Hence, $v(b,t)-v(b',t)$ is increasing in $t$. 
\end{proof}

\subsection{Proof of Theorem \ref{t1}}

\begin{proof}
    Denote the set of minimal optimal menus of the monopolist's problem (problem \eqref{eq5}) as $\mathcal{M}$, and the set of minimal optimal menus of the unconstrained problem (problem \eqref{eq6}) as $\mathcal{M}'$. I first show that $\mathcal{M}' \subseteq \mathcal{M}$.
    
    I begin with demonstrating for any $M' \in \mathcal{M}'$, $M'$ is a feasible menu of problem \eqref{eq5}, i.e., there exists a price scheme $\{p(t)\}_{t \in T}$ such that $\{b(t),p(t)\}_{t \in T}$ satisfies IC and IR, and $\{b(t)\}_{t \in T}=M'$. Claim that for any $b \in M',$ there exists $t' \in T$, which could be $b$-dependent, such that $\phi(b,t')>\max \limits_{\tilde{b} \neq b} \phi(\tilde{b},t')$. Suppose not, then there exists $b \in M'$ such that for all $t \in T$, there exists $\tilde{b} \neq b$, which could be $t$-dependent, such that $\phi(b,t) \leq \phi(\tilde{b},t)$. This implies that $M' \backslash \{b\}$ is also an optimal menu of problem \eqref{eq6}, a contradiction to $M'$ being minimally optimal. With this observation, by continuity of $\phi$, given $b$ and $t'$ such that $\phi(b,t')> \max \limits_{\tilde{b} \neq b} \phi(\tilde{b}, t'),$ there exists $\delta>0$ such that $\phi(b, \tilde{t})> \max \limits_{\tilde{b} \neq b} \phi(\tilde{b}, \tilde{t}),$ for any $\tilde{t} \in (t'-\delta, t'+\delta)$. Suppose there exist $t''>t'+\delta$ and $\delta'>0$ such that 
    \[\phi(b,t'') \leq \max \limits_{\tilde{b} \neq b} \phi(\tilde{b},t''); \phi(b,t''+\delta')>\max \limits_{\tilde{b} \neq b}\phi(\tilde{b},t''+\delta'),\] then there exists $\bar{b}$ such that \[\phi(b,t')-\phi(\bar{b},t')>0; \phi(b,t'')-\phi(\bar{b},t'') \leq 0; \phi(b, t''+\delta')-\phi(\bar{b},t''+\delta') >0,\] a contradiction to $\phi(b,t)-\phi(\bar{b},t)$ is monotonic in $t$. Therefore, $\forall b \in M'$, unless $\phi(b,t)>\max \limits_{\tilde{b} \neq b} \phi(\tilde{b},t), \forall t$, $b$ is the unique maximizer of $\phi$ on an open interval of $T$. 

    Consider $b,b' \in M'$ and $t<t'<t''<t'''$ such that $b$ is the unique maximizer of $\phi$ on $(t,t')$, $b'$ is the unique maximizer of $\phi$ on $(t'',t''')$, and there does not exist $b'' \in M'$ such that it is the unique maximizer of $\phi$ on an open interval within $[t',t'']$. Claim that it has to be the case that $\phi(b,\tilde{t})=\phi(b',\tilde{t})>\max \limits_{\tilde{b} \in M' \backslash \{b,b'\}} \phi(\tilde{b}, \tilde{t}), \forall \tilde{t} \in [t',t'']$. Suppose there exists $t^* \in [t',t'']$ such that $\max \limits_{\tilde{b} \in M' \backslash \{b,b'\}} \phi(\tilde{b},t^*) \geq \max \{\phi(b, t^*), \phi(b',t^*)\}$, then for the bundles belonging to $\argmax \limits_{\tilde{b} \in M' \backslash \{b,b'\}} \phi(\tilde{b},t^*)$, they are either not a unique maximizer, or are a unique maximizer on an open interval within $[t',t'']$, both lead to contradictions. Suppose there exists $t^{**} \in [t',t'']$ such that $\phi(b,t^{**})>(<) \phi(b',t^{**})$, then $b(b')$ is the unique maximizer at $t^{**}$, a contradiction to it being the unique maximizer only on $(t,t')((t'',t'''))$. Hence, the claim holds.  
   
   Therefore, every bundle $b$ in $M'$ gives the strictly highest value of $\phi$ for an open interval $I_b$, and the complementary $T \backslash \bigcup \limits_{b \in M'} I_b$ are regions where a pair of adjacent bundles give the same value of $\phi$.\footnote{``Adjacent'' in the sense of their unique maximizer open intervals.} In particular, a generic optimal allocation can be represented as:
    \begin{figure}[H]
    \centering
    \includegraphics[width=1\linewidth]{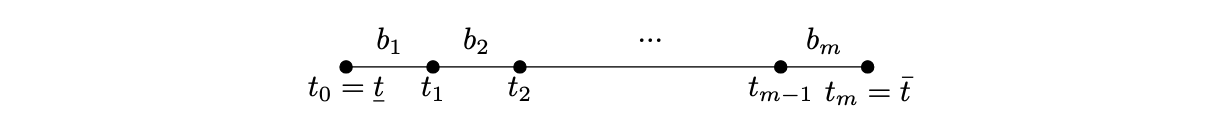}
    \caption{Interval structure of the optimal allocation}
\end{figure}
in which $M'=\{b_1,...,b_m\}$, and $\forall i \in \{1,...,m\}, \phi(b_i,t) \geq \max \limits_{b \neq b_i} \phi(b,t), \forall t \in [t_{i-1},t_i]$, with the inequality being strict on $(t_i^x, t_i^y)$, where $t_{i-1} \leq t_i^x<t_i^y\leq t_i$.
    
    I proceed to show that for any $i \in \{1,...,m-1\}$

    \begin{equation}
    \label{eq23}
\phi(b_i,\underline{t})>\phi(b_{i+1},\underline{t}); \phi(b_i, \bar{t})<\phi(b_{i+1}, \bar{t}).
    \end{equation}
    
    Suppose not, then there exists $i \in \{1,...,m-1\}$ such that one of the following cases has to hold:
    \begin{enumerate}
        \item $\phi(b_i,\underline{t}) \leq \phi(b_{i+1},\underline{t}); \phi(b_i, \bar{t})<\phi(b_{i+1}, \bar{t}).$
        \item $\phi(b_i,\underline{t})>\phi(b_{i+1},\underline{t}); \phi(b_i, \bar{t}) \geq \phi(b_{i+1}, \bar{t}).$
        \item $\phi(b_i,\underline{t}) \leq \phi(b_{i+1},\underline{t}); \phi(b_i, \bar{t})\geq \phi(b_{i+1}, \bar{t}).$
    \end{enumerate}

The first two cases are similar: under Assumption \ref{a0}, it is either $\phi(b_i,t) \leq \phi(b_{i+1},t), \forall t$, with a strict inequality at $\bar{t}$, or $\phi(b_i,t) \geq \phi(b_{i+1},t), \forall t$, with a strict inequality at $\underline{t}$. This, respectively, implies that $b_i$ and $b_{i+1}$ are redundant, a contradiction to minimal optimality. As for the third case, it suggests that \[\phi(b_{i+1},\underline{t})-\phi(b_i,\underline{t}) \geq 0; \phi(b_{i+1}, \dfrac{t_i^x+t_i^y}{2})-\phi(b_{i}, \dfrac{t_i^x+t_i^y}{2})<0; \phi(b_{i+1}, \dfrac{t_{i+1}^x+t_{i+1}^y}{2})-\phi(b_{i}, \dfrac{t_{i+1}^x+t_{i+1}^y}{2})>0.\] Since $\underline{t}<\dfrac{t_i^x+t_i^y}{2}<\dfrac{t_{i+1}^x+t_{i+1}^y}{2}$, this is a violation of Assumption \ref{a0}. (\ref{eq23}) thus holds, and by Lemma \ref{l1}, $v(b_{i+1},t)-v(b_i,t)$ is increasing in $t, \forall i \in \{1,...,m-1\}$. 

$\forall i \in \{1,...,m\}, t,t' \in (t_{i-1},t_i)$, IC requires that \begin{equation*} \begin{split} v(b_i,t)-p(t) \geq v(b_i,t)-p(t'); v(b_i,t')-p(t') \geq v(b_i,t')-p(t), \end{split} \end{equation*}
which implies $p(t)=p(t'), \forall t,t' \in (t_{i-1},t_i)$. To simplify notations, I write $p(t)=p(b_i), \forall t \in (t_{i-1},t_i)$. I construct $\{p(b_i)\}_{i \in \{1,...,m\}}$ such that $v(b_i,t_{i-1})-p(b_i)=v(b_{i-1},t_{i-1})-p(b_{i-1})$.\footnote{$v(b_0,t_0)-p(b_0)=v(\emptyset,\underline{t})$.} Claim that IC constraints hold if and only if \begin{equation*} \begin{split} v(b_i,t)-p(b_i) \geq v(b_{i-1},t)-p(b_{i-1}); v(b_i,t)-p(b_i) \geq v(b_{i+1},t)-p(b_{i+1}), \end{split} \end{equation*}
for all $t \in [t_{i-1},t_i], i \in \{1,...,m\}$, i.e., IC constraints hold for every pair of adjacent intervals. It suffices to show the ``if'' direction: if IC holds for every pair of adjacent intervals, then for all $i \in \{1,...,m\}, t \in [t_{i-1},t_i]$ \begin{equation*} \begin{split} v(b_i,t)-p(b_i) \geq v(b_{i-1},t)-p(b_{i-1}) &= v(b_{i-1},t)-v(b_{i-1},t_{i-1})+v(b_i,t_{i-1})-p(b_i) \\ 
&\Leftrightarrow \\ v(b_i,t)-v(b_{i-1},t) &\geq v(b_i,t_{i-1})-v(b_{i-1},t_{i-1}). \end{split} \end{equation*} By increasing differences, $\forall j<i, t \in (t_{i-1},t_i), v(b_i,t)-v(b_j,t)=\sum \limits_{k=i}^{j+1} v(b_k,t)-v(b_{k-1},t) \geq \sum \limits_{k=i}^{j+1} v(b_k,t_{k-1})-v(b_{k-1},t_{k-1})=p(b_i)-p(b_j)$. Similarly, for all $t \in [t_{i-1},t_i]$, \begin{equation*} \begin{split} v(b_i,t)-p(b_i) \geq v(b_{i+1},t)-p(b_{i+1})& =v(b_{i+1},t) +v(b_i,t_i)-v(b_{i+1},t_i)-p(b_i)
\\ & \Leftrightarrow \\ v(b_i,t)-v(b_{i+1},t) & \geq v(b_i,t_i)-v(b_{i+1},t_i) \end{split} \end{equation*}
By increasing differences, $\forall j>i, t \in [t_{i-1},t_i], v(b_i,t)-v(b_j,t) \geq \sum \limits_{k=i}^{j-1} v(b_k,t_k)-v(b_{k+1},t_k)=p(b_i)-p(b_j)$. The only thing left to show is that IC indeed holds for every pair of adjacent intervals. By increasing differences, for all $t \in [t_{i-1},t_i], i \in \{2,...,n\}, v(b_i,t)-v(b_{i-1},t) \geq v(b_i,t_{i-1})-v(b_{i-1},t_{i-1})= p(b_i)-p(b_{i-1}); v(b_i,t)-v(b_{i+1},t) \geq v(b_i,t_i)-v(b_{i+1},t_i)=p(b_i)-p(b_{i+1})$, which fulfills the claim.

Next I show that IR holds, and $v(a(\underline{t}),\underline{t})-p(\underline{t})=v(\emptyset,\underline{t})=0$. For any $i \in \{1,...,m\}, t \in [t_{i-1},t_i]$, $v(b_i,t)-p(b_i) \geq v(b_1,t)-p(b_1)=v(\emptyset,t)-p(\emptyset)$.\footnote{By the assumption that $\phi(\emptyset, \underline{t})>\max \limits_{b \neq \emptyset} \phi(b, \underline{t}) \Rightarrow b_1=\emptyset$. This assumption can be replaced by the assumption that $v(b, \cdot)$ is increasing in $t$ for any $b$ to still make IR hold: $v(b_i,t)-p(b_i) \geq v(b_1,t)-p(b_1) \geq v(b_1, \underline{t})-p(b_1)=0$.} Since $v(\emptyset,\underline{t})-p(\emptyset)=v(\emptyset,\underline{t})$, $p(\emptyset)=0$, which gives $v(b_i,t)-p(b_i) \geq v(\emptyset,t)=0$. This also gives $v(a(\underline{t}),\underline{t})-p(\underline{t})=v(b_1,\underline{t})-p(\underline{t})=v(\emptyset, \underline{t})-p(\emptyset)=v(\emptyset,\underline{t})=0$.

Therefore, for any $M' \in \mathcal{M}', M'$ is a feasible menu of problem \eqref{eq5}. As for why $M' \in \mathcal{M}'$ also generates the optimal revenue, note that for any $(a,p)$ that is IR,
\begin{equation*}
    \mathbb{E}[\sum \limits_{b \in \mathcal{B}}a_b(t)\phi(b,t)]-[v(a(\underline{t}),\underline{t})-p(\underline{t})] \leq \mathbb{E}[\max \limits_{b \in \mathcal{B}}\phi(b,t)],
\end{equation*}
which further implies that
\begin{equation*} \max \limits_{(a,p), \text{IR}} \mathbb{E}[\sum \limits_{b \in \mathcal{B}}a_b(t)\phi(b,t)]-[v(a(\underline{t}),\underline{t})-p(\underline{t})] \leq \mathbb{E}[\max \limits_{b \in \mathcal{B}}\phi(b,t)]. \end{equation*}
Hence, a feasible menu of \eqref{eq5} that solves \eqref{eq6} also solves \eqref{eq5}, which implies that $M' \in \mathcal{M}$, which further implies that $\mathcal{M}' \subseteq \mathcal{M}$. 

Regarding $\mathcal{M} \subseteq \mathcal{M}'$, suppose there exists $M \in \mathcal{M}$ such that $M \notin \mathcal{M}'$, then one of the following cases has to hold: (i) $M$ is not an optimal menu of problem \eqref{eq6}; (ii) $M$ is an optimal but not a minimal optimal menu of problem \eqref{eq6}.

By $\mathcal{M}' \subseteq \mathcal{M}$, case (i) and case (ii) imply $M$ is not an optimal menu of problem \eqref{eq5} and $M$ is not a minimal optimal menu of problem \eqref{eq5}, respectively: both lead to contradictions. Therefore, $\mathcal{M} \subseteq \mathcal{M}'$, which fulfills the proof.

\end{proof}

\section{Proofs for Characterizing the Upper Envelope}

\subsection{Proof of Proposition \ref{p1}}
\begin{proof}
    By Proposition 1 of \cite{kartik2023}, if $\phi$ has SCD$^{\star}$, and $T$ has a minimum ($\underline{t}$) and a maximum ($\bar{t}$), and $\{\phi(a,\cdot): T \rightarrow \mathbb{R}\}_{a \in \Delta(\mathcal{B})}$ is convex, then there exists $f_1: T \rightarrow \mathbb{R}$ and $f_2: T \rightarrow \mathbb{R}$, in which $f_1(\cdot)>0$, such that \begin{equation*} \tilde{\phi}(a,t)=\lambda(t)\tilde{\phi}(a,\bar{t})+(1-\lambda(t))\tilde{\phi}(a,\underline{t}), \end{equation*}
in which $\tilde{\phi}(a,t)=f_1(t)\phi(a,t)+f_2(t)$, and $\lambda: T \rightarrow [0,1]$ is an increasing function. $\tilde{\phi}$ can be re-expressed as $\tilde{\phi}(a,t)=\lambda(t)[\tilde{\phi}(a,\bar{t})-\tilde{\phi}(a,\underline{t})]+\tilde{\phi}(a,\underline{t})$ with $\lambda$ monotonic. By Theorem 3 of \cite{kartik2023}, $\tilde{\phi}$ has MD$^{\star}$, i.e., $\forall a,a' \in \Delta(\mathcal{B}), \tilde{\phi}(a',t)-\tilde{\phi}(a,t)$ is monotonic in $t$, which is equivalent to $\forall a,a' \in \Delta(\mathcal{B}), f_1(t)[\phi(a',t)-\phi(a,t)]$ is monotonic in $t$, which is further equivalent to $\forall a,a' \in \Delta(\mathcal{B}), f_1(t)[\phi(a',t)-\phi(a,t)]$ is quasi-concave in $t$.\footnote{$\forall a,a' \in \Delta(\mathcal{B})$, both $f_1(t)[\phi(a',t)-\phi(a,t)]$ and $f_1(t)[\phi(a,t)-\phi(a',t)]$ are quasi-concave in $t$ $\Leftrightarrow$ $f_1(t)[\phi(a',t)-\phi(a,t)]$ is monotonic in $t$, $\forall a,a' \in \Delta(\mathcal{B})$. } Furthermore, the following lemma helps me say more about the continuity of $f_1(t)[\phi(a',t)-\phi(a,t)]$.

\begin{lemma}
\label{l33}
    $f_1$ is continuous in $t$.
\end{lemma}

\begin{proof}
    Theorem 2 of \cite{kartik2023} shows that $\phi$ takes the form \begin{equation*} \phi(a,t)=g_1(a)h_1(t)+g_2(a)h_2(t)+h(t), \end{equation*}
    with $h_1$ and $h_2$ each single-crossing and ratio ordered. By Appendix B.5 (proof of Proposition 1) of \cite{kartik2023}, if $ (h_1(\underline{t}),h_2(\underline{t}))$ and $(h_1(\bar{t}),h_2(\bar{t}))$ are linearly dependent, then $f_1(t) \equiv 1$, which is trivially continuous. If $(h_1(\bar{t}),h_2(\bar{t}))$ are linearly independent, then $f_1(t)=\dfrac{1}{\alpha(t)+\beta(t)}$, in which $(h_1(t),h_2(t))=\alpha(t)(h_1(\bar{t}),h_2(\bar{t}))+\beta(t)(h_1(\underline{t}),h_2(\underline{t}))$. This gives
    \begin{equation*} \alpha(t)=\dfrac{h_1(\underline{t})h_2(t)-h_1(t)h_2(\underline{t})}{h_1(\underline{t})h_2(\bar{t})-h_1(\bar{t})h_2(\underline{t})}; \beta(t)=\dfrac{-h_1(\bar{t})h_2(t)+h_1(t)h_2(\bar{t})}{h_1(\underline{t})h_2(\bar{t})-h_1(\bar{t})h_2(\underline{t})}.\end{equation*}
    By Appendix B.4 (proof of Theorem 2) of \cite{kartik2023}, there exists $a_1,a_2 \in \Delta(\mathcal{B})$ such that $h_1(t)=\phi(a,t)-\phi(a_1,t); h_2(t)=\phi(a,t)-\phi(a_2,t)$. Since $\phi(a,t)$ is continuous in $t$ for any $a \in \Delta(\mathcal{B})$, both $h_1(t)$ and $h_2(t)$ are continuous in $t$, which further implies $\alpha(t)$ and $\beta(t)$, thus $f_1(t)=\dfrac{1}{\alpha(t)+\beta(t)}$ is continuous in $t$. 
\end{proof}

Now I proceed to prove the equivalence between never strict best response and very weakly dominance. For the ``if'' direction of Proposition \ref{p1}, if $b$ is very weakly dominated by $a' \neq b$, $\forall t$, consider $b^t \in \argmax \limits_{b \in supp(a')} \phi(b,t)$, then $\phi(b,t) \leq \phi(a,t) \leq \phi(b^t,t), \forall t$. Moreover, $\forall t$, there exists $b^t \in \argmax \limits_{b \in supp(a')} \phi(b,t)$ such that $b^t \neq b$. Otherwise, $\{b\}=\argmax \limits_{b \in supp(a')} \phi(b,t)$ and $a' \neq b$ implies that $\phi(b,t)>\phi(a',t)$, a contradiction to very weakly dominance. Hence, $b$ is a never strict best response because of the existence $b^t$ for any $t$.

As for the ``only if'' direction, denote the set of never strict best response of $\mathcal{B}_{NSB}$. For any $b \in \mathcal{B}_{NSB}$, let $u^b(a,t): \Delta(\mathcal{B}) \times T \rightarrow \mathbb{R}$ such that $u^b(a,t)=f_1(t)[\phi(b,t)-\phi(a,t)]$. Since both $\Delta(\mathcal{B})$ and $T$ are compact and convex, and $u^b$ is continuous and linear in $a$, and is continuous (Lemma \ref{l33}) and quasi-concave in $t$. By minmax theorem (in particular, \cite{sion1958}), \begin{equation*} \min \limits_{a \in \Delta(\mathcal{B})} \max \limits_{t \in T} u^b(a,t)= \max \limits_{t \in T} \min \limits_{a \in \Delta(\mathcal{B})} u^b(a,t). \end{equation*}

Note that $b \in \mathcal{B}_{NSB}$ implies that $\forall t \in T, \min \limits_{a \in \Delta(\mathcal{B})} u^b(a,t) \leq 0$, which further implies that $\min \limits_{a \in \Delta(\mathcal{B})} \max \limits_{t \in T} u^b(a,t)= \max \limits_{t \in T} \min \limits_{a \in \Delta(\mathcal{B})} u^b(a,t) \leq 0$. Consider $a^* \in \argmin \limits_{a \in \Delta(\mathcal{B})} \max \limits_{t \in T} u^b(a,t)$, then $\max \limits_{t \in T} u^b(a^*,t)= \min \limits_{a \in \Delta(\mathcal{B})} \max \limits_{t \in T} u^b(a,t) \leq 0 $, which further implies that $u^b(a^*,t)= f_1(t)[\phi(b,t)-\phi(a^*,t)]\leq 0, \forall t \in T$. As $f_1(t)>0, \forall t$, this is equivalent to say $\phi(b,t)-\phi(a^*,t) \leq 0, \forall t$. 

The only thing left to show is there exists $a^* \in \argmin \limits_{a \in \Delta(\mathcal{B})} \max \limits_{t \in T} u^b(a,t)$ such that $a^* \neq b$, i.e., $\{b\} \neq \argmin \limits_{a \in \Delta(\mathcal{B})} \max \limits_{t \in T} u^b(a,t)$. Suppose $\{b\}=\argmin \limits_{a \in \Delta(\mathcal{B})} \max \limits_{t \in T} u^b(a,t)$, then $\forall a \neq b, \max \limits_{t \in T} u^b(a,t)>0$.

\textbf{Step 1: }Consider $\mathcal{B}^b:=\{b' \in \mathcal{B}| \phi(b',t) < \phi(b,t), \forall t\}$, i.e., the set of (deterministic) bundles that are strictly dominated by $b$. Suppose that $\mathcal{B} \backslash \mathcal{B}^{b}= \emptyset$, then $\forall a \neq b, \forall t \in T, u^b(a,t) > 0$, a contradiction to the existence of $a$ and $t$ such that $u^b(a,t) \leq 0$. Hence, $\mathcal{B} \backslash \mathcal{B}^b \neq \emptyset$. 

\textbf{Step 2: }For any $\tilde{b} \in \mathcal{B} \backslash \mathcal{B}^b$, by Assumption \ref{a1}, the continuity of $\phi$, and $\forall a \neq b, \max \limits_{t \in T} u^b(a,t)>0$, one of the following two cases has to be correct:

\begin{enumerate}

\item there exists $t_{\tilde{b}}$ such that $\forall t \in [\underline{t}, t_{\tilde{b}}], \phi(b,t) \leq \phi(\tilde{b},t)$ and $ \forall t \in (t_{\tilde{b}},\bar{t}], \phi(b,t)>\phi(\tilde{b},t)$; \item there exists $t_{\tilde{b}}$ such that $\forall t \in [\underline{t}, t_{\tilde{b}}), \phi(b,t) > \phi(\tilde{b},t)$ and $ \forall t \in [t_{\tilde{b}},\bar{t}], \phi(b,t) \leq \phi(\tilde{b},t)$. 

\end{enumerate}

Consider $\{t_{\tilde{b}}|\tilde{b} \in \mathcal{B} \backslash \mathcal{B}^b, \phi(b,t)>\phi(\tilde{b},t), \forall t \in [\underline{t},t_{\tilde{b}}); \phi(b,t) \leq \phi(\tilde{b},t), \forall t \in [t_{\tilde{b}},\bar{t}]\}$ and $\{t_{\tilde{b}}|\tilde{b} \in \mathcal{B} \backslash \mathcal{B}^b, \phi(b,t)>\phi(\tilde{b},t), \forall t \in (t_{\tilde{b}}, \bar{t}]; \phi(b,t) \leq \phi(\tilde{b},t), \forall t \in [\underline{t},t_{\tilde{b}}]\}$. If the former set is empty, let $t^*$ be the largest element of the latter set, i.e., $t^*:=\max \{t_{\tilde{b}}|\tilde{b} \in \mathcal{B} \backslash \mathcal{B}^b, \phi(b,t)>\phi(\tilde{b},t), \forall t \in (t_{\tilde{b}}, \bar{t}]; \phi(b,t) \leq \phi(\tilde{b},t), \forall t \in [\underline{t},t_{\tilde{b}}]\}$. If $t^* < \bar{t}$, then $\forall \tilde{b} \in \mathcal{B} \backslash \mathcal{B}^b, \forall t \in (t^*,\bar{t}], \phi(b,t)>\phi(\tilde{b},t)$, which further implies that $\forall \tilde{b} \in \mathcal{B} \backslash \{b\}, t \in (t^*,\bar{t}], \phi(b,t)>\phi(\tilde{b},t)$, a contradiction to $b$ is a never strict best response. If $t^*=\bar{t}$, then there exists $\tilde{b}$ such that $\phi(b,t) \leq \phi(\tilde{b},t), \forall t \in T$, a contradiction to $\max \limits_{t \in T} u^b(\tilde{b},t) >0$. Similar discussions lead to contradiction for the case where the latter set is empty. 

\textbf{Step 3: }If both sets are nonempty, let $t_*$ be the smallest element of the former set, i.e., $t_{*}:= \min \{t_{\tilde{b}}|\tilde{b} \in \mathcal{B} \backslash \mathcal{B}^b, \phi(b,t)>\phi(\tilde{b},t), \forall t \in [\underline{t},t_{\tilde{b}}); \phi(b,t) \leq \phi(\tilde{b},t), \forall t \in [t_{\tilde{b}},\bar{t}]\}$.\footnote{$t_*$ and $t^*$ are well-defined because $\mathcal{B}$ is finite.} 

\textbf{Step 3.1: }If $t_*>t^*$, then for all $\tilde{b} \in \mathcal{B} \backslash \{b\}$ and $t \in (t^{*},t_*), \phi(b,t)>\phi(\tilde{b},t)$, a contradiction to $b$ is a never strict best response. 

\textbf{Step 3.2: }If $t_* \leq t^*$, consider $\tilde{b}_1 $ and $\tilde{b}_2$ such that $t_{\tilde{b}_1}=t_*, t_{\tilde{b}_2}=t^*$. By construction, for all  $t \in [t_*,t^*], \phi(b,t) \leq \phi(\tilde{b}_1,t), \phi(b,t) \leq \phi(\tilde{b}_2,t)$. Consider $a \in \Delta(\{\tilde{b}_1,\tilde{b}_2\})$ such that $\phi(a,\underline{t})=\phi(b,\underline{t})$.\footnote{The existence of $a$ is guaranteed because $\phi(\tilde{b}_1,\underline{t}) < \phi(b,\underline{t}) \leq \phi(\tilde{b}_2,t)$.}

\textbf{Step 3.2.1: }If there exists $\tilde{b} \in supp(a)$ and $\tilde{t} \in [t_*,t^*]$ such that $\phi(\tilde{b},\tilde{t}) > \phi(b, \tilde{t})$, then $\phi(a,\tilde{t})>\phi(b,\tilde{t})$. By Assumption \ref{a1}, $\phi(a,t) \geq \phi(b,t), \forall t$, a contradiction to $\forall a \neq b, \max \limits_{t \in T} u^b(a,t)>0$ as clearly the constructed $a$ is different from $b$. 

\textbf{Step 3.2.2: }If for all $\tilde{b} \in supp(a)$ and $\tilde{t} \in [t_*,t^*], \phi(\tilde{b},\tilde{t})=\phi(b,\tilde{t})$, then $\phi(a,\tilde{t})=\phi(b,\tilde{t}), \forall \tilde{t} \in [t_*,t^*]$. If $\phi(a,\bar{t}) \geq \phi(b,\bar{t})$, then by Assumption \ref{a1}, $\phi(a,t) \geq \phi(b,t), \forall t$. Since $a \neq b$, this is also a contradiction to $\forall a \neq b, \max \limits_{t \in T} u^b(a,t)>0$.\footnote{$a \neq b$ holds even when $\phi(a,t)=\phi(b,t), \forall t$, as $a \in \Delta(\mathcal{B} \backslash \mathcal{B}^b)$.} If $\phi(a,\bar{t})<\phi(b,\bar{t})$, then there exists $\varepsilon>0$ such that for $a'=(1-\varepsilon)a+\varepsilon b_1$, $\phi(a',\underline{t})<\phi(b,\underline{t}), \phi(a', \tilde{t}) \geq (=) \phi(b,\tilde{t}), \forall \tilde{t} \in [t_*,t^*], \phi(a',\bar{t})<\phi(b,\bar{t})$. In that case, $\phi(a',t)-\phi(b,t)$ is not single-crossing in $t$, a contradiction to Assumption \ref{a1}. Thus, there exists $a^* \in \argmin \limits_{a \in \Delta(\mathcal{B})} \max \limits_{t \in T} u^b(a,t)$ such that $a^* \neq b$, and $b$ is very weakly dominated by $a^*$.

\end{proof}

\subsection{Proof of Proposition \ref{p4}}

\begin{proof}
    \textbf{Step 1: }I first show that $\{b_1,...,b_m\}$ is an optimal menu if and only if satisfying everything but the ``$\forall i \in \{1,...,m\}, b_i$ is not very weakly dominated'' condition in Proposition \ref{p4}. 

    ``only if'': If $\{b_1,...,b_m\}$ is an optimal menu, then it contains every strict best response bundles. By Proposition \ref{p1}, it is equivalent to say that $\{b_1,...,b_m\}$ contains all not very weakly dominated bundles. As every very weakly dominated bundle is very weakly dominated by a convex combination of not very weakly dominated bundles, any $b \in \mathcal{B} \backslash \{b_1,...,b_m\}$, which is very weakly dominated, is very weakly dominated by a convex combination of $b_1,...,b_m$. 

    ``if'': Suppose that every very weakly dominated bundle is very weakly dominated by a convex combination of $b_1,...,b_m$, but $\{b_1,...,b_m\}$ is not an optimal menu for problem \eqref{eq6}, then there exists $b \in \mathcal{B}$ and $t_0 \in T$ such that $\phi(b,t_0) > \max \limits_{i \in \{1,...,m\}} \phi(b_i,t_0)$. Consider $\underline{a}$ which is very weakly dominated, thus there exists $\lambda_1,..,\lambda_m \in [0,1], \sum \limits_{i=1}^m \lambda_i=1$ such that $\underline{a}$ is very weakly dominated by $\sum \limits_{i=1}^m \lambda_ib_i$, i.e., $\phi(\underline{a},t) \leq \sum \limits_{i=1}^m \lambda_i \phi(b_i,t), \forall t \in T$. Since $\phi(b,t_0) > \max \limits_{i \in \{1,...,m\}} \phi(b_i,t_0)$, there exists $\lambda \in (0,1)$ such that $\lambda \phi(\underline{a},t_0)+(1-\lambda) \phi(b,t_0) > \max \limits_{i \in \{1,...,m\}} \phi(b_i,t_0)$, indicating that $\lambda \underline{a} +(1-\lambda) b$ cannot be very weakly dominated by any convex combination of $b_1,...,b_m$. Meanwhile, 
    \[\phi(\lambda \underline{a}+(1-\lambda) b, t) \leq \sum \limits_{i=1}^m \lambda \lambda_i \phi(b_i,t)+(1-\lambda) \phi(b,t), \forall t \in T.\] Note that $\underline{a} \neq \sum \limits_{i=1}^m \lambda_i b_i$ implies that $\lambda \underline{a}+(1-\lambda) b \neq  \sum \limits_{i=1}^m \lambda \lambda_i b_i +(1-\lambda) b$, which further implies that $\lambda \underline{a}+(1-\lambda)b$ is a very weakly dominated bundle, in particular, is very weakly dominated by $\sum \limits_{i=1}^m \lambda \lambda_i b_i+(1-\lambda)b$. This is a contradiction to the condition that every very weakly dominated bundle is very weakly dominated by a convex combination of $b_1,...,b_m$. Hence, by contradiction, $\{b_1,...,b_m\}$ is an optimal menu for problem \eqref{eq6}. 

    \textbf{Step 2: }I proceed to discuss the minimal optimality when considering the previous ignored condition. 
    
    ``only if'': If $\{b_1,..,b_m\}$ is a minimal optimal menu, then by the argument above, every weakly dominated bundle is very weakly dominated by a convex combination of $b_1,...,b_m$. In addition, if there exists $b_i \in \{b_1,...,b_m\}$ such that $b_i$ is very weakly dominated, then by Proposition \ref{p1}, it is a never strict best response, so that $\{b_1,...,b_{i-1},b_{i+1},...,b_m\}$ is also an optimal menu, a contradiction to $\{b_1,...,b_m\}$ being minimal. Hence, $\forall i \in \{1,...,m\}, b_i$ is not very weakly dominated.   

    ``if'': If every weakly dominated bundle is very weakly dominated by a convex combination of $b_1,...,b_m$, then by the argument above, $\{b_1,...,b_m\}$ is an optimal menu for problem \eqref{eq6}. Suppose $\{b_1,...,b_m\}$ is not a minimal optimal menu, then there exists $B \subsetneq \{b_1,...,b_m\}$ such that $B$ is an optimal menu for problem \eqref{eq6}. Consider $b' \in B \backslash \{b_1,...,b_m\}$, as $b'$ is not very weakly dominated, by Proposition \ref{p1}, $b'$ is a strict best response, i.e., for any $b \in B$, there exists $t \in T$ (does not depend on $b$) such that $\phi(b',t) > \phi(b,t)$, a contradiction to $B$ being optimal. Thus, by contradiction, $\{b_1,...,b_m\}$ is a minimal optimal menu for problem \eqref{eq6}. 

     \textbf{Step 3: }Regarding the uniqueness result, here is a direct argument without the specification that $\{b_1,...,b_m\}$ is exactly the set of not very weakly dominated bundles: suppose there exist two distinctive minimal optimal menus $B_1=\{b_1,....b_x\},B_2=\{b_1',...,b_y'\}$, then there must exist $b_i' \in B_2$ such that $b_i' \notin B_1$.\footnote{Otherwise, $B_1=B_2$, as $B_2 \subsetneq B_1$ is impossible. } By Proposition \ref{p4}, $b_i'$ is not very weakly dominated, so that by Proposition \ref{p1}, $b_i'$ is a strict best response, indicating that there exists $t_i' \in T$ such that $\phi(b_i',t_i')> \max \limits_{b_i \in B_1} \phi(b_i,t_i')$. Again by Proposition \ref{p4}, there exists a very weakly dominated bundle $\underline{b}$ and $\lambda_1,...,\lambda_x \in [0,1], \sum \limits_{i=1}^x \lambda_i=1$ such that $\underline{b}$ is very weakly dominated by $\sum \limits_{i=1}^x \lambda_i b_i$. Since $\phi(b_i',t_i')> \max \limits_{b_i \in B_1} \phi(b_i,t_i')$, there exists $\lambda \in (0,1)$ such that $\phi(\lambda b_i'+(1-\lambda) \underline{b},t_i')>\max \limits_{b_i \in B_1} \phi(b_i,t_i')$, so that $\lambda b_i'+(1-\lambda) \underline{b}$ cannot be very weakly dominated by any convex combination of $b_1,...,b_x$. However, $\lambda b_i'+(1-\lambda) \underline{b}$ is very weakly dominated by $\lambda b_i'+(1-\lambda) \sum \limits_{i=1}^x \lambda_i b_i$, a contradiction to $B_1$ being a minimal optimal menu (Proposition \ref{p4}). Therefore, the minimal optimal menu for problem \eqref{eq6} is unique. 

\end{proof}

\section{Proofs for Computing the Upper Envelope}

\subsection{Proof of Lemma \ref{l2}}
\begin{proof}
By Corollary \ref{c2}, $b'$ is very weakly dominated by a convex combination of $b$ and $b''$ if and only if there exists $\lambda \in [0,1]$ such that $\phi(b',\bar{t}) \leq \phi(\lambda b +(1-\lambda) b'',\bar{t}); \phi(b', \underline{t}) \leq \phi(\lambda b +(1-\lambda) b'', \underline{t}) $. 

``if'': If $\dfrac{\phi(b,\bar{t})-\phi(b',\bar{t})}{\phi(b',\bar{t})-\phi(b'',\bar{t})} \geq \dfrac{\phi(b,\underline{t})-\phi(b',\underline{t})}{\phi(b',\underline{t})-\phi(b'',\underline{t})}$, then there exists $\lambda \in (0,1)$ such that \begin{equation*} \dfrac{\phi(b,\bar{t})-\phi(b',\bar{t})}{\phi(b',\bar{t})-\phi(b'',\bar{t})} \geq \dfrac{1-\lambda}{\lambda} \geq \dfrac{\phi(b,\underline{t})-\phi(b',\underline{t})}{\phi(b',\underline{t})-\phi(b'',\underline{t})}, \end{equation*} which immediately gives $\phi(b',\bar{t}) \leq \lambda \phi(b,\bar{t})+(1-\lambda) \phi(b'',\bar{t}) = \phi(\lambda b +(1-\lambda) b'',\bar{t}); \phi(b', \underline{t}) \leq \lambda \phi(b,\underline{t})+(1-\lambda) \phi(b'',\underline{t})= \phi(\lambda b +(1-\lambda) b'', \underline{t}) $, i.e., $b'$ is very weakly dominated by $\lambda b+(1-\lambda)b'$.\footnote{$\lambda \neq 1$ because $\phi(b,\bar{t})>\phi(b',\bar{t})$ and $\phi(b,\underline{t})<\phi(b',\underline{t})$. }

``only if'': If $b'$ is very weakly dominated by $\lambda b+(1-\lambda) b''$, then $\dfrac{\phi(b,\bar{t})-\phi(b',\bar{t})}{\phi(b',\bar{t})-\phi(b'',\bar{t})} \geq \dfrac{1-\lambda}{\lambda} \geq \dfrac{\phi(b,\underline{t})-\phi(b',\underline{t})}{\phi(b',\underline{t})-\phi(b'',\underline{t})}$. 

\end{proof}

\subsection{Proof of Theorem \ref{t2}}

As discussed in Section \ref{s22}, it suffices to show that the output of Algorithm \ref{al1} is the minimal optimal menu for problem \eqref{eq6} under Assumption \ref{a1}. 

If there do not exist very weakly dominated bundles, then no bundles will be removed in neither Step 1 nor Step 2, and the algorithm stops at Step 2, with the outcome being $\mathcal{B}$, which, by Proposition \ref{p4}, is the minimal optimal menu for problem \eqref{eq6}.

  If there exist very weakly dominated bundles, by Proposition \ref{p4}, it suffices to show that the output of Algorithm \ref{al1}, denoted as $\{b_1,...,b_m\}$ , satisfies
    \begin{enumerate}
        \item every very weakly dominated bundle is very weakly dominated by a convex combination of $b_1,...,b_m$, \label{cdd1}
        \item for any $i \in \{1,...,m\}, b_i$ is not very weakly dominated. \label{cdd2}
    \end{enumerate}
    First, it is straightforward to see that every bundle belongs to $\mathcal{B} \backslash \{b_1,...,b_m\}$ is very weakly dominated, for they are either very weakly dominated by $b_i \in \{b_1,...,b_m\}$ (meaning that they are deleted in Step 1) or are very weakly dominated by a convex combination of $b_i$ and $b_{i+1}, i \in \{1,...,m-1\}$ (meaning that they are deleted in Step 2 and its iterations). 

Next, I am going to show that for any $i \in \{1,...,m\}, b_i$ is not very weakly dominated. Showing this not only directly makes $\{b_1,..,b_m\}$ satisfy condition \ref{cdd2} above, but also makes $\{b_1,...,b_m\}$ satisfy condition \ref{cdd1} above, as $\mathcal{B} \backslash \{b_1,...,b_m\}$ is exactly the set of very weakly dominated bundles. Let $\{b_1,...,b_m\}$ satisfy $\phi(b_1,\bar{t})>...>\phi(b_m,\bar{t}); \phi(b_1,\underline{t})<...<\phi(b_m,\underline{t})$. To show for any $i \in \{1,...,m\}$, $b_i$ is not very weakly dominated, again by each bundle belonging to $\mathcal{B} \backslash \{b_1,...,b_m\}$ is very weakly dominated by a convex combination of $b_1,...,b_m$, it suffices to show that $b_i$ is not very weakly dominated by any convex combination of $b_1,...,b_{i-1},b_{i+1},...,b_m, \forall i \in \{1,...,m\}$. Clearly, since $\{b_1\}=\argmax \limits_{b_i \in \{b_1,...,b_m\}} \phi(b_i,\bar{t}), \{b_m\}=\argmax \limits_{b_i \in \{b_1,...,b_m\}} \phi(b_i,\underline{t})$, both $b_1$ and $b_m$ are not very weakly dominated. As for $i \in \{2,...,m-1\}$, $b_i$ is not very weakly dominated by any convex combination of $b_{i-1}$ and $b_{i+1}$, otherwise $b_i$ will be deleted by (iterated) Step 2. 
    
    Now I am going to introduce the following lemma to help me establish the equivalence between $b_i$ not very weakly dominated by any convex combination of $b_{i-1}$ and $b_{i+1}$ and $b_i$ not very weakly dominated.
    \begin{lemma}
    \label{l5}
        Under Assumption \ref{a1}, for $b,b',b'' \in \mathcal{B}$ such that $\phi(b,\underline{t})<\phi(b',\underline{t})<\phi(b'',\underline{t}); \phi(b,\bar{t})>\phi(b',\bar{t})>\phi(b'',\bar{t})$, let the intersection of $\phi(b,\cdot)$ and $\phi(b',\cdot)$ be $t_{bb'}$, and the intersection of $\phi(b',\cdot)$ and $\phi(b'',\cdot)$ be $t_{b'b''}$. $b'$ is very weakly dominated by a convex combination of $b$ and $b''$ if and only if $t_{bb'} \leq t_{b'b''}$.
    \end{lemma}

    \begin{proof}
        By $\phi(b,\underline{t})<\phi(b',\underline{t})<\phi(b'',\underline{t}); \phi(b,\bar{t})>\phi(b',\bar{t})>\phi(b'',\bar{t})$ and Assumption \ref{a1}, both $t_{bb'}$ and $t_{b'b''}$ are well-defined and unique. ``only if'': If $b'$ is very weakly dominated by a convex combination of $b$ and $b''$, but $t_{bb'}>t_{b'b''}$, then by Assumption \ref{a1}, for any $t \in (t_{bb'},t_{b'b''}), \phi(b',t)>\phi(b,t); \phi(b',t)>\phi(b'',t)$. In that case, for any $t \in (t_{bb'},t_{b'b''})$ and $\lambda \in [0,1], \phi(b',t)>\phi(\lambda b +(1-\lambda) b'',t)$, a contradiction to $b'$ is very weakly dominated by a convex combination of $b$ and $b''$. Hence, $t_{bb'} \leq t_{b'b''}$. 

        ``if'': If $t_{bb'} < t_{b'b''}$, then by Assumption \ref{a1}, $\forall t \in (t_{bb'},t_{b'b''}), \phi(b,t)>\phi(b',t); \phi(b'',t)>\phi(b',t)$. Since $\phi(b,\underline{t})<\phi(b',\underline{t})<\phi(b'',\underline{t})$, there exists $\lambda \in (0,1)$ such that $\phi(b',\underline{t})=\phi(\lambda b +(1-\lambda) b'',\underline{t})$. For $t \in (t_{bb'},t_{b'b''}), \phi(b',t)<\phi(\lambda b+(1-\lambda)b'',t)$, thus by Assumption \ref{a1}, $\phi(b',t) \leq \phi(\lambda b+(1-\lambda) b'',t), \forall t \in T$, i.e., $b'$ is very weakly dominated by $\lambda b+(1-\lambda) b''$. 

        If $t_{bb'}=t_{b'b''}$, then it has to be the case that $t_{bb'}=t_{b'b''}=t_{bb''}=\dot{t}$, i.e., $\phi(b,\cdot),\phi(b',\cdot)$, and $\phi(b'',\cdot)$ co-intersect. Again consider $\lambda$ such that $\phi(b',\underline{t})=\phi(\lambda b +(1-\lambda) b'',\underline{t})$. Claim that $\phi(b',\bar{t}) \leq \phi(\lambda b+(1-\lambda)b'',\bar{t})$. If $\phi(b',\underline{t}) > \phi(\lambda b +(1-\lambda) b'',\bar{t})$, then there exists $\varepsilon >0 $ such that $\phi(b',\underline{t})>\phi((\lambda+\varepsilon)b +(1-\lambda - \varepsilon)b'',\underline{t}), \phi(b',\bar{t})>\phi((\lambda +\varepsilon)b+(1-\lambda-\varepsilon)b'',\bar{t})$ but $\phi (b',\dot{t})=\phi((\lambda+\varepsilon)b+(1-\lambda-\varepsilon),\dot{t})$. Since $\underline{t}<\dot{t}<\bar{t}$, $\phi(b',\cdot)-\phi((\lambda+\varepsilon)b+(1-\lambda-\varepsilon)b'',\cdot)$ is not single-crossing, a contradiction to Assumption \ref{a1}. Therefore, $\phi(b',\bar{t}) \leq \phi(\lambda b+(1-\lambda)b'',\bar{t})$. By Assumption \ref{a1}, $\phi(b',t) \leq \phi(\lambda b+(1-\lambda) b'', t), \forall t \in T$, i.e., $b'$ is very weakly dominated by $\lambda b+(1-\lambda)b''$. 

        In conclusion, $t_{bb'} \leq t_{b'b''}$ implies that $b'$ is very weakly dominated by a convex combination of $b$ and $b''$. 
    \end{proof}
    As $\forall i \in \{2,...,m-1\}, b_i$ is not very weakly dominated by any convex combination of $b_{i-1}$ and $b_{i+1}$, by Lemma \ref{l5}, it has to be the case that $t_{b_{i-1}b_i}>t_{b_ib_{i+1}}, \forall i \in \{2,...,m-1\}$. Claim that $\forall i \in \{2,...,m-1\}, t \in (t_{b_ib_{i+1}},t_{b_{i-1}b_i}), \{b_i\}=\argmax \limits_{b \in \{b_1,...,b_m\}} \phi(b,t)$, i.e., $b_i$ is a strict best response when the set of bundle is $\{b_1,...,b_m\}$.\footnote{In fact, since bundles in $\mathcal{B} \backslash \{b_1,...,b_m\}$ are very weakly dominated, $b_i$ is also a strict best response when the set of bundles is $\mathcal{B}$. } This is because for any $t \in (t_{b_ib_{i+1}}, t_{b_{i-1}b_i})$, $t>t_{b_ib_{i+1}}>t_{b_{i+1}b_{i+2}}>...>t_{b_{m-1}b_m}$ implies that $\phi(b_i,t)>\phi(b_{i+1},t)>...>\phi(b_m,t)$. Meanwhile, $t<t_{b_{i-1}b_i}<t_{b_{i-2}b_{i-1}}<...<t_{b_1b_2}$ implies that $\phi(b_i,t)>\phi(b_{i-1},t)>...>\phi(b_1,t)$.  By Proposition \ref{p1}, $b_i$ is not very weakly dominated by any convex combination of $b_1,..,b_{i-1},b_{i+1},...,b_m$, hence not very weakly dominated.

\section{Proofs for Tree Bundling}

\subsection{Proofs of Proposition \ref{p33} and \ref{p34}}

\begin{proof}
    For Proposition \ref{p33}, it is not hard to see that $\forall b$ such that $b \not \supset b^{root}$, $b$ is either very weakly dominated by $b^{root}$ or very weakly dominated by a convex combination of $b^{root}$ and $\emptyset$: if $\phi(b,\underline{t}) \leq \phi(b^{root}, \underline{t})$, then by $\phi(b, \bar{t})=v(b,\bar{t})<v(b^{root}, \bar{t})=\phi(b^{root},\bar{t})$, $b$ is very weakly dominated by $b^{root}$. Meanwhile, if $\phi(b, \underline{t})>\phi(b^{root}, \underline{t})$, then $\phi(b^{root}, \underline{t})<\phi(b, \underline{t})<\phi(\emptyset, \underline{t}); \phi(b^{root}, \bar{t})>\phi(b, \bar{t})>\phi(\emptyset, \bar{t})$. Following the notations in Lemma \ref{l5}, claim that $t_{b^{root}b}<t_{b\emptyset}$. If $t_{b^{root}b} \geq t_{b\emptyset}$, then by Assumption \ref{a0}, $\phi(b^{root}, t_{b \emptyset}) \leq \phi(b, t_{b \emptyset})=0 \Rightarrow t_{b^{root} \emptyset} \geq t_{b \emptyset} \Rightarrow Q(b^{root}) \leq Q(b)$, a contradiction to $b^{root}$ has the strictly largest sold-alone quantity. The claim thus holds, and by Lemma \ref{l5}, $b$ is very weakly dominated by a convex combination of $b^{root}$ and $\emptyset$. To fulfill the proof, it suffices to show that $b^{root}$ is not very weakly dominated, so that it will be included in the minimal optimal menu and plays the role of ``root bundle''. Suppose that there exists $b' \supsetneq b^{root}$ such that $\phi(b', \underline{t}) \geq \phi(b^{root}, \underline{t})$, then by Assumption \ref{a0}, $\phi(b^{root}, t_{b' \emptyset}) \leq \phi(b', t_{b' \emptyset})=0 \Rightarrow t_{b^{root} \emptyset} \geq t_{b' \emptyset} \Rightarrow Q(b^{root}) \leq Q(b')$, again a contradiction to $b^{root}$ has the strictly largest sold-alone quantity. Hence, $\forall b' \supsetneq b^{root}$, $\phi(b', \underline{t})<\phi(b^{root}, \underline{t})$, i.e., $b^{root}$ is not very weakly dominated by any $b' \subsetneq b^{root}$. Meanwhile, if there exists $b' \supsetneq b^{root}$ such that $b^{root}$ is very weakly dominated by a convex combination of $b'$ and $\emptyset$, then by Lemma \ref{l5}, $t_{b'b^{root}} \leq t_{b^{root} \emptyset} \Rightarrow t_{b' \emptyset} \leq t_{b^{root} \emptyset}$, which is also a contradiction. Therefore, when running Algorithm \ref{al1}, $b$ can never be eliminated, so it will be included in the minimal optimal menu.

    As for Proposition \ref{p34}, by the same arguments as above, $b^{root}$ has to be included in the minimal optimal menu. To show that the minimal optimal menu is either a tree menu or a nested menu, it suffices to show that $b$ is very weakly dominated for all $b \not \supset b^{root}$. Since $Q(b^{root})>Q(b), Q(b^{root} \cup b) \geq \min \{Q(b^{root}), Q(b)\}$ implies that $Q(b^{root} \cup b) \geq Q(b)$. In addition, $\phi(b^{root} \cup b, \bar{t})>\phi(b, \bar{t})$. Hence, with similar arguments as above, $b$ is either very weakly dominated by $b^{root} \cup b$ or is very weakly dominated by a convex combination of $b^{root} \cup b$ and $\emptyset$. 
\end{proof}

\subsection{Proof of Theorem \ref{t33}}

\begin{proof}
    Consider $\tilde{b}$ such that $b^{root} \subsetneq \tilde{b} \subsetneq b^*$. By the monotonicity of $v$ at $\bar{t}$, there exists $\lambda \in (0,1)$ such that $\phi(\tilde{b},\bar{t})=\lambda \phi(b^*,\bar{t})+(1-\lambda) \phi(b^{root}, \bar{t})$. In that case, by the existence of $\mathfrak{f}$,
    \begin{equation*} \dfrac{\phi(\tilde{b},\bar{t})-\phi(b^{root},\bar{t})}{\phi(b^*,\bar{t})-\phi(b^{root},\bar{t})}= \lambda \Rightarrow \dfrac{\phi(\tilde{b},\underline{t})-\phi(b^{root},\underline{t})}{\phi(b^*,\underline{t})-\phi(b^{root},\underline{t})}=\mathfrak{f}(\lambda) \Rightarrow \phi(\tilde{b},\underline{t})=\mathfrak{f}(\lambda) \phi(b^*, \underline{t})+(1-\mathfrak{f}(\lambda))\phi(b^{root}, \underline{t}). \end{equation*} For any $b^{root} \subsetneq \tilde{b}_1, \tilde{b}_2 \subsetneq b^*$ such that  \begin{equation} \label{eq55} \phi(\tilde{b}_1,\bar{t})>\phi(\tilde{b},\bar{t})>\phi(\tilde{b}_2,\bar{t}), \end{equation} it has to be the case that $\lambda_1>\lambda>\lambda_2$, in which
    \begin{equation*} \begin{pmatrix} \phi(\tilde{b}_1, \bar{t}) \\ \phi(\tilde{b}_2, \bar{t}) \end{pmatrix} = \begin{pmatrix} \lambda_1 & 1-\lambda_1 \\ \lambda_2 & 1-\lambda_2 \end{pmatrix} \begin{pmatrix} \phi(b^*,\bar{t}) \\ \phi(b^{root}, \bar{t}) \end{pmatrix}.
    \end{equation*}
    
    $\mathfrak{f}$ being increasing and strictly convex implies that $\mathfrak{f}$ is strictly increasing.\footnote{Suppose there exist $x_1,x_2 \in [0,1]$ such that $x_1<x_2$ but $\mathfrak{f}(x_1)=\mathfrak{f}(x_2)$, then by strict convexity, $\mathfrak{f}(\frac{x_1+x_2}{2})<\frac{1}{2} \mathfrak{f}(x_1)+\frac{1}{2} \mathfrak{f}(x_2)=\mathfrak{f}(x_1)=\mathfrak{f}(x_2)$, a violation of $\mathfrak{f}$ being increasing.} Hence \begin{equation} \label{eq57} \mathfrak{f}(\lambda_1)>\mathfrak{f}(\lambda)>\mathfrak{f}(\lambda_2) \Rightarrow \phi(\tilde{b}_1,\underline{t})<\phi(\tilde{b},\underline{t})<\phi(\tilde{b}_2, \underline{t}). \end{equation} (\ref{eq55}) and (\ref{eq57}) together satisfy the preliminary condition for Lemma \ref{l2}. To apply Lemma \ref{l2}, consider the ratio of differences of $\phi$: \begin{equation*} \dfrac{\phi(\tilde{b}_1,\bar{t})-\phi(\tilde{b},\bar{t})}{\phi(\tilde{b},\bar{t})-\phi(\tilde{b}_2,\bar{t})}=\dfrac{\lambda_1-\lambda}{\lambda-\lambda_2}; \dfrac{\phi(\tilde{b}_1,\underline{t})-\phi(\tilde{b},\underline{t})}{\phi(\tilde{b},\underline{t})-\phi(\tilde{b}_2,\underline{t})}=\dfrac{\mathfrak{f}(\lambda_1)-\mathfrak{f}(\lambda)}{\mathfrak{f}(\lambda)-\mathfrak{f}(\lambda_2)}. \end{equation*} Let $\lambda=\mu \lambda_1+(1-\mu)\lambda_2, \mu \in (0,1)$, then by strict convexity of $\mathfrak{f}$,  \begin{equation*} \mathfrak{f}(\lambda)< \mu \mathfrak{f}(\lambda_1)+(1-\mu)\mathfrak{f}(\lambda_2) \Rightarrow \dfrac{\mathfrak{f}(\lambda_1)-\mathfrak{f}(\lambda)}{\mathfrak{f}(\lambda)-\mathfrak{f}(\lambda_2)}>\dfrac{1-\mu}{\mu}=\dfrac{\lambda_1-\lambda}{\lambda-\lambda_2}. \end{equation*} In that case, by Lemma \ref{l2}, no $\tilde{b}$ is very weakly dominated by any convex combination of any ($\tilde{b}_1$, $\tilde{b}_2$) that satisfy (\ref{eq55}) and (\ref{eq57}). Therefore, by Algorithm \ref{al1}, no $\tilde{b}$ can be eliminated, so that no $\tilde{b}$ is very weakly dominated, which further implies that the minimal optimal menu is the set of bundles that (weakly) contain $b^{root}$. 
\end{proof}

\subsection{The Branching Condition for Theorem \ref{t34}}
\label{app:t34}

The branching condition is the following:
There exist $\lambda_1,\lambda_2,\mu_1,\mu_2$ such that
\begin{equation} \label{eq3310} \begin{pmatrix} \phi(b^{root} \cup \{i\}, \bar{t}) \\ \phi(b^* \backslash \{i\}, \bar{t}) \end{pmatrix} = \begin{pmatrix} \lambda_1 & 1-\lambda_1 \\ \lambda_2 & 1-\lambda_2 \end{pmatrix} \begin{pmatrix} \phi(b^*,\bar{t}) \\ \phi(b^{root}, \bar{t}) \end{pmatrix}; \quad \begin{pmatrix} \phi(b^{root} \cup \{i\}, \underline{t}) \\ \phi(b^* \backslash \{i\}, \underline{t}) \end{pmatrix} = \begin{pmatrix} \mu_1 & 1-\mu_1 \\ \mu_2 & 1-\mu_2 \end{pmatrix} \begin{pmatrix} \phi(b^*,\underline{t}) \\ \phi(b^{root}, \underline{t}) \end{pmatrix}, \end{equation}
with
\begin{equation} \label{eq18} \text{for any } k \in \{1,2\},\ \mu_k, \lambda_k \in (0,1); \quad \lambda_1<\lambda_2; \quad \dfrac{\mu_1}{\lambda_1}<\dfrac{\mu_2-\mu_1}{\lambda_2-\lambda_1}<\dfrac{1-\mu_2}{1-\lambda_2}. \end{equation}
Moreover, for any $\tilde{b}$ such that $b^{root} \subsetneq \tilde{b} \subsetneq b^*$, $\tilde{b} \notin \{b^{root} \cup \{i\}, b^* \backslash \{i\}\}$, there exist $\lambda^{\tilde{b}}, \mu^{\tilde{b}}$ such that $\phi(\tilde{b},\bar{t})=\lambda^{\tilde{b}} \phi(b^*,\bar{t})+(1-\lambda^{\tilde{b}}) \phi(b^{root}, \bar{t})$, $\phi(\tilde{b},\underline{t})=\mu^{\tilde{b}} \phi(b^*,\underline{t})+(1-\mu^{\tilde{b}}) \phi(b^{root}, \underline{t})$, and
\begin{equation} \label{eq19} (\lambda^{\tilde{b}}, \mu^{\tilde{b}}) \in \{(x,y) \in \mathbb{R}^2 \mid y > \dfrac{\mu_1}{\lambda_1}x;\ y > \dfrac{1-\mu_2}{1-\lambda_2}x+\dfrac{\mu_2-\lambda_2}{1-\lambda_2},\ x\geq \lambda_1;\ y \leq \mu_2\}.\footnote{This set is the intersection of $\{(x,y) \mid \lambda_1\leq x \leq 1; \mu_1 \leq y \leq 1; y>\frac{\mu_1}{\lambda_1}x\}$ and $\{(x,y) \mid 0\leq x \leq \lambda_2; 0 \leq y \leq \mu_2; y>\frac{1-\mu_2}{1-\lambda_2}x+\frac{\mu_2-\lambda_2}{1-\lambda_2}\}$.} \end{equation}

Equations \eqref{eq3310} and \eqref{eq18} guarantee no dominance among the four corner bundles $b^{root}$, $b^{root}\cup\{i\}$, $b^*\setminus\{i\}$, $b^*$: $b^{root} \cup \{i\}$ is not very weakly dominated by any convex combination of $b^{root}$ and $b^* \backslash \{i\}$, and $b^* \backslash \{i\}$ is not very weakly dominated by any convex combination of $b^{root} \cup \{i\}$ and $b^*$; this also implies both are undominated by combinations of $b^{root}$ and $b^*$. Equation \eqref{eq19} constrains all other intermediate bundles. It exploits the least-favorite property of $i$: if $b^{root} \cup \{i\}$ is not very weakly dominated by a deterministic bundle, then it is very weakly dominated only if dominated by a convex combination of $b^{root}$ and another bundle; symmetrically, $b^* \backslash \{i\}$ only by a combination of $b^*$ and another bundle. Equation \eqref{eq19} rules out both. The efficacy of these restrictions runs through Lemma \ref{l2}: under \eqref{eq18} and \eqref{eq19}, inequality \eqref{eq66} cannot hold once $b^{root} \cup \{i\}$ or $b^* \backslash \{i\}$ is the middle bundle. 

\subsection{Proof of Theorem \ref{t34}}

\begin{proof}

The proof aims to show that both $b^{root} \cup \{i\}$ and $b^* \backslash \{i\}$ are not very weakly dominated. 

\textbf{Step 1: }I first show that (\ref{eq18}) guarantees that $b^{root} \cup \{i\}$ is not very weakly dominated by any convex combination of $b^{root}$ and $b^* \backslash \{i\}$, and $b^* \backslash \{i\}$ is not very weakly dominated by any convex combination of $b^{root} \cup \{i\}$ and $b^*$. By (\ref{eq18}), $\mu_2-\mu_1>(\lambda_2-\lambda_1) \dfrac{\mu_1}{\lambda_1}>0$, so that $0<\lambda_1<\lambda_2<1$ and $0<\mu_1<\mu_2<1$, which satisfies the preliminary conditions of Lemma \ref{l2}. In that case, $b^{root} \cup \{i\}$ is not very weakly dominated by any convex combination of $b^{root}$ and $b^* \backslash \{i\}$ if and only if
\begin{equation*}
    \dfrac{\phi(b^* \backslash \{i\}, \bar{t})-\phi(b^{root} \cup \{i\},\bar{t})}{\phi(b^{root} \cup \{i\},\bar{t})- \phi(b^{root}, \bar{t})} <     \dfrac{\phi(b^* \backslash \{i\}, \underline{t})-\phi(b^{root} \cup \{i\},\underline{t})}{\phi(b^{root} \cup \{i\},\underline{t})- \phi(b^{root}, \underline{t})}, 
\end{equation*}
which is equivalent to
\begin{equation*}
    \dfrac{\lambda_2-\lambda_1}{\lambda_1}<\dfrac{\mu_2-\mu_1}{\mu_1} \Leftrightarrow \dfrac{\mu_1}{\lambda_1}< \dfrac{\mu_2-\mu_1}{\lambda_2-\lambda_1}.
\end{equation*}
Similarly, $b^* \backslash \{i\}$ is not very weakly dominated by any convex combination of $b^{root} \cup \{i\}$ and $b^*$ if and only if $\dfrac{1-\lambda_2}{\lambda_2-\lambda_1}<\dfrac{1-\mu_2}{\mu_2-\mu_1} \Leftrightarrow \dfrac{\mu_2-\mu_1}{\lambda_2-\lambda_1}<\dfrac{1-\mu_2}{1-\lambda_2}$. 

It is not hard to see that (\ref{eq18}) also ensures that both $b^{root} \cup \{i\}$ and $b^* \backslash \{i\}$ are not very weakly dominated by any convex combination of $b^{root}$ and $b^*$, either by directly checking Lemma \ref{l2} like what has been argued above or following the adjacent pair of bundles check demonstrated in Algorithm \ref{al1}.  

\textbf{Step 2: }Next I show that (\ref{eq19}) guarantees that $b^{root} \cup \{i\}$ and $b^* \backslash \{i\}$ are not very weakly dominated by other bundles as well. Claim that if $\forall \tilde{b}$ such that $b^{root} \subsetneq \tilde{b} \subsetneq b^*, $  $\exists \lambda^{\tilde{b}}, \mu^{\tilde{b}}$ such that $\phi(\tilde{b},\bar{t})=\lambda^{\tilde{b}} \phi(b^*,\bar{t})+(1-\lambda^{\tilde{b}}) \phi(b^{root}, \bar{t}); \phi(\tilde{b},\underline{t})=\mu^{\tilde{b}} \phi(b^*,\underline{t})+(1-\mu^{\tilde{b}}) \phi(b^{root}, \underline{t})$, and \begin{equation} \label{eq62} \begin{split}(\lambda^{\tilde{b}}, \mu^{\tilde{b}}) \in &\{(x,y) \in \mathbb{R}^2| \lambda_1\leq x \leq 1; \mu_1 \leq y \leq 1; y>\dfrac{\mu_1}{\lambda_1}x\} \\ &\bigcap \{(x,y) \in \mathbb{R}^2| 0\leq x \leq \lambda_2; 0 \leq y \leq \mu_2; y>\dfrac{1-\mu_2}{1-\lambda_2}x+\dfrac{\mu_2-\lambda_2}{1-\lambda_2}\}, \end{split} \end{equation} then $b^{root} \cup \{i\}$ is very weakly dominated if and only if there exists $\tilde{b}$ such that $b^{root} \cup \{i\}$ is very weakly dominated by a convex combination of $b^{root}$ and $\tilde{b}$. This is because there does not exist $(\lambda^{\tilde{b}}, \mu^{\tilde{b}})$ such that $\lambda^{\tilde{b}}> \lambda_1, \mu^{\tilde{b}}<\mu_1$, i.e., no $\tilde{b}$ that very weakly dominates $b^{root} \cup \{i\}$. Proceed to show that no $\tilde{b}$ combining with $\emptyset$ can very weakly dominates $b^{root} \cup \{i\}$ given (\ref{eq62}) holds for any $(\lambda^{\tilde{b}}, \mu^{\tilde{b}})$. As $\lambda^{\tilde{b}}>\lambda_1>0; \mu^{\tilde{b}}>\mu_1>0$, by Lemma \ref{l2}, $b^{root} \cup \{i\}$ is very weakly dominated by a convex combination of $\tilde{b}$ and $b^{root}$ if and only if $\dfrac{\lambda^{\tilde{b}}-\lambda_1}{\lambda_1} \leq \dfrac{\mu^{\tilde{b}}-\mu_1}{\mu_1} \Leftrightarrow \dfrac{\lambda^{\tilde{b}}}{\mu^{\tilde{b}}} \leq \dfrac{\lambda_1}{\mu_1}$. However, by (\ref{eq62}), $\mu^{\tilde{b}}> \dfrac{\mu_1}{\lambda_1}\lambda^{\tilde{b}} \Rightarrow \dfrac{\lambda^{\tilde{b}}}{\mu^{\tilde{b}}}>\dfrac{\lambda_1}{\mu_1}$, so that $b^{root} \cup \{i\}$ is not very weakly dominated. Similarly, there does not $\tilde{b}$ that very weakly dominates $b^* \backslash \{i\}$, so $b^* \backslash \{i\}$ is very weakly dominated if and only if it is very weakly dominated by a convex combination of $b^*$ and $\tilde{b}$, which holds if and only if $\dfrac{1-\lambda_2}{\lambda_2-\lambda^{\tilde{b}}} \leq \dfrac{1-\mu_2}{\mu_2-\mu^{\tilde{b}}}$. Again, by $\mu^{\tilde{b}}>\dfrac{1-\mu_2}{1-\lambda_2}\lambda^{\tilde{b}}+\dfrac{\mu_2-\lambda_2}{1-\lambda_2}$, the previous inequality does not hold, so that $b^* \backslash \{i\}$ is not very weakly dominated as well.

\textbf{Step 3: }Further arguments can show that the intersection in (\ref{eq62}) is exactly the set demonstrated in (\ref{eq19}), and this set is non-empty: it is straightforward that the intersection in (\ref{eq62}) is a subset of the set in (\ref{eq19}), so it suffices to show the other direction of set inclusion. $\forall (x,y) \in \{(x,y) \in \mathbb{R}^2| y > \dfrac{\mu_1}{\lambda_1}x; y > \dfrac{1-\mu_2}{1-\lambda_2}x+\dfrac{\mu_2-\lambda_2}{1-\lambda_2}, x\geq \lambda_1; y \leq \mu_2\}, $
\begin{equation*} y>\dfrac{\mu_1}{\lambda_1} \cdot \lambda_1=\mu_1; x< \dfrac{1-\lambda_2}{1-\mu_2}(y-\dfrac{\mu_2-\lambda_2}{1-\lambda_2})=\lambda_2, \end{equation*} which fulfills the proof of ``(\ref{eq19}) $\subseteq$ (\ref{eq62})''. As for non-emptiness, since \begin{equation} \label{eq3362} \mu_2> \mu_1= \dfrac{\mu_1}{\lambda_1} \cdot \lambda_1; \mu_2(1-\lambda_2)>(1-\mu_2)\lambda_1+(\mu_2-\lambda_2), \end{equation} $(\lambda_1,\mu_2) \in \{(x,y) \in \mathbb{R}^2| y > \dfrac{\mu_1}{\lambda_1}x; y > \dfrac{1-\mu_2}{1-\lambda_2}x+\dfrac{\mu_2-\lambda_2}{1-\lambda_2}, x\geq \lambda_1; y \leq \mu_2\}$.\footnote{(\ref{eq3362}) holds because $\mu_2(1-\lambda_2)>(1-\mu_2)\lambda_1+(\mu_2-\lambda_2) \Leftrightarrow \mu_2(\lambda_1-\lambda_2)>\lambda_1-\lambda_2$, which holds as $\lambda_1<\lambda_2$ and $\mu_2<1$. }
\end{proof}

\section{Proofs for Other Applications}

\subsection{Proof of Theorem \ref{t4}}

\begin{proof}

I first show that every bundle that is not included in $\{\emptyset, b_{j_1}, \bigcup_{i=1}^2 \{b_{j_i}\}, ..., \bigcup_{i=1}^{n-1}\{b_{j_i}\}, b^*\}$ is very weakly dominated. In particular, $\bigcup_{i=1}^{n'} \{b_{\rho(j_i)}\}$ is very weakly dominated for all $n' \in \{1,...,n\}$ if $\{\rho(j_1),...,\rho(j_{n'})\}$ is a size $n'$ subset of $\{1,...,n\}$ but $\{\rho(j_1),...,\rho(j_{n'})\} \neq \{j_1,...,j_{n'}\}$. Consider such a $\bigcup_{i=1}^{n'} \{b_{\rho(j_i)}\}$, let \begin{equation*} i =\min \{\tilde{i} \in \{1,...,n'\}| j_{\tilde{i}} \notin \{\rho(j_1),...,\rho(j_{n'})\} . \end{equation*} Meanwhile, there exists $i' \geq n'+1>i$ such that $j_{i'} \in \{\rho(j_1),...,\rho(j_{n'})\}$. I proceed to show that $\bigcup_{i=1}^{n'} \{b_{\rho(j_i)}\}$ is very weakly dominated by a convex combination of $ \bigcup_{i=1}^{n'} \{b_{\rho(j_i)}\} \backslash \{b_{j_{i'}}\}$ and $\bigcup_{i=1}^{n'} \{b_{\rho(j_i)}\} \bigcup \{b_{j_i}\}$. By additivity,   \begin{equation*}  \begin{split} \phi(\bigcup_{i=1}^{n'} \{b_{\rho(j_i)}\}\backslash \{b_{j_{i'}} \},\bar{t}) \leq \phi(\bigcup_{i=1}^{n'} \{b_{\rho(j_i)}\},\bar{t})  \leq \phi(\bigcup_{i=1}^{n'} \{b_{\rho(j_i)}\} \bigcup \{b_{j_i}\},\bar{t}), \\
\phi(\bigcup_{i=1}^{n'} \{b_{\rho(j_i)}\}\backslash \{b_{j_{i'}} \},\underline{t}) \geq \phi(\bigcup_{i=1}^{n'} \{b_{\rho(j_i)}\},\underline{t})  \geq \phi(\bigcup_{i=1}^{n'} \{b_{\rho(j_i)}\} \bigcup \{b_{j_i}\},\underline{t}), 
\end{split}
\end{equation*}

which satisfies the preliminary conditions of Lemma \ref{l2}. Hence, by Lemma \ref{l2}, showing the very weakly dominance of $\bigcup_{i=1}^{n'} \{b_{\rho(j_i)}\}$ is equivalent to showing \begin{equation*} \dfrac{\bigcup_{i=1}^{n'} \{b_{\rho(j_i)}\} \bigcup \{b_{j_i}\},\bar{t})-\phi(\bigcup_{i=1}^{n'} \{b_{\rho(j_i)}\},\bar{t})}{\phi(\bigcup_{i=1}^{n'} \{b_{\rho(j_i)}\},\bar{t})-\phi(\bigcup_{i=1}^{n'} \{b_{\rho(j_i)}\} \backslash \{b_{j_{i'}}\},\bar{t})} \geq \dfrac{\phi(\bigcup_{i=1}^{n'} \{b_{\rho(j_i)}\} \bigcup \{b_{j_i}\},\underline{t})-\phi(\bigcup_{i=1}^{n'} \{b_{\rho(j_i)}\},\underline{t})}{\phi(\bigcup_{i=1}^{n'} \{b_{\rho(j_i)}\},\underline{t})-\phi(\bigcup_{i=1}^{n'} \{b_{\rho(j_i)}\} \backslash \{b_{j_{i'}}\},\underline{t})}, \end{equation*}
which is further equivalent to showing
\begin{equation*} \dfrac{\phi(b_{j_i},\bar{t})}{\phi(b_{j_{i'}},\bar{t})} \geq \dfrac{\phi(b_{j_i},\underline{t})}{\phi(b_{j_{i'}},\underline{t})}. \end{equation*}
By the strict order of the ratio and $i'>i$, $\dfrac{\phi(b_{j_i},\bar{t})}{\phi(b_{j_i},\underline{t})} < \dfrac{\phi(b_{j_{i'}},\bar{t})}{\phi(b_{j_{i'}},\underline{t})}$. Since $\phi(b_{j_i},\underline{t})<0, \phi(b_{j_{i'}},\bar{t})\geq 0$, this directly implies that $ \dfrac{\phi(b_{j_i},\bar{t})}{\phi(b_{j_{i'}},\bar{t})} \geq \dfrac{\phi(b_{j_i},\underline{t})}{\phi(b_{j_{i'}},\underline{t})}$. Thus, any bundle that is not included in the proposed menu is very weakly dominated. 

By Proposition \ref{p4}, the only thing left to show is every bundle that is included in the proposed menu is not very weakly dominated. To show this, by Proposition \ref{p4}, it suffices to show that every bundle in the menu is not very weakly dominated by any convex combinations of other bundles in the menu. As \begin{equation*} \begin{split} \phi(\emptyset, \bar{t}) \leq \phi(b_{j_1},\bar{t}) \leq ... \leq \phi(\bigcup_{i=1}^{n-1}\{b_{j_i}\},\bar{t}) \leq \phi(b^*,\bar{t}); \phi(\emptyset, \underline{t}) \geq \phi(b_{j_1},\underline{t}) \geq ... \geq \phi(\bigcup_{i=1}^{n-1}\{b_{j_i}\},\underline{t}) \geq \phi(b^*,\bar{t}), \end{split} \end{equation*}
by Algorithm \ref{al1}, it suffices to show that $\bigcup_{i=1}^{l}\{b_{j_i}\}$ is not very weakly dominated by any convex combinations of $\bigcup_{i=1}^{l-1}\{b_{j_i}\}$ and $\bigcup_{i=1}^{l+1}\{b_{j_i}\}$, for all $l \in \{1,...,n-1\}$.\footnote{Let $\bigcup_{i=1}^{0}\{b_{i}\}$ be $\emptyset$.} By Lemma \ref{l2}, it suffices to show that \begin{equation} \label{eq69} \dfrac{\phi(b_{j_{l+1}},\bar{t})}{\phi(b_{j_l},\bar{t})} < \dfrac{\phi(b_{j_{l+1}},\underline{t})}{\phi(b_{j_l},\underline{t})}, \forall l \in \{1,...,n-1\} \end{equation}
By the strict order of the ratio and $\phi(b_{j_l},\bar{t})>0, \phi(b_{j_{l+1}},\underline{t})<0$, (\ref{eq69}) holds.\footnote{Alternatively, one can also directly show that each bundle in the menu cannot be very weakly dominated by any convex combinations of other bundles. It can be regarded as a result of the following: if $a_i \geq 0, b_i<0, \forall i \in \{1,...,n\}$ and $\frac{a_1}{b_1}<...<\frac{a_n}{b_n}$, then $\frac{a_1+...+a_{n'}}{b_1+...+b_{n'}} < \frac{a_{n'+1}+...+a_n}{b_{n'+1}+...+b_n}, \forall n' \in \{1,...,n-1\}$. } As a result, by Theorem \ref{t2}, $\{\emptyset, b_{j_1}, \bigcup_{i=1}^2 \{b_{j_i}\}, ..., \bigcup_{i=1}^{n-1}\{b_{j_i}\}, b^*\}$ is the minimal optimal menu for problem \eqref{eq5}. 

\end{proof}

\subsection{Proof of Proposition \ref{lemma:quasiconcave}}

     \begin{proof}
        ``if'': If there does not exist $t<t'$ such that $\sign[z'(t)]<0, \sign[z'(t')]>0$, and $z'$ is not identically zero on any interval, then one of the following cases has to hold:
        \begin{itemize}
            \item for any $t, z'(t) \geq 0$, with the equality only can hold discretely,
            \item for any $t, z'(t) \leq 0$, with the equality only can hold discretely,
            \item there exists $t_0$ such that $z'(t)\geq0, t \in [\underline{t}, t_0); z'(t)=0, t=t_0; z'(t)\leq0, t \in (t_0, \bar{t}]$, with all equalities only can hold discretely.
        \end{itemize}
        The first two cases imply that $z$ is strictly monotonic, which is strictly quasi-concave. The third case implies that $z$ is single-peaked, with the peak being a unique point, and $z$ is strictly increasing (decreasing) on the left(right) side of the peak, so that $z$ is strictly quasi-concave. 

        ``only if'': If $z$ is strictly quasi-concave but there exists $t<t'$ such that $\sign[z'(t)]<0, \sign[z(t')]>0$, then by continuity of $z'$, there exists $t_* \in (t,t')$ such that $z'(t_*)=0$ and there exists $\delta>0$ such that $z'(t_*-\delta)<0, z'(t_*+\delta)>0$. In that case, $z(t_*)<\min \{z(t_*-\delta), z(t_*+\delta)\}$. However, $t_*=\dfrac{1}{2}(t_*-\delta)+\dfrac{1}{2}(t_*+\delta)$, a contradiction to quasi-concavity. Meanwhile, if $f$ is strictly quasi-concave but the zeros of $z'$ is not discrete, then there exists an interval $I \subseteq [\underline{t}, \bar{t}]$ such that for any $t \in I, z'(t)=0$, indicating that for any $t,t' \in I, z(t)=z(t')$, a contradiction to strict quasi-concavity. 
    \end{proof}

\subsection{Proof of Theorems \ref{t5} and \ref{t6}}

\begin{proof}
    \textbf{Proof of Theorem \ref{t5}: }For any $b$ such that $\phi(b,\underline{t}) \leq \phi(b^*,\underline{t})$, since $\{b^*\} = \argmax \limits_{\tilde{b} \in \mathcal{B}} v(\tilde{b},\bar{t})$, i.e., $v(b,\bar{t})=\phi(b,\bar{t}) < \phi(b^*,\bar{t})$, by Corollary \ref{c2}, $b$ is very weakly dominated by $b^*$, so that $b$ is excluded from the minimal optimal menu. 
    
    For any $b$ such that $\phi(b,\underline{t}) > \phi(b^*, \underline{t})$, it has to be the case that $\phi(\emptyset,\bar{t}) \leq \phi(b,\underline{t}) \leq \phi(b^*,\underline{t}); \phi(\emptyset, \underline{t}) \geq \phi(b,\underline{t}) \geq \phi(b^*,\underline{t})$. Following the notation in Lemma \ref{l5}, $Q(b)=1-F(t_{\emptyset b}); Q(b^*)=1-F(t_{\emptyset b^*})$. Suppose $t_{bb^*} > t_{\emptyset b}$, for any $t<t_{bb^*}, \phi(b^*,t)<\phi(b,t)$ implies that $\phi(b^*,t_{\emptyset b})<\phi(b, t_{\emptyset b})=0$. Since $\phi(b^*,t) <0$ if and only if $t< t_{\emptyset b^*}$, $t_{\emptyset b}<t_{\emptyset b^*}$, which further implies that $Q(b)>Q(b^*)$, a contradiction. Therefore, $t_{bb^*} \leq t_{\emptyset b}$. By Lemma \ref{l5}, $b$ is very weakly dominated by a convex combination of $\emptyset$ and $b^*$. Hence, every $b \notin \{\emptyset, b^*\}$ is very weakly dominated, and it is easy to check that $\emptyset$ and $b^*$ are not very weakly dominated, as $\{\emptyset\}=\argmax \limits_{b \in \mathcal{B}} \phi(b,\underline{t}); \{b^*\}=\argmax \limits_{b \in \mathcal{B}} \phi(b, \bar{t})$. By Theorem \ref{t2}, $\{\emptyset,b^*\}$ is the minimal optimal menu. 

    \end{proof}

  \begin{proof}   \textbf{Proof of Theorem \ref{t6}: } Suppose that there exist $b_1$ and $b_2$ in the minimal optimal menu such that neither $b_1 \subseteq b_2$ nor $b_2 \subseteq b_1$. Without loss of generality, let $Q(b_1) \leq Q(b_2)$. If $\phi(b_1, \underline{t}) \leq \phi(b_1 \cup b_2,\underline{t})$, since $\phi(b_1, \bar{t}) \leq \phi(b_1 \cup b_2, \bar{t})$,  $b_1$ is very weakly dominated by $b_1 \cup b_2$, a contradiction to $b_1$ belonging to the minimal optimal menu. If $\phi(b_1,\underline{t}) > \phi(b_1 \cup b_2,\underline{t})$, then it has to be the case that $\phi(\emptyset,\bar{t}) \leq \phi(b_1,\bar{t}) \leq \phi(b_1 \cup b_2,\bar{t}; \phi(\emptyset,\underline{t}) \geq \phi(b_1,\underline{t}) \geq \phi(b_1 \cup b_2,\underline{t})$. As $Q(b_1 \cup b_2) \geq Q(b_1)$, by the proof of Theorem \ref{t5}, $b_1$ is very weakly dominated by a convex combination of $\emptyset$ and $b_1 \cup b_2$, again a contradiction. Therefore, the minimal optimal menu is nested. 
 
 Next, consider the menu $\{\emptyset, b_1^\star, b_1^\star \cup b_2^\star,...,b^*\}$. By Proposition \ref{p4}, to show that it is an optimal menu, it suffices to show that every bundle not included in it is very weakly dominated by convex combinations of bundles included in the menu. By the proof of Proposition 9 in \cite{yang2023nested}, for any $b \notin \{\emptyset, b_1^\star, b_1^\star \cup b_2^\star,..., b^*\}$, there exists $j \in \{1,2,...,2^n\}$ such that $b=b_j^\star$, and $Q(b_j^\star) \leq Q(b_1^\star \cup b_2^\star \cup ... \cup b_j^\star)$. By the argument above, $b_j^\star$ is either very weakly dominated by $b_1^\star \cup b_2^\star \cup ... \cup b_j^\star$, or is very weakly dominated by a convex combination of $\emptyset$ and $b_1^\star \cup ... \cup b_j^\star$. The result thus follows. 

\end{proof}

\newpage

\pagenumbering{arabic}
\setcounter{page}{1}
\begin{center}
    {\LARGE \textbf{Online Appendix}}\\[1em]

\end{center}
\addcontentsline{toc}{section}{Online Appendix}

\section{Examples: Failure of Assumption \ref{a0} or \ref{a1}}
\label{s622}

I first show an example where Lemma \ref{l1} does not hold if only Assumption \ref{a1} but not both Assumptions \ref{a0} and \ref{a1} are satisfied:
\begin{example}
\label{e7}
 Consider the following example: $T=[0,2], \mathcal{B}=\{b,b'\}, v(b,t)=-t^2+2t+1, v(b',t)=0, F(t)=\dfrac{t}{2}$. This gives $\phi(b,t)=-t^2+2t+1-(2-t)(-2t+2)=-3t^2+8t-3, \phi(b',t)=0$. In that case, \begin{equation*} \phi(b,0)=-3<\phi(b',0); \phi(b,2)=1>\phi(b',1), \end{equation*} so that the condition specified in Lemma \ref{l1} is satisfied. Meanwhile, for any $a=\lambda_a b+(1-\lambda_a)b'$  and  $a'=\lambda_{a'} b+(1-\lambda_{a'})b'$, \begin{equation*} \begin{split} \phi(a,t)-\phi(a',t)&=(\lambda_a-\lambda_{a'})\phi(b,t) \\&=(\lambda_a-\lambda_{a'})(-3t^2+8t-3)\\&=-3(\lambda_a-\lambda_{a'})(t-\dfrac{4-\sqrt{7}}{3})(t-\dfrac{4+\sqrt{7}}{3}), \end{split} \end{equation*} which is indeed single-crossing in $t$ on $[0,2]$ for any $\lambda_a,\lambda_{a'} \in [0,1]$, satisfying Assumption \ref{a1}. 

 However, $\phi(b,t)-\phi(b',t)=-3t^2+8t-3$, which is not monotonic on $[0,2]$, a violation of Assumption \ref{a0}. Correspondingly, $v(b,t)-v(b',t)=-t^2+2t+1$ is indeed not increasing in $t$ on $[0,2]$. 
\end{example}

I then demonstrate several examples to show the role that Assumption \ref{a1} plays in obtaining the results discussed in Section \ref{s22}. Example \ref{e2} below demonstrates that a never strict best response bundle may not be very weakly dominated if Assumption \ref{a1} fails:

\begin{example}
    \label{e2}
    Consider the following example: $T=[0,1],\mathcal{B}=\{b,b',b''\}, \phi(b,t)=t, \phi(b',t)=1-t, \phi(b'',t)=  -\dfrac{1}{3}t+\dfrac{2}{3} \text{ if } t \leq \dfrac{1}{2}; =\dfrac{1}{3}t+\dfrac{1}{3} \text{ if } t > \dfrac{1}{2}. $

\begin{figure}[H]
    \centering
    \includegraphics[width=0.45\linewidth]{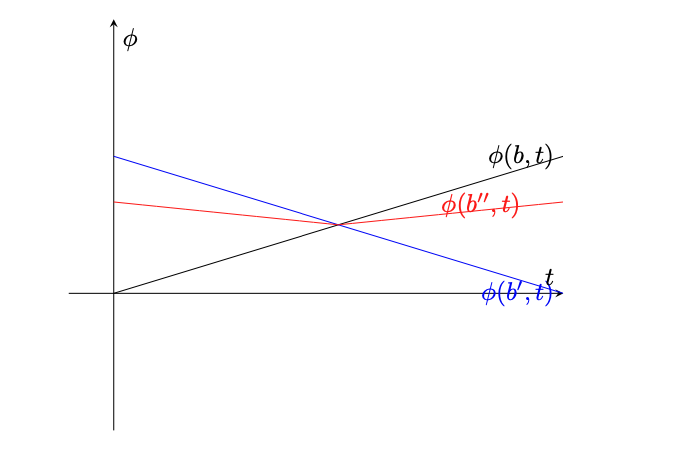}
    \caption{Failure of Proposition \ref{p1}}
    \label{fig:placeholder3}
\end{figure}
 
    $b''$ is a never strict best response, but $b''$ is not very weakly dominated: suppose $b''$ is very weakly dominated, then there exists $\lambda \in [0,1]$ such that $\phi(b'',t) \leq \lambda \phi(b,t)+(1-\lambda)\phi(b',t), \forall t \in T$. $\phi(b'',0) \leq \lambda \phi(b,0)+(1-\lambda)\phi(b',0)$ implies that $\lambda \leq \dfrac{1}{3}$; $\phi(b'',1) \leq \lambda \phi(b,1)+(1-\lambda)\phi(b',1)$ implies that $\lambda \geq \dfrac{2}{3}$, a contradiction to the existence of $\lambda$, which means $b''$ is not very weakly dominated. 

However, in this example, Assumption \ref{a1} does not hold. In particular, $\phi(b'',t)-\phi(\dfrac{b+b'}{2},t)$ is not single-crossing: $\sign[\phi(b'',0)-\phi(\dfrac{b+b'}{2},0)]=1, \sign[\phi(b'',\dfrac{1}{2})-\phi(\dfrac{b+b'}{2},\dfrac{1}{2})]=0,\sign[\phi(b'',1)-\phi(\dfrac{b+b'}{2},1)]=1 $. 
\end{example}

The following example shows that the violation of Assumption \ref{a1} leads to the failure of Proposition \ref{p4}. 

\begin{example}
\label{e5}

Consider the following example: $T=[0,1], \mathcal{B}=\{b,b',b'',b'''\}, \phi(b,t)=t, \phi(b',t)=1-t,$ \[ \phi(b'',t)= \begin{cases} -\dfrac{1}{3}t+\dfrac{2}{3} & \text{if } t \leq \dfrac{1}{2}, \\ \dfrac{1}{3}t+\dfrac{1}{3} & \text{if } t > \dfrac{1}{2}. \end{cases} ,  \phi(b''',t)= \begin{cases} -\dfrac{3}{4}t+\dfrac{7}{8} & \text{if } t \leq \dfrac{1}{2}, \\ \dfrac{3}{4}t+\dfrac{1}{8} & \text{if } t > \dfrac{1}{2}. \end{cases} \] 
\begin{figure}[H]
    \centering
    \includegraphics[width=0.45\linewidth]{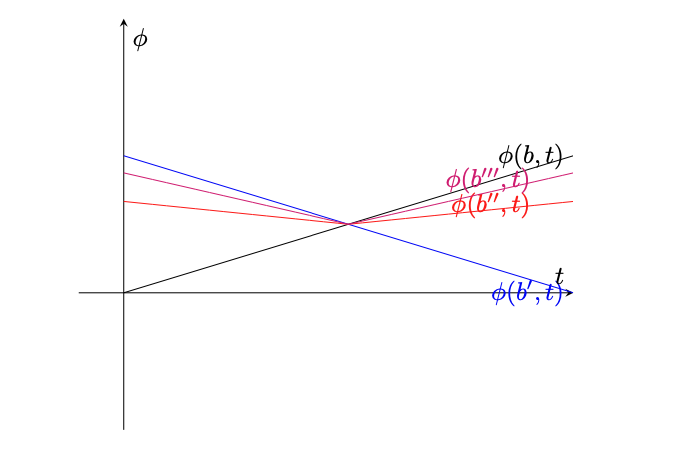}
    \caption{Failure of Proposition \ref{p4}}
    \label{fig:placeholder2}
\end{figure}
Consider the menu $\{b,b',b'''\}$. I first claim that every very weakly dominated bundle is very weakly dominated by a convex combination of $b,b',b'''$. This is because for any very weakly dominate bundle $\underline{a}$, there exist $\lambda_1,\lambda_2,\lambda_3 \in [0,1], \lambda_1+\lambda_2+\lambda_3 \leq 1$ such that $\phi(\underline{a},t) \leq \lambda_1 \phi(b,t)+\lambda_2 \phi(b',t)+\lambda_3 \phi(b'',t)+(1-\lambda_1-\lambda_2-\lambda_3)\phi(b''',t), \forall t$. Since $\phi(b'',t) \leq \phi(b''',t), \forall t$, this directly gives $\phi(\underline{a},t) \leq \lambda_1 \phi(b,t)+\lambda_2 \phi(b',t)+(1-\lambda_1-\lambda_2)\phi(b''',t)$. I next claim that $b,b'$ and $b'''$ are all not very weakly dominated. As $\{b\}=\argmax \limits_{\tilde{b} \in \mathcal{B}} \phi(\tilde{b},1), \{b'\}=\argmax \limits_{\tilde{b} \in \mathcal{B}} \phi(\tilde{b},0)$, $b$ and $b'$ are straightforwardly not very weakly dominated. Since $b''$ is very weakly dominated by $b'''$, if $b'''$ is very weakly dominated, then $b'''$ can only be very weakly dominated by convex combinations of $b$ and $b'$, which is impossible, for $b''$ is not very weakly dominated by any convex combination of $b$ and $b'$. Therefore, $b'''$ is not very weakly dominated. 

However, $\{b,b',b'''\}$ is not a minimal optimal menu, because $\{b,b'\} \subsetneq \{b,b',b'''\}$ is also an optimal menu, indicating that Proposition \ref{p4} does not hold without Assumption \ref{a1}.
    
\end{example}

Example \ref{e6} below argues that the uniqueness of the minimal optimal menu does not hold without Assumption \ref{a1}. It also shows that the existence of strict best response is not guaranteed without Assumption \ref{a1}. 

\begin{example} \label{e6} 
    Consider the following: $T=[0,1],\mathcal{B}=\{b,b',b''\}, \phi(b,t)=t, \phi(b',t)=1-t, \phi(b'',t)=  1-t  \text{ if } t \leq \dfrac{1}{2}; =t \text{ if } t > \dfrac{1}{2}. $
    \begin{figure}[H]
        \centering
        \includegraphics[width=0.45\linewidth]{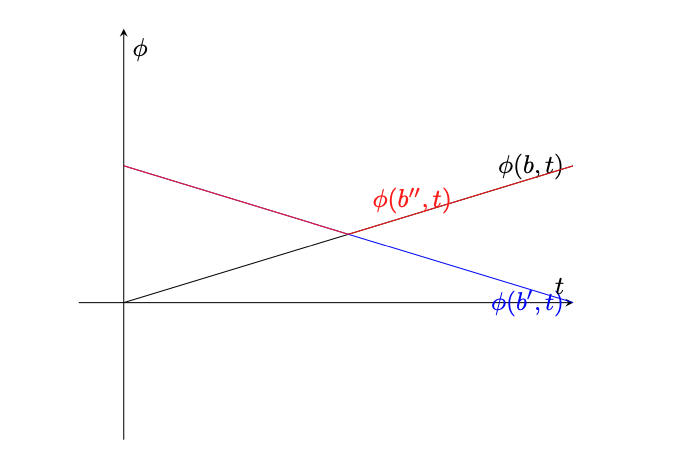}
        \caption{Multiple minimal optimal menus}
        \label{fig:placeholder1}
    \end{figure}
In this example, both $\{b,b'\}$ and $\{b''\}$ are minimal optimal menus, but this example does not satisfy Assumption \ref{a1}: $\phi(b'',t)$ and $\phi(\dfrac{b+b'}{2},t)$ are not single-crossing.\footnote{One can get $\phi(\frac{b+b'}{2},t)$ from $\frac{1}{2}\phi(b,t)+\frac{1}{2}\phi(b',t)$. } 

 As for the nonexistence of strict best response, note that for $t \in [0,\dfrac{1}{2}], \phi(b,t)<\phi(b',t)=\phi(b'',t)$, and for $t \in (\dfrac{1}{2},1], \phi(b',t)<\phi(b,t)=\phi(b'',t)$. Hence, there exist two best responses for any $t$, so that strict best bundles do not exist. 
    \end{example}

\section{Existence of Strict Best Response Bundles}
\label{sec:existence}

Under Assumption \ref{a1}, a strict best response bundle always exists. Formally, 
\begin{corollary}
\label{c1}
    Suppose Assumption \ref{a1} holds. There exists $b \in \mathcal{B}$ that is a strict best response: for some $t \in T$, $\phi(b,t) > \max_{b' \neq b} \phi(b',t)$.
\end{corollary}
Note that Corollary \ref{c1} cannot be interpreted as a corollary of the following result (\cite{bernheim1984}, \cite{pearce1984}, \cite{fudenberg1991}): the set of rationalizable strategies is nonempty and contains at least one pure strategy for each player. The main difference is that I focus on strict best response'' while the previous result is a characterization of ``best response''. In fact, as what has been shown in Example \ref{e6} in Online Appendix \ref{s622}, without Assumption \ref{a1}, the existence can fail.

\begin{proof}
    Suppose that there does not exist $b \in \mathcal{B}$ such that $b$ is a strict best response. By Proposition \ref{p1}, this implies that every $b \in \mathcal{B}$ is very weakly dominated. I proceed to show this is impossible by the following lemma.
\begin{lemma}
\label{l4}
    For any $n \in \mathbb{N}$ and distinctive $x_1,...,x_n \in \mathbb{R}$, there does not exist $\Lambda=(\lambda_{ij})_{n \times n}$ such that the followings hold:
    \begin{enumerate}
        \item $\forall i,j, \lambda_{ij} \in [0,1], \lambda_{ii} \neq 1, \sum \limits_{j=1}^n \lambda_{ij}=1,$
        \item $(x_1,...,x_n)' \leq \Lambda (x_1,...,x_n)'.$\footnote{A component-wise inequality.}
    \end{enumerate}
\end{lemma}
\begin{proof}
    When $n=2$, suppose that the lemma does not hold, i.e., there exist $x_1 \neq x_2$ and $\Lambda$ such that 
    \begin{equation*} \begin{split} x_1 \leq \lambda_{11} x_1 +(1-\lambda_{11}) x_2;  
     x_2 \leq (1-\lambda_{22})x_1+\lambda_{22}x_2. \end{split} \end{equation*}
    Since $\lambda_{11},\lambda_{22} \neq 1$, this directly gives $x_1=x_2$, a contradiction. 

    Assume that the lemma holds when $n=k-1$, consider the case where $n=k$. Suppose the lemma does not hold for $n=k$, then $x_i \leq \sum \limits_{j=1}^k \lambda_{ij}x_j, \forall i \in \{1,...,k\}$. In particular, $x_1 \leq \dfrac{\sum \limits_{j=2}^k \lambda_{1j}x_j}{1-\lambda_{11}}$. Plugging into the rest $k-1$ inequalities, \begin{equation*} x_i \leq \sum \limits_{j=2}^k (\lambda_{ij}+\lambda_{i1} \dfrac{\lambda_{1j}}{1-\lambda_{11}})x_j, \forall i \in \{2,...,k\}. \end{equation*}
    $\forall i \in \{2,...,k\}, \sum \limits_{j=2}^k \lambda_{ij}+\lambda_{i1} \dfrac{\lambda_{1j}}{1-\lambda_{11}}=\sum \limits_{j=2}^k \lambda_{ij}+\lambda_{i1} \dfrac{\sum \limits_{j=2}^k \lambda_{1j}}{1-\lambda_{11}}=\sum \limits_{j=1}^k \lambda_{ij}=1$. In addition, $\forall i,j \in \{2,...,k\}, \lambda_{ij}+\lambda_{i1} \dfrac{\lambda_{1j}}{1-\lambda_{11}} \geq 0$ and $\lambda_{ij}+\lambda_{i1} \dfrac{\lambda_{1j}}{1-\lambda_{11}} \leq \lambda_{ij}+\lambda_{i1} \leq 1$. Specifically, $\lambda_{ii}+\lambda_{i1} \leq 1$. Suppose the equality holds, then it holds only when $\lambda_{11}+\lambda_{1i}=\lambda_{ii}+\lambda_{i1}=1$, which leads to \begin{equation*} \begin{split} x_1 \leq \lambda_{11} x_1 +(1-\lambda_{11}) x_i;
     x_i \leq (1-\lambda_{ii})x_1+\lambda_{ii}x_i. \end{split} \end{equation*}
    By the discussion of the $n=2$ case, this implies that $x_1=x_i$, a contradiction to the distinctiveness. Hence, $\forall i \in \{2,...,k\}, \lambda_{ii}+\lambda_{i1}<1$. In that case, $x_2,...,x_n$ and $\Lambda'$, in which $\lambda'_{ij}=\lambda_{ij}+\lambda_{i1}\dfrac{\lambda_{1j}}{1-\lambda_{11}}, \forall i,j \in \{2,...,k\}$ satisfy the conditions in the lemma, a contradiction to the induction assumption that the lemma holds for $n=k-1$. Thus, the lemma holds for $n=k$. By induction, the lemma holds for any $n$.
\end{proof}
Let $\mathcal{B}:=\{b_1,...,b_{2^n}\}$. If all $b \in \mathcal{B}$ are very weakly dominated, then there exists $\Lambda=(\lambda_{ij})_{2^n \times 2^n}$ such that the followings hold: 
\begin{enumerate}
    \item $\forall i,j, \lambda_{ij} \in [0,1], \lambda_{ii} \neq 1, \sum \limits_{j=1}^{2^n} \lambda_{ij}=1$
    \item $\forall i \in \{1,...,2^n\}, t \in T, \phi(b_i,t) \leq \sum \limits_{j=1}^{2^n} \lambda_{ij} \phi(b_j,t)$
\end{enumerate}
By Lemma \ref{l4}, the existence of $\Lambda$ implies that $\exists i,j \in \{1,...,2^n\}, i \neq j, \phi(b_i,t)=\phi(b_j,t), \forall t \in T$, which is the case that is excluded by the preliminary assumption that no two different bundles have the same virtual value functions. As a result, it cannot be the case that all $b \in \mathcal{B}$ are very weakly dominated. 
\end{proof}

\section{Proofs and Discussions for Section \ref{s23}}

\subsection{Connections between Assumptions \ref{a1} and \ref{a2}}
\label{sec:a1a2}

\begin{proposition}
\label{p133}
    If $\phi$ satisfies 
    \begin{enumerate}
        \item Assumption \ref{a1}, and
        \item there exist $a_0,a_0' \in \Delta(\mathcal{B})$ and $c \in \mathbb{R} \backslash \{0\}$ such that for any $t$, $\phi(a_0,t)-\phi(a_0',t)=c$,
    \end{enumerate} 
    then $\phi$ satisfies Assumption \ref{a2}.
    
\end{proposition}


\begin{proof}
Suppose that there exist $a,a' \in \Delta(\mathcal{B})$ such that $\phi(a,t)-\phi(a',t)$ is not monotonic in $t$, then, without loss of generality, there exist $\underline{t} \leq t_1<t_2<t_3 \leq \bar{t}$ such that 
\begin{equation*} \begin{split} \phi(a,t_1)-\phi(a',t_1)>\phi(a,t_2)-\phi(a',t_2); \phi(a,t_3)-\phi(a',t_3)>\phi(a,t_2)-\phi(a',t_2). \end{split} \end{equation*}
Consider $a^1,a^2 \in \Delta(\mathcal{B})$ such that \[a^1=\lambda_1 a+\lambda_2 a'+ \lambda_3 a_0+ \lambda_4 a_0'; a^2= (\lambda_1+\varepsilon_1)a+(\lambda_2-\varepsilon_1)a'+(\lambda_3+\varepsilon_2)a_0+(\lambda_4-\varepsilon_2)a_0',\] in which $0 \leq \lambda_1,\lambda_1+\varepsilon_1, \lambda_2, \lambda_2-\varepsilon_1, \lambda_3, \lambda_3+\varepsilon_2, \lambda_4, \lambda_4 -\varepsilon_2 \leq 1$, $\varepsilon_1 \neq 0$,  and $\lambda_1+\lambda_2+\lambda_3+\lambda_4=1$. It is straightforward that for any $t$,
\begin{equation*} \begin{split} \phi(a^2,t)-\phi(a^1,t) &=\varepsilon_1[\phi(a,t)-\phi(a',t)]+\varepsilon_2[\phi(a_0,t)-\phi(a_0',t)] \\ &= \varepsilon_1[\phi(a,t)-\phi(a',t)+\dfrac{\varepsilon_2}{\varepsilon_1}(\phi(a_0,t)-\phi(a_0',t))]. \end{split} \end{equation*}

As $\dfrac{\varepsilon_2}{\varepsilon_1}$ can be anything in $(-\infty,+\infty)$ and $\phi(a_0,t)-\phi(a_0',t) \equiv c \neq 0$, there exist $\varepsilon_1,\varepsilon_2$ such that 
\begin{equation*} \max \{ \phi(a',t_1)-\phi(a,t_1), \phi(a',t_3)-\phi(a,t_3)\} < \dfrac{\varepsilon_2}{\varepsilon_1}c<\phi(a',t_2)-\phi(a,t_2). \end{equation*}
In that case, $\phi(a^2,t_1)-\phi(a^1,t_1)>0, \phi(a^2,t_2)-\phi(a^1,t_2)<0, \phi(a^2,t_3)-\phi(a^1,t_3)>0$, a contradiction to Assumption \ref{a1}: $\phi(a^2, \cdot)-\phi(a^1, \cdot)$ is not single-crossing. The proposition thus follows. 
\end{proof}

\subsection{Solving the ODE}
\label{d2}

Recall (\ref{eq12}), for any $b \in \mathcal{B}$,
$$ v(b,t)-\dfrac{1-F(t)}{f(t)}v_t(b,t)=g_1(b)h_1(t)+g_2(b)+h_2(t).$$
Given $b$ and $F$, this can be regarded as a linear ODE of order 1, for which the solutions are
\begin{equation*} v(b,t)=c(b)e^{\int \frac{f(t)}{1-F(t)}dt}-e^{\int \frac{f(t)}{1-F(t)}dt}\displaystyle{\int}\dfrac{f(t)}{1-F(t)}[g_1(b)h_1(t)+g_2(b)+h_2(t)]e^{-\int \frac{f(t)}{1-F(t)}dt} dt, \end{equation*}
for $t \in [\underline{t},\bar{t})$ and $v(b,\bar{t})=g_1(b)h_1(\bar{t})+g_2(b)+h_2(\bar{t})$, in which $c: \mathcal{B} \rightarrow \mathbb{R}$. 
The solutions for $t \in [\underline{t},\bar{t})$ can be further simplified into
\begin{equation*} v(b,t)=\dfrac{c(b)-\displaystyle{\int}f(t)[g_1(b)h_1(t)+g_2(b)+h_2(t)]dt}{1-F(t)}. \end{equation*}

Moreover, since $\phi(a,t)=\mathbb{E}_{b \sim a }[\phi(b,t)]$, it is not hard to see that $g_i(a)=\mathbb{E}_{b \sim a}[g_i(b)], \forall i \in \{1,2\}$.\footnote{$\phi(a,t)=g_1(a)h_1(t)+g_2(a)+h_2(t)=\mathbb{E}_{b \sim a}[g_1(b)]h_1(t)+\mathbb{E}_{b \sim a}[g_2(b)]+h_2(t)$. If there exists $t,t'$ such that $h_1(t) \neq h_1(t')$, then $g_i(a)=\mathbb{E}_{b \sim a}[g_i(b)], \forall i \in \{1,2\}$. } Hence, for $a \in \Delta(\mathcal{B}), v(a,t)=\sum \limits_{b \in \mathcal{B}} a_b v(b,t) \\ = \dfrac{\sum \limits_{b \in \mathcal{B}}a_b c(b) -\displaystyle{\int} f(t)[\sum \limits_{b \in \mathcal{B}} a_b g_1(b)h_1(t)+\sum \limits_{b \in \mathcal{B}} a_bg_2(b)+h_2(t)]dt}{1-F(t)} \\ = \dfrac{\sum \limits_{b \in \mathcal{B}}a_b c(b)- \displaystyle{\int}f(t)[g_1(a)h_1(t)+g_2(a)+h_2(t)]dt}{1-F(t)}$, which directly gives $$ \phi(a,t)=v(a,t)-\dfrac{1-F(t)}{f(t)}v_t(a,t)=g_1(a)h_1(t)+g_2(a)+h_2(t).$$
As a result, it suffices to focus on the expression for $v(b,t)$, where $h_1$ is monotonic, and $g_1$ and $g_2$ satisfy the expected utility property. Moreover, by continuity, $\lim \limits_{t \rightarrow \bar{t}} v(b,t)=v(b,\bar{t})$, indicating that 
\begin{equation*} \begin{split} c(b) & = \lim \limits_{t \rightarrow \bar{t}} \displaystyle{\int} f(t)[g_1(b)h_1(t)+g_2(b)+h_2(t)]dt \\ & = \displaystyle{\int} f(t)[g_1(b)h_1(t)+g_2(b)+h_2(t)]dt \bigg|_{t=\bar{t}},  \end{split} \end{equation*}
in which the last equality is by continuity of $f(t)[g_1(b)h_1(t)+g_2(b)+h_2(t)]$. Hence, $$ v(b,t)=\dfrac{\displaystyle{\int} f(t)[g_1(b)h_1(t)+g_2(b)+h_2(t)]dt \bigg|_{t=\bar{t}}-\displaystyle{\int} f(t)[g_1(b)h_1(t)+g_2(b)+h_2(t)]dt}{1-F(t)}.$$
This indeed gives $\lim \limits_{t \rightarrow \bar{t}} v(b,t)=v(b,\bar{t}):$ by L'H\^{o}pital's rule, $\lim \limits_{t \rightarrow \bar{t}} v(b,t)= \dfrac{\lim \limits_{t \rightarrow \bar{t}} f(t)[g_1(b)h_1(t)+g_2(b)+h(t)]}{\lim \limits_{t \rightarrow \bar{t}}f(t)}=g_1(b)h_1(\bar{t})+g_2(b)+h(\bar{t})=v(b,\bar{t})$.

\subsection{Quantile Transformation}
\label{sec:quantile}
It is standard that
\begin{equation*} \displaystyle{\int}_{\underline{t}}^{\bar{t}} \phi(b(t),t)dF(t)=\displaystyle{\int}_0^1 \tilde{\phi}(\tilde{b}(q),q)dq, \end{equation*}

where $\tilde{\phi}: \mathcal{B} \times [0,1]\rightarrow \mathbb{R}$ satisfies $\tilde{\phi}(\tilde{b}(q),q)=\phi(b \circ F^{-1}(q),F^{-1}(q))$ for any $q \in [0,1]$, in which $\tilde{b}(\cdot):= b \circ F^{-1}(\cdot)$. Hence, if $\phi$ satisfies Assumption \ref{a2}, then \begin{equation*} \begin{split} \tilde{\phi}(\tilde{b}(q),q)&=g_1(b \circ F^{-1}(q))h_1(F^{-1}(q))+g_2(b \circ F^{-1}(q))+h_2(F^{-1}(q)) \\ & = g_1(\tilde{b}(q)) h_1\circ F^{-1}(q)+g_2(\tilde{b}(q))+h_2 \circ F^{-1}(q). \end{split} \end{equation*} As $h_1$ is monotonic in $t$, $h_1 \circ F^{-1}$ is monotonic in $q$, indicating that for any $\tilde{a}, \tilde{a}' \in \Delta(\mathcal{B}), \tilde{\phi}(\tilde{a},q)-\tilde{\phi}(\tilde{a}',q)$ is monotonic in $q$, i.e., $\tilde{\phi}$ satisfies Assumption \ref{a2} on $\mathcal{B} \times [0,1]$.\footnote{Since $\tilde{b}:[0,1] \rightarrow \mathcal{B}$.}

The focus can thus be on the case where $T=[0,1]$ and $F \sim \mathcal{U}[0,1]$, the uniform distribution on $[0,1]$, which in turn gives:
 \begin{equation} \label{eq14} \begin{split} v(b,t) & =\dfrac{\displaystyle{\int} g_1(b)h_1(t)+g_2(b)+h_2(t)dt \bigg|_{t=1}-\displaystyle{\int}g_1(b)h_1(t)+g_2(b)+h_2(t)dt}{1-t}, t \in [0,1), \\ & v(b,1)=g_1(b)h_1(1)+g_2(b)+h_2(1), \end{split} \end{equation}
 for any $b \in \mathcal{B}$. Note that \eqref{eq14} can be easily rewritten as the ``$g_1(b)x(t)+g_2(b)+y(t)$'' form demonstrated in Section \ref{sec:discussionofassump}. 

 \subsection{Proof of Proposition \ref{prop:mdphiv}}

\begin{proof}
   Note that for any $b,b' \in \mathcal{B}$
   \begin{equation*} v(b,t)-v(b',t)=\dfrac{\displaystyle{\int}f(t)[\phi(b,t)-\phi(b',t)]dt\bigg|_{t=\bar{t}}-\displaystyle{\int}f(t)[\phi(b,t)-\phi(b',t)]dt}{1-F(t)}, \end{equation*} which implies that \small \begin{equation*} v_t(b,t)-v_t(b',t) \\ =\dfrac{f(t)[\displaystyle{\int}f(t)[\phi(b,t)-\phi(b',t)]dt\bigg|_{t=\bar{t}}-\displaystyle{\int}f(t)[\phi(b,t)-\phi(b',t)]dt-(1-F(t))(\phi(b,t)-\phi(b',t))]}{(1-F(t))^2}. \end{equation*} \normalsize
    Let \begin{equation*} \zeta(t)=\displaystyle{\int}f(t)[\phi(b,t)-\phi(b',t)]dt\bigg|_{t=\bar{t}}-\displaystyle{\int}f(t)[\phi(b,t)-\phi(b',t)]dt-(1-F(t))(\phi(b,t)-\phi(b',t)). \end{equation*} It is not hard to observe that $\zeta(\bar{t})=0$. Moreover, \begin{equation*} \begin{split} \zeta'(t) &=-f(t)[\phi(b,t)-\phi(b',t)]+f(t)[\phi(b,t)-\phi(b',t)]-(1-F(t))[\phi_t(b,t)-\phi_t(b',t)] \\ &=-(1-F(t))[\phi_t(b,t)-\phi_t(b',t)]. \end{split} \end{equation*} For any $t \in T $, $\phi_t(b,t)-\phi_t(b',t) \geq 0$ implies that for any $t \in T$, $ \zeta'(t) \leq 0$, which further implies that for any $t \in T, \zeta(t) \geq \zeta(\bar{t})=0$. As a result, for any $t \in T$, $v_t(b,t)-v_t(b',t)=\dfrac{f(t)\zeta(t)}{(1-F(t))^2} \geq 0$, i.e., $v(b,t)-v(b',t)$ is increasing in $t$. Similarly, if for any $t \in T, \phi_t(b,t)-\phi_t(b',t) \leq 0$, then $v(b,t)-v(b',t)$ is decreasing in $t$. As a result, the monotonicity of $\phi(b,t)-\phi(b',t)$ implies the monotonicity of $v(b,t)-v(b',t)$. 
\end{proof}

 \subsection{Proof of Proposition \ref{prop:equivalence}}

 \begin{proof}
     If $v$ satisfies \eqref{eq:multiadd}, then for any $a \in \Delta(\mathcal{B}), t \in T$, 
\begin{align*}
\phi(a,t)
&= \mathbb{E}_{b \sim a}\!\left[\phi(b,t)\right] \\
&= \mathbb{E}_{b \sim a}\!\left[
    g_1(b)\!\left(x(t)-\frac{1-F(t)}{f(t)}x'(t)\right)
    + g_2(b)
    + y(t)-\frac{1-F(t)}{f(t)}y'(t)
\right] \\
&= g_1(a)\!\left(x(t)-\frac{1-F(t)}{f(t)}x'(t)\right)
   + g_2(a)
   + y(t)-\frac{1-F(t)}{f(t)}y'(t).
\end{align*}
By \cite{kartik2023}, if $x(t)-\frac{1-F(t)}{f(t)}x'(t)$ is monotonic in $t$, then Assumption \ref{a2} holds. 
 \end{proof}

\section{Distributional Robustness}
\label{s522}

\cite{haghpanah2021when} studies distributional robustness in optimal bundling problems. In particular, they characterize when pure bundling is optimal for all type distributions. The first result of this section also answers the question ``when pure bundling is optimal''. 

\begin{proposition}
\label{p6}
     Assume that $v(b,t) \leq v(b^*,t), \forall b \in \mathcal{B}, t \in T$. If $\dfrac{v(b,t)}{v(b^*,t)}$ is nondecreasing in $t$ for all $b$, then pure bundling is optimal for all distributions under which $\phi$ satisfy Assumption \ref{a0}.
\end{proposition}

\begin{proof}
    For all $b \in \mathcal{B}$, let $c(b)=\displaystyle{\int}f(t)\phi(b,t)dt \bigg|_{t=\bar{t}}$. By the assumption that $\forall b \in \mathcal{B}, t \in T, v(b,t) \leq v(b^*,t)$,  \begin{equation} \label{eq3371} 0 \leq c(b)-\displaystyle{\int}f(t)\phi(b,t)dt \leq c(b^*)-\displaystyle{\int}f(t)\phi(b^*,t)dt, \forall b \in \mathcal{B}. \end{equation}
    In addition, $\dfrac{v(b,t)}{v(b^*,t)}$ is nondecreasing in $t$ for all $b \in \mathcal{B}$ implies that 
    \begin{equation*}  \label{eq3372} \begin{split} (c(b)-\displaystyle{\int} f(t)\phi(b,t)dt)'(c(b^*)-\displaystyle{\int} f(t)\phi(b^*,t)dt) \geq (c(b)-\displaystyle{\int} f(t)\phi(b,t)dt)(c(b^*)-\displaystyle{\int} f(t)\phi(b^*,t)dt)', \forall b \in \mathcal{B}, \end{split} \end{equation*}
    which can be simplified into 
    \begin{equation}
    \label{eq3373}
    \begin{split}
     \phi(b^*,t)(c(b)-\displaystyle{\int} f(t)\phi(b,t)dt) \geq  \phi(b,t)(c(b^*)-\displaystyle{\int} f(t)\phi(b^*,t)dt), \forall b \in \mathcal{B}.
     \end{split}
     \end{equation}
(\ref{eq3371}) and (\ref{eq3373}) together imply that when $\phi(b^*,t) \geq 0$
    \begin{equation*} \phi(b^*,t) \geq \phi(b,t), \forall b \in \mathcal{B}. \end{equation*}
     In addition, when $\phi(b^*,t)<0$, by (\ref{eq3373}), for $b \neq \emptyset$, $\phi(b,t)<0$. by $\{\emptyset\}=\argmax \limits_{\tilde{b}} \phi(\tilde{b}, \underline{t})$ and $v(\emptyset, \bar{t}) \leq v(b^*, \bar{t})$. Hence, when $\phi(b^*,t)<0, \phi(b,t)<\phi(\emptyset,t), \forall b \neq \emptyset$. As a result, a minimal optimal menu of \eqref{eq6} is $\{\emptyset, b^*\}$. Since $\phi$ satisfies Assumption \ref{a0}, $\{\emptyset, b^*\}$ is also a minimal optimal menu of \eqref{eq5}, i.e., pure bundling is optimal for all distributions under which $\phi$ satisfies Assumption \ref{a0}. 
\end{proof}

Proposition \ref{p6} is a weaker result than the ``if'' direction of Proposition 1 of \cite{haghpanah2021when}, as it uses the same conditions on the valuation functions to show the optimality of pure bundling for some but not all distributions. The role of Assumption \ref{a0} is to give Theorem \ref{t1}, which limits the discussions to the virtual values themselves and getting rid of discussions related to implementability. 

The partial converse of Proposition \ref{p6}, which corresponds to the ``only if'' direction of Proposition 1 of \cite{haghpanah2021when} will be shown in the following. To better study the partial converse, I first come up with two assumptions below: 

\begin{assumption}
\label{a4}
$\phi$ is increasing in $t$.
\end{assumption}

\begin{assumption}
\label{a5}
 $\forall b_1,b_2 \in \mathcal{B}$, if $\forall t, \phi(b_1,t) \geq \phi(b_2,t)$, then $\forall t, \phi_t(b_1,t) \leq \phi_t(b_2,t)$, i.e., $\phi(b_1,t)-\phi(b_2,t)$ is nonincreasing.    
\end{assumption}

Assumption \ref{a4} is arguably standard, although is redundant in all previous discussions, as Assumption \ref{a0} only requires $\phi$ to be monotonic but not necessarily increasing. Assumption \ref{a5} can be interpreted as given $b_2$ is very weakly dominated by $b_1$, the tendency of the difference between $\phi(b_1,t)$ and $\phi(b_2,t)$ is not very important, and in particular can be regarded as nonincreasing. 

With these two assumptions, the converse is stated as the following, in which a stronger optimality notion is needed. 

\begin{proposition}
\label{p7}
    If the minimal optimal menu of \eqref{eq5} is $\{b^*\}$ for a distribution under which $\phi$ satisfy Assumption \ref{a0}, \ref{a4} and \ref{a5}, then $\dfrac{v(b,t)}{v(b^*,t)}$ is nondecreasing in $t$ for all $b$.\footnote{Note that the converse does not need the assumption in Proposition \ref{p6} that $b^*$ has the largest valuation among all bundles for all types. }
\end{proposition}

\begin{proof}
    Let $F$ and $f$ be the cdf and the pdf of the distribution under which pure bundling is optimal and $\phi$ satisfy Assumption \ref{a0}, \ref{a4} and \ref{a5}. For any $b \in \mathcal{B}$, let $c(b)=\displaystyle{\int}f(t)\phi(b,t)dt \bigg|_{t=\bar{t}}$, so that $\dfrac{v(b,t)}{v(b^*,t)}=\dfrac{c(b)-\int f(t)\phi(b,t)dt}{c(b^*)-\int f(t)\phi(b^*,t)dt}$. To show that $\dfrac{v(b,t)}{v(b^*,t)}$ is nondecreasing in $t$ for all $b$, it suffices to show that (\ref{eq3373}) holds. 
    Consider \begin{equation*}
    \begin{split}
    \xi(t) &=\phi(b^*,t)[c(b)-\displaystyle{\int} f(t)\phi(b,t)dt]-\phi(b,t)[c(b^*)-\displaystyle{\int} f(t)\phi(b^*,t)dt], \\ &= \phi(b^*,t)[\displaystyle{\int} f(t)\phi(b,t)dt \bigg|_{t=\bar{t}}-\displaystyle{\int} f(t)\phi(b,t)dt]-\phi(b,t)[\displaystyle{\int} f(t)\phi(b^*,t)dt \bigg|_{t=\bar{t}}-\displaystyle{\int} f(t)\phi(b^*,t)dt],
    \end{split}
    \end{equation*}
\begin{equation*} \begin{split} \xi'(t) &=\phi_t(b^*,t)[\displaystyle{\int} f(t)\phi(b,t)dt \bigg|_{t=\bar{t}}-\displaystyle{\int} f(t)\phi(b,t)dt]-\phi(b^*,t)f(t)\phi(b,t) \\ &-\phi_t(b,t)[\displaystyle{\int} f(t)\phi(b^*,t)dt \bigg|_{t=\bar{t}}-\displaystyle{\int} f(t)\phi(b^*,t)dt]+\phi(b,t)f(t)\phi(b^*,t), \\ &= \phi_t(b^*,t)[\displaystyle{\int} f(t)\phi(b,t)dt \bigg|_{t=\bar{t}}-\displaystyle{\int} f(t)\phi(b,t)dt]- \phi_t(b,t)[\displaystyle{\int} f(t)\phi(b^*,t)dt \bigg|_{t=\bar{t}}-\displaystyle{\int} f(t)\phi(b^*,t)dt]. \end{split} \end{equation*} 
    
    $\{b^*\}$ is the minimal optimal menu of \eqref{eq5} under $F$ indicates that $\forall b \in \mathcal{B}, t \in T, \phi(b,t) \leq \phi(b^*,t)$, which further implies that $\forall b \in \mathcal{B}$, $\displaystyle{\int} f(t)[\phi(b^*,t)-\phi(b,t)]$ is increasing in $t$. Hence, 
    \begin{equation*} \begin{split} \displaystyle{\int} f(t)[\phi(b^*,t)-\phi(b,t)] &\leq \displaystyle{\int} f(t)[\phi(b^*,t)-\phi(b,t)] \bigg|_{t=\bar{t}} \\ \Rightarrow \displaystyle{\int} f(t)\phi(b,t)dt \bigg|_{t=\bar{t}}-\displaystyle{\int} f(t)\phi(b,t)dt &\leq \displaystyle{\int} f(t)\phi(b^*,t)dt \bigg|_{t=\bar{t}}-\displaystyle{\int} f(t)\phi(b^*,t)dt. \end{split} \end{equation*}

    By Assumption \ref{a4}, $\phi_t(b^*,t) \geq 0$, so that \begin{equation*} \xi'(t) \leq [\phi_t(b^*,t)-\phi_t(b,t)](f(t)\phi(b^*,t)dt \bigg|_{t=\bar{t}}-\displaystyle{\int} f(t)\phi(b^*,t)dt). \end{equation*} By Assumption \ref{a5}, $\forall t, \phi(b^*,t)-\phi(b,t) \geq 0 \Rightarrow \phi_t(b^*,t)-\phi_t(b,t) \leq 0$. Meanwhile, $\phi \geq 0$ implies that $\displaystyle{\int}f(t)\phi(b^*,t)$ is increasing in $t$, so that $f(t)\phi(b^*,t)dt \bigg|_{t=\bar{t}} \geq \displaystyle{\int} f(t)\phi(b^*,t)dt, $
which leads to $\xi'(t) \leq 0,\forall t \in T$. As $\xi(\bar{t})=0, \xi(t) \geq 0, \forall t \in T$, which directly gives (\ref{eq3373}). 

\end{proof}

Unlike the ``only if'' direction of Proposition 1 of \cite{haghpanah2021when}, this result only requires pure bundling to be optimal under a particular distribution instead of all distributions to obtain the nondecreasing valuation ratio property. However, the optimality of pure bundling is indeed stronger, as it requires $\{b^*\}$ instead of $\{\emptyset, b^*\}$ to be the minimal optimal menu of \eqref{eq5}. 

\cite{yang2023nested} has a characterization for two-tier menus, which is an immediate corollary of their characterization of robust nesting. With Assumption \ref{a4} and \ref{a5}, I am able to obtain a similar result that shares the fashion of both their two-tier menu characterization and their robust nesting characterization:

\begin{theorem}
\label{t3}
    Assume that $v(b,t) \leq v(b^*,t), \forall b \in \mathcal{B}, t \in T$. If there exists $b_1,...,b_l \in \mathcal{B}$ such that $v(b_1,\bar{t}) < v(b_2,\bar{t})<...<v(b_l,\bar{t})<v(b^*,\bar{t}) $, in which $b_1=\emptyset$, and
    \begin{itemize}
      \item  $\forall b \notin \{b_1,...,b_l\}, \dfrac{v(b,t)}{v(b^*,t)}$ is nondecreasing in $t$,
      \item $\forall i \in \{1,...,l\}, \dfrac{v(b_i,t)}{v(b_{i+1},t)}$is strictly decreasing in $t$,
      \item $\forall i \in \{2,...,l\}, \dfrac{v(b_{i-1},t)-v(b_{i+1},t)}{v(b_i,t)}$ is strictly increasing in $t$,
    \end{itemize}
    then $\{b_1,...,b_l,b^*\}$ is the minimal optimal menu for all distributions under which $\phi$ satisfy Assumption \ref{a0}, \ref{a1}, \ref{a4} and \ref{a5}.\footnote{Define $b_{l+1}$ as $b^*$. In addition, the reason why Assumption \ref{a1} is mentioned in Theorem \ref{t3} but not in Proposition \ref{p6} and \ref{p7} is because the proofs of \ref{p6} and \ref{p7} make ``global'' comparisons of $\phi$, whereas the proof of Theorem \ref{t3} only compares $\phi$ at $\underline{t}$ and $\bar{t}$. }
\end{theorem}

\begin{proof}

By proof of Proposition \ref{p6}, if $\forall b \notin \{b_1,..,b_l\}, \dfrac{v(b,t)}{v(b^*,t)}$ is nondecreasing in $t$, then $\forall b \in \{b_1,...,b_l\}$, $b$ is very weakly dominated by a convex combination of $\emptyset$ and $b^*$. Hence, by Proposition \ref{p4}, it suffices to show every bundle of $\{b_1,...,b_l,b^*\}$ is not very weakly dominated. I first show that $\phi(b_1,\underline{t})>\phi(b_2,\underline{t})>...>\phi(b_l,\underline{t})>\phi(b^*,\underline{t})$. Suppose that there exists $l' \in \{1,...,l\}$ such that $\phi(b_{l'},\underline{t}) \leq \phi(b_{l'+1},\underline{t})$, then since $\phi(b_{l'},\bar{t})=v(b_{l'},\bar{t})<v(b_{l'+1},\bar{t})=\phi(b_{l'+1},\bar{t})$, which, by Assumption \ref{a0}, further implies that $\phi(b_{l'},t) \leq \phi(b_{l'+1},t), \forall t$. By similar arguments as the ones in proof of Proposition \ref{p7}, this implies that $\dfrac{v(b_{l'},t)}{v(b_{l'+1},t)}$ is nondecreasing in $t$, a contradiction to the strictly decreasing condition. 

With $\phi(b_1,\underline{t})>\phi(b_2,\underline{t})>...>\phi(b_l,\underline{t})>\phi(b^*,\underline{t}); \phi(b_1,\bar{t})<\phi(b_2,\bar{t})<...<\phi(b_l,\bar{t})<\phi(b^*,\bar{t})$, by Algorithm \ref{al1} (or more precisely, the proof of Theorem \ref{t2}), it suffices to show that $b_i$ is not very weakly dominated by any convex combinations of $b_{i-1}$ and $b_{i+1}$, for all $i \in \{1,...,l\}$. Suppose there exists $i$ and $\lambda \in [0,1]$ such that $b_i$ is very weakly dominated by $\lambda b_{i-1}+(1-\lambda) b_{i+1}$, again by similar arguments as the ones in proof of Proposition \ref{p7}, this implies that $\dfrac{v(b_i,t)}{\lambda v(b_{i-1},t)+(1-\lambda) v(b_{i+1},t)}$ is nondecreasing in $t$. Consider \begin{equation*} \dfrac{\lambda v(b_{i-1},t)+(1-\lambda) v(b_{i+1},t)}{v(b_i,t)}=\lambda \dfrac{v(b_{i-1},t)-v(b_{i+1},t)}{v(b_i,t)}+ \dfrac{v(b_{i+1},t)}{v(b_i,t)}. \end{equation*} It is strictly increasing in $t$ because both $\dfrac{v(b_{i-1},t)-v(b_{i+1},t)}{v(b_i,t)}$ and $\dfrac{v(b_{i+1},t)}{v(b_i,t)}$ are strictly increasing in $t$, which implies that $\dfrac{v(b_i,t)}{\lambda v(b_{i-1},t)+(1-\lambda) v(b_{i+1},t)}$ is strictly decreasing in $t$, a contradiction to nondecreasing. As a result, every bundle of $\{b_1,...,b_l,b^*\}$ is not very weakly dominated. Hence, $\{b_1,...,b_l,b^*\}$ is the minimal optimal menu of \eqref{eq5}. 
    
\end{proof}

Similar to Proposition \ref{p6}, the first condition is used to demonstrate that every $b \notin \{b_1,...,b_l\}$ is very weakly dominated. The last two conditions, on the other hand, are helpful for showing every bundle in the menu is not very weakly dominated. In particular, when $l=1$, i.e., the minimal optimal menu is the two-tier $\{b_1,b^*\}$, the conditions reduce to $\dfrac{v(b_1,t)}{v(b^*,t)}$ is strictly decreasing in $t$ and $\forall b \neq b_1, \dfrac{v(b,t)}{v(b^*,t)}$ is nondecreasing in $t$, corresponding with Corollary 6 of \cite{yang2023nested}.

\section{Numerical Examples of Tree Menus}
\label{f4}

It is worth discussing and coming up with numerical examples of what $\phi$, and consequently, $v$ satisfy the conditions specified in Theorems \ref{t33} and \ref{t34} and Assumptions \ref{a0} and \ref{a1} at the same time. For better characterizations, as demonstrated in Section \ref{s23}, I directly let $\phi$ satisfy Assumption \ref{a2}, i.e., equation (\ref{eq9}) holds for $\phi$. For Theorem \ref{t33}, first focus on the second condition. It is equivalent to for any $\tilde{b}$, there exists $\lambda \in [0,1]$ and $\mathfrak{f}: [0,1] \rightarrow [0,1]$ such that \begin{equation*} \begin{split} \text{if } g_1(\tilde{b})h_1(\bar{t})+g_2(\tilde{b}) &=[\lambda g_1(b^*)+(1-\lambda) g_1(b)]h_1(\bar{t})+[\lambda g_2(b^*)+(1-\lambda) g_2(b)], \\ \text{then } g_1(\tilde{b})h_1(\underline{t})+g_2(\tilde{b}) &=[\mathfrak{f}(\lambda) g_1(b^*)+(1-\mathfrak{f}(\lambda)) g_1(b)]h_1(\underline{t})+[\mathfrak{f}(\lambda) g_2(b^*)+(1-\mathfrak{f}(\lambda)) g_2(b)]. \end{split} \end{equation*} Meanwhile, $b$ being the root bundle implies that $g_1(b^*)h_1(\underline{t})+g_2(b^*)<g_1(b)h_1(\underline{t})+g_2(b)$. Once $b, g_1(b), g_2(b), g_1(b^*), g_2(b^*), h_1(\bar{t}), h_1(\underline{t}), \lambda, \mathfrak{f}$ is chosen to satisfy the conditions in the theorem and the $b$ being the root bundle inequality, the design of $g_1$ and $g_2$ is equivalent to solving \begin{equation*} \begin{pmatrix}  h_1(\bar{t}) & 1 \\ h_1(\underline{t}) & 1\end{pmatrix} \begin{pmatrix} g_1(\tilde{b}) \\ g_2(\tilde{b}) \end{pmatrix}= \begin{pmatrix} [\lambda g_1(b^*)+(1-\lambda) g_1(b)]h_1(\bar{t})+[\lambda g_2(b^*)+(1-\lambda) g_2(b)]\\ [\mathfrak{f}(\lambda) g_1(b^*)+(1-\mathfrak{f}(\lambda)) g_1(b)]h_1(\underline{t})+[\mathfrak{f}(\lambda) g_2(b^*)+(1-\mathfrak{f}(\lambda)) g_2(b)]\end{pmatrix}, \end{equation*} in which $g_1(\tilde{b})$ and $g_2(\tilde{b})$ are the unknowns, and clearly when $h_1(\bar{t}) \neq h_1(\underline{t})$, the existence of the solution is guaranteed for any $\lambda$. Example \ref{ex:tree} is exactly an example that makes $\{\{1\}, \{1,2\}, \{1,3\}, \{1,2,3\}\}$ minimally optimal when the consumer's type is uniformly distributed on $[0,1]$. 

For Theorem \ref{t34}, when selling goods $\{1,2,3,4\}$, an example that makes $\{\{1\}, \{1,4\}, \{1,2,3\}, \{1,2,3,4\}\}$ a subset of the minimal optimal menu when the consumer's type is uniformly distributed on $[0,1]$, and the monopolist has four goods to sell is: $h_1(t)=t, h_2(t)=0, g_1, g_2$ satisfy the following:

\begin{center}
\begin{tabular}{lccccccccc}
$b$ & $\{1\}$ & $\{1,4\}$ & $\{1,3\}$ & $\{1,2\}$ &
$\{1,2,4\}$ & $\{1,3,4\}$ & $\{1,2,3\}$ &
$\{1,2,3,4\}$ & others \\
\hline
$g_1(b)$ &
$1$ &
$\frac{13}{6}$ &
$\frac{8}{3}$ &
$\frac{29}{12}$ &
$3$ &
$\frac{23}{8}$ &
$\frac{43}{12}$ &
$6$ &
$1$ \\
$g_2(b)$ &
$4$ &
$\frac{7}{2}$ &
$3$ &
$\frac{13}{4}$ &
$3$ &
$\frac{25}{8}$ &
$\frac{11}{4}$ &
$1$ &
$<4$
\end{tabular}
\end{center}

A corresponding graph can be found here: \href{https://www.geogebra.org/m/hungg7gm}{https://www.geogebra.org/m/hungg7gm}. From the graph, $\{\{1\},\{1,4\},\{1,2,3\},\{1,2,3,4\}\}$ is not only a subset but is exactly the minimal optimal menu.\footnote{This is not generally implied by Theorem \ref{t34}: Theorem \ref{t34} only says that the minimal optimal menu includes $\{b^{root}, b^{root} \cup \{i\}, b^* \backslash \{i\}, b^*\}$.}

\section{Supplementary Discussions}

\subsection{Monotone Comparative Statics}

This subsection mainly focuses on the relationship between Theorem \ref{t2} and certain monotone comparative statics (MCS) theorems, particularly the ones in \cite{kartik2023} and \cite{milgrom1994}. There exists a superficial discrepancy that Theorem \ref{t2} requires monotonic differences, whereas all these MCS theorems only require single-crossing differences. I am going to show that Theorem \ref{t2} is not obtained because the MCS theorems naturally hold under a stronger condition. In other words, if I only require $\phi$ to have single-crossing differences, then those MCS theorems cannot be applied, and MCS is not guaranteed. 

I first formalize my problem in the language of a MCS problem. Specifically, I will define a partial order on the set of bundles $\mathcal{B}$. Let $\succeq$ be a partial order on $\mathcal{B}$ such that for any $b,b' \in \mathcal{B}$, \begin{equation*} b \succeq b' \Leftrightarrow v(b,t)-v(b',t) \geq v(b,t')-v(b',t'), \forall t \geq t'. \end{equation*}
In words, $b \succeq b'$ if and only if $v(b,t)-v(b',t)$ is increasing in $t$. It is straightforward to see that this order is induced by the IC constraints. To put it in another way, if IC constraints are satisfied, then the bundles in the minimal optimal menu shall be monotonically ordered, or shall satisfy MCS induced by the order above. 

I proceed to show that this is correct under the definition of MCS in \cite{kartik2023}. Formally, let the minimal optimal menu be $\{b_1,...,b_m\}$, under Assumptions \ref{a0} and \ref{a1}, for any $B \subseteq \{b_1,...,b_m\}, t \leq t'$, \begin{equation*} \argmax \limits_{b \in B} \phi(b, t') \succeq_{SSO} \argmax \limits_{b \in B} \phi(b,t), \end{equation*}
where $\succeq_{SSO}$ is the strong set order. The proof is not complicated: as $(\{b_1,...,b_m\}, \succeq)$ is a totally ordered set (by Lemma \ref{l1}), the only cases I need to contradict are $\exists B \subseteq \{b_1,...,b_m\}, b^h,b^l \in B, b^h \succ b^l, b^h \in \argmax \limits_{b \in B} \phi(b,t), b^l \in \argmax \limits_{b \in B} \phi(b,t')$ and i) $b^l \notin \argmax \limits_{b \in B} \phi(b,t)$ or ii) $b^h \notin \argmax \limits_{b \in B} \phi(b,t')$. Suppose i) is correct, then $\phi(b^h,t)>\phi(b^l,t); \phi(b^h,t') \leq \phi(b^l,t')$. By the proof of Lemma \ref{l1}, this indicates that $v(b^l,t)-v(b^h,t)$ is increasing in $t$, i.e., $b^l \succeq b^h$, a contradiction to $b^h \succ b^l$. Similar arguments can be used to show that ii) is also incorrect, which gives $\argmax \limits_{b \in B} \phi(b, t') \succeq_{SSO} \argmax \limits_{b \in B} \phi(b,t)$.

Theorem 4 of \cite{kartik2023} says that if the above results hold, then $\phi$ has single-crossing differences on $\{b_1,...,b_m\} \times T$, and $\succeq$ is a refinement of $\succeq_{SCD}$, in which $\succeq_{SCD}$ is defined as $b \succeq_{SCD} b'$ if and only if $\sign[\phi(b,t)-\phi(b',t)]$ is increasing on $T$. At a surface level, $\phi$ only needs to have single-crossing differences. However, since the specific $\succeq$ is $\phi-$ dependent, whether $\succeq$ is a refinement of $\succeq_{SCD}$ depends on $\phi$ as well, which possibly brings extra structures to $\phi$. More explicitly, if I only let $\phi$ have single-crossing differences, then $\succeq$ may not be a refinement of $\succeq_{SCD}$. Example \ref{e7} offers a supporting example: $\{b,b'\}$ is the minimal optimal menu for problem \eqref{eq6}, $b \succeq_{SCD} b'$, but $b \nsucceq b'$ because $v(b,t)-v(b',t)$ is not increasing in $t$. On the other hand, if $\phi$ has monotonic differences, then $\succeq$ is indeed a refinement of $\succeq_{SCD}$: $b \succeq_{SCD} b' \Rightarrow \phi(b,t)-\phi(b',t)$ increases in $t$, which, by the proof of Lemma \ref{l1}, further implies that $v(b,t)-v(b',t)$ increases in $t$, which is exactly $b \succeq b'$. 

Similar issues exist when applying the monotonicity theorem (Theorem 4) of \cite{milgrom1994} if $\phi$ is only assumed to have single-crossing differences. In particular, if $\phi$ has single-crossing differences but not monotonic differences, then $(\{b_1,...,b_m\},\succeq)$ may not be a lattice: in Example \ref{e7}, since $b \nsucceq b'$ and $b' \nsucceq b$, neither $b \vee b'$ nor $b \wedge b'$ exists, thus not an element of $\{b,b'\}$.On the other hand, if $\phi$ has monotonic differences, then by Lemma \ref{l1}, $(\{b_1,...,b_m\},\succeq)$ is totally ordered, which is naturally a lattice. \cite{milgrom1994}'s monotonicity theorem also naturally holds, as $\phi$ is trivially quasisupermodular in $b$ when $(\{b_1,...,b_m\},\succeq)$ is totally ordered. 

It is worth mentioning that $\phi$ only has MCS on $(\{b_1,...,b_m\},\succeq)$ but not on $(\mathcal{B}, \succeq)$, so that neither \cite{kartik2023}'s theorem nor \cite{milgrom1994}'s theorem can be directly applied to the original optimal bundling problem.

\subsection{Optimality of Singleton-Bundle Menus}

The question that ``when the menu $\{\emptyset, \{1\}, ..., \{n\}\}$ is minimally optimal'' has not been answered much in the literature. Algorithm \ref{al1} shed light on structured and characterized solutions to the question. To begin with, recall that for any $b \in \mathcal{B}, \phi(b,t)=g_1(b)h_1(t)+g_2(b)+h_2(t)$, in which $h_1$ is monotonic, and $g_1$, $g_2$, and $h_2$ are arbitrary. The ``sandwich'' condition $\phi(b,\bar{t})>\phi(b',\bar{t})>\phi(b'',\bar{t}); \phi(b,\underline{t})<\phi(b',\underline{t})<\phi(b'',\underline{t})$ can be translated into \begin{equation*} \begin{split} [g_1(b)-g_1(b')]h_1(\bar{t}) > g_2(b')-g_2(b); [g_1(b')-g_1(b'')]h_1(\bar{t}) > g_2(b'')-g_2(b'), \\
 [g_1(b)-g_1(b')]h_1(\underline{t}) < g_2(b')-g_2(b); [g_1(b')-g_1(b'')]h_1(\underline{t}) < g_2(b'')-g_2(b'),\end{split} \end{equation*}
and the condition (Lemma \ref{l2}) for $b'$ to be very weakly dominated by a convex combination of $b$ and $b''$ can be translated into
\begin{equation*} \dfrac{[(g_1(b)-g_1(b'))h_1(\bar{t})]+[g_2(b)-g_2(b')]}{[(g_1(b')-g_1(b''))h_1(\bar{t})]+[g_2(b')-g_2(b'')]} \geq \dfrac{[(g_1(b)-g_1(b'))h_1(\underline{t})]+[g_2(b)-g_2(b')]}{[(g_1(b')-g_1(b''))h_1(\underline{t})]+[g_2(b')-g_2(b'')]}. \end{equation*}
To simplify notations, with an abuse of notation ``$b$'', 
\begin{table}[H]
\centering
\begin{tabular}{|c|c|}
\hline
$g_1(b)-g_1(b')$& $a$ \\ \hline
$g_2(b')-g_2(b)$& $b$ \\ \hline
$g_1(b')-g_1(b'')$& $c$ \\ \hline
$g_2(b'')-g_2(b')$& $d$ \\ \hline
$h_1(\bar{t})$& $x$ \\ \hline
$h_1(\underline{t})$& $y$ \\ \hline

\end{tabular}
\end{table}

The conditions required are 
\begin{equation*} ax > b; cx > d; ay < b; cy < d,\end{equation*} and 
\begin{equation} \label{eq85} \dfrac{ax-b}{cx-d} \geq \dfrac{ay-b}{cy-d} \Leftrightarrow (ad-bc)(x-y) \geq 0. \end{equation}
Without loss of generality, assume that $x > y$ ($h_1(\bar{t}) > h_1(\underline{t})$), then (\ref{eq85}) is equivalent to $ad \geq bc$.\footnote{The $x=y$ case is simple: the minimal optimal menu is $\argmax \limits_{b \in \mathcal{B}} g_1(b)h_1(\bar{t})+g_2(b)+h_2(\bar{t}) \cup \argmax \limits_{b \in \mathcal{B}} g_1(b)h_1(\underline{t})+g_2(b)+h_2(\underline{t}) $. } In conclusion, $b'$ is neither very weakly dominated by $b$ nor very weakly dominated by $b''$ but is very weakly dominated by a convex combination of them if and only if $x>y;ax>b;cx>d;ay<b;cy<d;ad \geq bc$.\footnote{``If and only if'' when $h_1$ is restricted to monotone increasing.}


I proceed to characterize an equivalent condition for the inequalities above to hold by  the following proposition.

\begin{proposition}
\label{p9}
    There exist $x$ and $y$ satisfying the above inequalities if and only if one of the following conditions hold:
    \begin{enumerate}
        \item \label{cond1} $a,b,c,d > 0; ad \geq bc$
        \item \label{cond2} $a,c>0; b,d<0; ad \geq bc$
        \item \label{cond3} $a,c>0; b\leq 0; d \geq 0$

    \end{enumerate}
\end{proposition}

\begin{proof}
    ``if'': The ``if'' direction can be shown by checking the four conditions one-by-one. 
    \begin{itemize}
        \item Under condition \ref{cond1}, any $x$ and $y$ satisfying $x>\dfrac{d}{c}$ and $y<\dfrac{b}{a}$ can satisfy the above inequalities.
        \item Under condition \ref{cond2}, any $x$ and $y$ satisfying $x > \dfrac{d}{c}$ and $y< \dfrac{b}{a}$ can satisfy the above inequalities.
        \item Under condition \ref{cond3}, $ad \geq 0 \geq bc$, and any $x$ and $y$ satisfying $x> \dfrac{d}{c}$ and $y<\dfrac{b}{a}$ can satisfy the above inequalities. 
    \end{itemize}

    ``only if'': I first show that the above inequalities imply $a,c>0$. If $a<0$, then $x<\dfrac{b}{a}<y$, a contradiction to $x>y$; if $a=0$, then $0<b<0$, again a contradiction. The same arguments for $c$. 

    Next consider the situation where $b,d \neq 0$. If $d<0, b>0$, then $ad<0<bc$, a contradiction. If $b,d>0$, then an additional $ad \geq bc$ is required, which is exactly condition \ref{cond1}, and any $x > \dfrac{d}{c}, y<\dfrac{b}{a}$ can satisfy the above inequalities. Similarly, if $b,d <0$, then the additional condition is needed, which is exactly condition \ref{cond2}. If $b<0,d>0$, then $ad>0>bc$, and still $x > \dfrac{d}{c}$ and $y<\dfrac{b}{a}$ can satisfy the above inequalities. 

    Now consider $b=0$ or $d=0$. If $b=0$, then $ad \geq bc$ implies that $d \geq 0$, and any $x > \dfrac{d}{c}$ and $y<0$ can satisfy the above inequalities. If $d=0$, then similarly, $b \leq 0$, and any $x>0$ and $y < \dfrac{b}{a}$ can satisfy the above inequalities. Combining with the previous discussion regarding $b<0, d>0$, condition \ref{cond3} is obtained. 
\end{proof}

Let the minimal optimal menu be $\{b_1,b_2,...,b_m\}$, in which $\phi(b_m,\bar{t})>...>\phi(b_1,\bar{t}); \phi(b_m,\underline{t})<...<\phi(b_1,\underline{t})$. Assume that for any $i \in \{1,2,...,m-1\}$, there exists bundle $b \in \mathcal{B} \backslash \{b_1,...,b_m\}$ such that $b$ is very weakly dominated by a convex combination of $b_i$ and $b_{i+1}$ but is not very weakly dominated by either $b_i$ or $b_{i+1}$. By Proposition \ref{p9}, $a,c>0$ implies that $g_1(b_m)>...>g_1(b_1)$. A condition that can ensure $\forall i, b_i$ is not very weakly dominated is $g_2(b_m)>...>g_2(b_1)$ and $g_1(b_i)=g_2(b_i)^2, \forall i$. Under such a condition, $a,c>0; b,d<0$, yet \begin{equation*} \begin{split} ad &=[g_2(b_{i+1})+g_2(b_i)][g_2(b_{i+1})-g_2(b_i)][g_2(b_{i-1})-g_2(b_i)] \\ &< [g_2(b_i)+g_2(b_{i-1})][g_2(b_i)-g_2(b_{i-1})][g_2(b_i)-g_2(b_{i+1})]=bc. \end{split} \end{equation*}

For any $i$, what are the conditions that the bundles ``between'' $b_{i-1}$ and $b_i$ ($b^{(i-1,i)}$) should satisfy, assuming that every bundle in between is only very weakly dominated by a non-degenerated convex combination? It is not hard to observe that among the three conditions in Proposition \ref{p9}, condition \ref{cond3} is more ``handy'' in the sense that $ad \geq bc$ is naturally satisfied. For condition \ref{cond3} to be satisfied,
\begin{equation*} \begin{split} g_1(b_i) > &g_1(b^{(i-1,i)}) > g_1(b_{i-1}), \\ 
  g_2(b_{i-1}) \geq g_2(b^{(i-1,i)}) &; g_2(b_i) \geq g_2(b^{(i-1,i)}). \end{split} \end{equation*}
The inequality regarding $g_2$ gives potential to $g_2(b_i) > g_2(b_{i-1})$, which is condition discussed above for $b_i$ to be not very weakly dominated. 

The characterization of $\{\emptyset, \{1\}, ...., \{n\}\}$ being the minimal optimal menu follows from the above discussion with more specifications. Assume that $\phi(\emptyset, \bar{t})< \phi(\{1\}, \bar{t})<...<\phi(\{n\},\bar{t}); \phi(\emptyset, \underline{t}) > \phi(\{1\}, \underline{t})>...>\phi(\{n\},\underline{t})$. Given $0<g_1(\{1\})<...<g_1(\{n\}), \forall i \in \{1,...,n\}$, for any $\{x_1,...,x_m\} \subseteq \{1,...,n\}$, in which $x_1<...<x_m$, it is not hard to see $g_1(\{x_1,...,x_m\})=\dfrac{g_1(\{x_1\})+...+g_1(\{x_m\})}{m}=\dfrac{g_2(\{x_1\})^2+...+g_2(\{x_m\})^2}{m}; g_2(\{x_1,...,x_m\})=\dfrac{g_2(\{x_1\})+...+g_2(\{x_m\})}{m}$ can make the singleton bundle menu minimal optimal.\footnote{$x_1<...<x_m$ implies that $m \geq 2$.} In particular, such $g_1$ and $g_2$ satisfy condition \ref{cond2} in Proposition \ref{p9}. However, under such $g_1$ and $g_2$, the optimality of the singleton-bundle menu can be argued through the standard Myersonian approach by comparing the virtual value pointwise, as for any $t, \phi(\{x_1,...,x_m\},t) \leq \max \limits_{i \in \{1,...,m\}} \phi(\{x_i\},t)$.

Instead of focusing on condition \ref{cond2} (or \ref{cond1}), if following condition \ref{cond3} in Proposition \ref{p9}, the conditions needed for $\{\emptyset, \{1\}, ...., \{n\}\}$ to be minimal optimal are

\begin{itemize}
    \item $g_1(\{x_m\})>g_1(\{x_1,...,x_m\})>g_1(\{x_1\})$
    \item $g_2(\{x_m\})>g_2(\{x_1,...,x_m\}); g_2(\{x_1\})>g_2(\{x_1,...,x_m\})$
    \item $[g_1(\{i+1\})-g_1(\{i\})][g_2(\{i-1\})-g_2(\{i\})]<[g_2(\{i\})-g_2(\{i+1\})][g_1(\{i\})-g_1(\{i-1\})], \forall i$
\end{itemize}

Here is a concrete examples that both satisfy all the above conditions and does not satisfy $\phi(\{x_1,...,x_m\},t) \leq \dfrac{1}{m}(\phi(\{x_1\},t)+...+\phi(\{x_m\},t))$: if the inequality holds, the minimal optimal menu can again be easily obtained through direct pointwise comparisons of virtual values. To avoid this, since $g_2(\{x_1,...,x_m\})<g_2(\{x_1\})<\dfrac{1}{m}(g_2(\{x_1\})+...+g_2(\{x_m\}))$, it has to be the case that $g_1(\{x_1,...,x_m\}) \neq \dfrac{1}{m}(g_1(\{x_1\})+...+g_1(\{x_m\})$. An example can be $g_1(\{x_1,...,x_m\})=\dfrac{g_2(\{x_1\})^2+...+g_2(\{x_m\})^2}{m-1+\frac{g_2(\{x_1\})^2}{g_2(\{x_m\})^2}}$, and it is not hard to check that $g_1(\{x_1\})<g_1(\{x_1,...,x_m\})<g_1(\{x_m\})$.The corresponding valuation function can be \begin{equation*} v(\{x_1,...,x_m\},t)= rt \dfrac{\sum \limits_{i=1}^m g_2(\{x_i\})^2}{m-1+\frac{g_2(\{x_1\})^2}{g_2(\{x_m\})^2}}+  \dfrac{\sum \limits_{i=1}^m g_2(\{x_i\})}{m\frac{g_2(\{x_m\})}{g_2(\{x_1\})}}, \end{equation*} where $r$ is a constant determined by $h_1(\bar{t})$ and $h_1(\underline{t})$. Under such a valuation function, the standard approach is not able to generate a solution without tedious and redundant comparisons of virtual values on each $t$, whereas the discussion in this section directly pins down $\{\emptyset, \{1\}, ..., \{n\}\}$ to be optimal. 

\subsection{Characterization of Minimal Optimal Menu: A ``Dominance-Free'' Approach}

The notion of dominance seems to mostly play an intermediary role: it bridges the gap between conditions needed for $\phi$ and the corresponding solution to the optimal bundling problem. A natural question is whether there are more straightforward characterizations that directly use the properties of $\phi$ rather ran going over the dominance notion. The answer is positive, but as will be shown below, completely getting rid of the dominance notion, or more precisely, completely getting rid of the equivalence between never strict best response and very weakly dominance will cost the loss of certain properties that only exist under the aforementioned equivalence, making the characterization rather tedious. Once bringing back the equivalence, although the characterization can be stated in a way that has nothing to do with the dominance notion, the underlying intuition and the proof does need to rely on the idea of dominance. In that case, the role that the notion of dominance plays is not only intermediary but fairly fundamental. Another angle to interpret the fundamental role is that all the results below can be regarded as corollaries of previous results that are build up from the dominance notion. 

Corresponding to a previous remark that the condition $\forall b,b' \in \mathcal{B}, \phi(b,t)-\phi(b',t)$ is monotonic in $t$ can be a sufficient condition for Lemma \ref{l1}, a characterization of minimal optimal menus without the introduction of Algorithm \ref{al1} and the notion of dominance is demonstrated by the following corollary.\footnote{I say ``minimal optimal menus'' because, as shown by Example \ref{e6}, minimal optimal menu may not be unique when Assumption \ref{a1} does not hold.}
\begin{corollary}
\label{c5}
    Under Assumption \ref{a0}, $\{b_1,...,b_m\}$ is a minimal optimal menu if and only if   
    \begin{itemize}
    \item $ \{t| \exists b_i \in \{b_1,...,b_m\}, b_i \in \argmax \limits_{b \in \mathcal{B}} \phi(b,t) \} =T$

    \item $\forall i \in \{1,...,m\}, \{t| b_i \in \argmax \limits_{b \in \mathcal{B}} \phi(b,t) \} \neq \emptyset$

    \item $\forall i \in \{1,...,m\}, \bigcup \limits_{j \neq i} [\{t|b_i \in \argmax \limits_{b \in \mathcal{B}} \phi(b,t) \} \cap \{t| b_j \in \argmax \limits_{b \in \mathcal{B}} \phi(b,t) \} ]\neq \{t| b_i \in \argmax \limits_{b \in \mathcal{B}} \phi(b,t) \}$

    \end{itemize}
\end{corollary}

\begin{proof}

    ``if'': First claim that $\forall i,j \in \{1,...,m\}, i \neq j, \phi(b_i,\bar{t}) \neq \phi(b_j,\bar{t})$. If there exists $i \neq j$ such that $\phi(b_i,\bar{t})=\phi(b_j,\bar{t})$, then by monotonicity of $\phi(b_i,t)-\phi(b_j,t)$, it is either $\phi(b_i,t) \leq \phi(b_j,t), \forall t$, which further implies that $\{t|b_i \in \argmax \limits_{b \in \mathcal{B}} \phi(b,t)\} \subseteq \{t|b_j \in \argmax \limits_{b \in \mathcal{B}} \phi(b,t)\}$, a contradiction to $\bigcup \limits_{j \neq i} [\{t|b_i \in \argmax \limits_{b \in \mathcal{B}} \phi(b,t) \} \cap \{t| b_j \in \argmax \limits_{b \in \mathcal{B}} \phi(b,t) \} ]\neq \{t| b_i \in \argmax \limits_{b \in \mathcal{B}} \phi(b,t) \}$, or $\phi(b_j,t) \leq \phi(b_i,t), \forall t$, which further implies that $\{t|b_j \in \argmax \limits_{b \in \mathcal{B}} \phi(b,t)\} \subseteq \{t|b_i \in \argmax \limits_{b \in \mathcal{B}} \phi(b,t)\}$, a contradiction to $\bigcup \limits_{i \neq j} [\{t|b_j \in \argmax \limits_{b \in \mathcal{B}} \phi(b,t) \} \cap \{t| b_i \in \argmax \limits_{b \in \mathcal{B}} \phi(b,t) \} ]\neq \{t| b_j \in \argmax \limits_{b \in \mathcal{B}} \phi(b,t) \}$.
    
 With this claim, it is without loss to let $\phi(b_1,\underline{t})<\phi(b_2,\underline{t})<...<\phi(b_m,\underline{t})$. Next claim that for any $i<j$, it has to be the case that $\phi(b_i,\bar{t})>\phi(b_j,\bar{t})$. Suppose not, then by the same logic above, it cannot be the case that $\phi(b_i,\bar{t})=\phi(b_j,\bar{t})$. Consider $\phi(b_i,\bar{t})<\phi(b_j,\bar{t})$, then by monotonicity, $\forall t, \phi(b_i,t)<\phi(b_j,t)$, a contradiction to $\{t| b_i \in \argmax \limits_{b \in \mathcal{B}} \phi(b,t) \} \neq \emptyset$. In that case, $\phi(b_1,\bar{t})>...>\phi(b_m,\bar{t})$. By Lemma \ref{l1}, $\{b_1,...,b_m\}$ is a feasible menu. Moreover, $ \{t| \exists b_i \in \{b_1,...,b_m\}, b_i \in \argmax \limits_{b \in \mathcal{B}} \phi(b,t) \} =T$ suggests that $\{b_1,..,b_m\}$ is an optimal menu for problem \eqref{eq6}, which, combined with feasibility, further suggests that $\{b_1,..,b_m\}$ is an optimal menu. Suppose that there exists $B \subsetneq \{b_1,...,b_m\}$ such that $B$ is also an optimal menu. $\{b_1,..,b_m\}$ is an optimal menu for problem \eqref{eq6} indicates that $B$ is also an optimal menu for problem \eqref{eq6}, which is a contradiction to $\forall i,j \in \{1,...,m\}, i \neq j, \{t|  b_i \in \argmax \limits_{b \in \mathcal{B}} \phi(b,t) \} \cap \{t| b_j \in \argmax \limits_{b \in \mathcal{B}} \phi(b,t) \}= \emptyset$. As a result, $\{b_1,...,b_m\}$ is an minimal optimal menu. 

 ``only if'': Let $\{b_1,...,b_m\}$ be a minimal optimal menu for problem \eqref{eq6}. Optimality implies that $ \{t| \exists b_i \in \{b_1,...,b_m\}, b_i \in \argmax \limits_{b \in \mathcal{B}} \phi(b,t) \} =T$. Suppose that there exists $i \in \{1,...,m\}$ such that $\{t| b_i \in \argmax \limits_{b \in \mathcal{B}} \phi(b,t) \} = \emptyset$, then $\{b_1,...,b_m\} \backslash \{b_i\}$ is also an optimal menu for problem \eqref{eq6}, a contradiction to the minimality. Similarly, suppose that there exists $i \in \{1,...,m\}, $ such that $\bigcup \limits_{j \neq i} [\{t|b_i \in \argmax \limits_{b \in \mathcal{B}} \phi(b,t) \} \cap \{t| b_j \in \argmax \limits_{b \in \mathcal{B}} \phi(b,t) \} ]= \{t| b_i \in \argmax \limits_{b \in \mathcal{B}} \phi(b,t) \}$, then $\{b_1,...,b_m\} \backslash \{b_i\}$ is also an optimal menu for problem \eqref{eq6}, again a contradiction to the minimality. Hence, $\{b_1,...,b_m\}$ satisfies the three conditions in Corollary \ref{c5}, which, by the ``if'' direction, indicates that $\{b_1,...,b_m\}$ is also a minimal optimal menu. Therefore, for any minimal optimal menu $B$, it has to be a minimal optimal menu for problem \eqref{eq6} as well, which in turn satisfies the three conditions by the argument above. 

 \end{proof}
 
The first condition addresses the fact that $\{b_1,...,b_m\}$ is an optimal menu, whereas the second and third condition focus on the ``minimality'' : every bundle has to be a best response to certain types and cannot be replaced by other bundles in the menu. Note that this is different from saying every bundle in the menu has to be a strict best response to certain types, which is not a correct characterization of minimal optimal menus under the condition that $\forall b,b' \in \mathcal{B}, \phi(b,t)-\phi(b',t)$ is monotonic in $t$: Example \ref{e6} satisfies this condition, and $\{b,b'\}$ is a minimal optimal menu, but both $b$ and $b'$ are never strict best responses.

It is true that the above corollary focuses solely on $\phi$ and has nothing to do with the dominance notion. However, this costs the loss of uniqueness of the minimal optimal menu, leading to a less straightforward characterization. The corollary below brings back Assumption \ref{a1}, coming up with a cleaner characterization. However, the proof (and the uniqueness of the minimal optimal menu) largely relies on the equivalence between very weakly dominance and never strict best response. Moreover, both corollaries are only characterizations but cannot concretely produce a minimal optimal menu like Algorithm \ref{al1} does, which is why discussing the notion of dominance is more worthwhile than only focusing on assumptions on $\phi$. 
\begin{corollary}
\label{c6}
    Under Assumptions \ref{a0} and \ref{a1}, $\{b_1,...,b_m\}$ is the minimal optimal menu if and only if 

    \begin{itemize}

    \item $ \{t| \exists b_i \in \{b_1,...,b_m\}, b_i \in \argmax \limits_{b \in \mathcal{B}} \phi(b,t) \} =T$
    
   \item $\forall i \in \{1,...,m\}, \exists \tilde{t} \in T$ such that $\phi(b_i,\tilde{t})> \max \limits_{b \neq b_i} \phi(b,\tilde{t})$

    \end{itemize}
\end{corollary}

 \begin{proof}

 Under Assumption \ref{a0} and \ref{a1}, by Theorem \ref{t2}, $\{b_1,...,b_m\}$ is the minimal optimal menu if and only if $\{b_1,...,b_m\}$ is the minimal optimal menu of problem \eqref{eq6}. By Proposition \ref{p4}, $\{b_1,...,b_m\}$ is the minimal optimal menu of problem \eqref{eq6} if and only if $\{b_1,...,b_m\}$ is an optimal menu of problem \eqref{eq6} and $\forall i \in \{1,...,m\}, b_i$ is not very weakly dominated. Note that $\{b_1,...,b_m\}$ is an optimal menu of problem \eqref{eq6} if and only if $ \{t| \exists b_i \in \{b_1,...,b_m\}, b_i \in \argmax \limits_{b \in \mathcal{B}} \phi(b,t) \} =T$, and by Proposition \ref{p1}, $\forall i \in \{1,...,m\}, b_i$ is not very weakly dominated if and only if $\forall i \in \{1,...,m\}, \exists \tilde{t} \in T$ such that $\phi(b_i,\tilde{t})> \max \limits_{b \neq b_i} \phi(b,\tilde{t})$. The result thus follows.
 
\end{proof}

Corollary \ref{c6} resembles the characterization that Proposition \ref{p4} introduces: the first condition is the same as the one in Corollary \ref{c5}, which still implies that $\{b_1,...,b_m\}$ is an optimal menu. The second condition says that every $b_i$ is a strict best response, which by the equivalence between strict best response and very weakly dominance, says that every $b_i$ is not very weakly dominated. 
 
This further implies that, under Assumption \ref{a0} and \ref{a1}, the three conditions raised in Corollary \ref{c5} is equivalent to the two conditions raised in Corollary \ref{c6}. Note that this is something incorrect in general, in particular, $\forall i \in \{1,...,m\}, \{t| b_i \in \argmax \limits_{b \in \mathcal{B}} \phi(b,t) \} \neq \emptyset$ and $\forall i \in \{1,...,m\}, \bigcup \limits_{j \neq i} [\{t|b_i \in \argmax \limits_{b \in \mathcal{B}} \phi(b,t) \} \cap \{t| b_j \in \argmax \limits_{b \in \mathcal{B}} \phi(b,t) \} ]\neq \{t| b_i \in \argmax \limits_{b \in \mathcal{B}} \phi(b,t) \}$ do not imply $\forall i \in \{1,...,m\}, \exists \tilde{t} \in T$ such that $\phi(b_i,\tilde{t})> \max \limits_{b \neq b_i} \phi(b,\tilde{t})$ even when $\forall b, b' \in \mathcal{B}, \phi(b,t)-\phi(b',t)$ is monotonic in $t$. For instance, in Example \ref{e6}, it is not hard to check that $\forall b, b' \in \mathcal{B}, \phi(b,t)-\phi(b',t)$ is monotonic in $t$. However, $\{b,b'\}$ is a menu that satisfies the last conditions (in fact, all three) in Corollary \ref{c5}, but it does not satisfy the second condition in Corollary \ref{c6}, i.e., neither of $b$ or $b'$ is a strict best response. The reason is again the violation of Assumption \ref{a1}. 

Having Assumption \ref{a1} not only brings back the equivalence between never strict best response and very weakly dominance, but also brings back Algorithm \ref{al1}, which helps develop a more detailed characterization of the minimal optimal menu. See the corollary below.

\begin{corollary}
\label{c7}
Under Assumption \ref{a0} and \ref{a1}, $\{b_1,...,b_m\}$ is the minimal optimal menu if and only if
\begin{itemize}
    \item $ \{t| \exists b_i \in \{b_1,...,b_m\}, b_i \in \argmax \limits_{b \in \mathcal{B}} \phi(b,t) \} =T$
    \item $\exists \rho: \{1,...,m\} \rightarrow \{1,...,m\}$ bijective such that $\phi(b_{\rho(1)},\underline{t})<...<\phi(b_{\rho(m)},\underline{t}), \phi(b_{\rho(1)},\bar{t})>...>\phi(b_{\rho(m)},\bar{t})$, and $\forall i \in \{2,...,m-1\}$, \begin{equation} \label{eq89} \dfrac{\phi(b_{\rho({i-1})},\bar{t})-\phi(b_{\rho(i)},\bar{t})}{\phi(b_{\rho(i)},\bar{t})-\phi(b_{\rho({i+1})},\bar{t})} < \dfrac{\phi(b_{\rho({i-1})},\underline{t})-\phi(b_{\rho(i)},\underline{t})}{\phi(b_{\rho(i)},\underline{t})-\phi(b_{\rho({i+1})},\underline{t})}. \end{equation}
\end{itemize}
\end{corollary}

\begin{proof}
    Want to show that the two conditions raised in Corollary \ref{c7} hold if and only if $\forall b \in \mathcal{B} \backslash \{b_1,...,b_m\}, b$ is very weakly dominated, and $\forall i \in \{1,...,m\}$, $b_i$ is not very weakly dominated. The corollary then follows Proposition \ref{p4}.

    ``if'': If $\forall b \in \mathcal{B} \backslash \{b_1,...,b_m\}, b$ is very weakly dominated, and $\forall i \in \{1,...,m\}$, $b_i$ is not very weakly dominated, then by Proposition \ref{p4}, $\{b_1,...,b_m\}$ is the minimal optimal menu for problem \eqref{eq6}, which immediately gives $\{t| \exists b_i \in \{b_1,...,b_m\}, b_i \in \argmax \limits_{b \in \mathcal{B}} \phi(b,t)\}=T$. Let the bijection $\rho: \{1,...,m\} \rightarrow \{1,...,m\}$ satisfy $\phi(b_{\rho(1)},\underline{t}) \leq ... \leq \phi(b_{\rho(m)},\underline{t})$. If there exists $i \in \{1,...,m-1\}$ such that $\phi(b_{\rho(i)},\underline{t})=\phi(b_{\rho(i+1)},\underline{t})$, then by Corollary \ref{c2}, one of $b_{\rho(i)}$ and $b_{\rho(i+1)}$ is very weakly dominated by the other, a contradiction. In that case, $\phi(b_{\rho(1)},\underline{t}) < ... < \phi(b_{\rho(m)},\underline{t})$, which in turn gives $\phi(b_{\rho(1)},\bar{t}) > ... > \phi(b_{\rho(m)},\bar{t})$. Meanwhile, by Theorem \ref{t1}, $\{b_1,...,b_m\}$ is the output of Algorithm \ref{al1}, which directly leads to (\ref{eq89}).

    ``only if'': $\{t| \exists b_i \in \{b_1,...,b_m\}, b_i \in \argmax \limits_{b \in \mathcal{B}} \phi(b,t)\}=T$ implies that $\forall b \in \mathcal{B} \backslash \{b_1,...,b_m\}$, $b$ is a never strict best response, which, by Proposition \ref{p1}, further implies that $b$ is very weakly dominated. Consider $\mathcal{B}'=\{b_1,...,b_m\}$. Run Algorithm \ref{al1} on $\mathcal{B}$', and the output is exactly $\mathcal{B}$', which implies that $b_i$ is not very weakly dominated by any convex combination of $b_1,...,b_m$, $\forall i \in \{1,..,m\}$. This, together with every bundle not included in $\{b_1,...,b_m\}$ is very weakly dominated,  further implies that $b_i$ is not very weakly dominated for all $i \in \{1,...,m\}$.

\end{proof}

The second condition says that $\forall i \in \{1,...,m\}, b_i$ is not very weakly dominated by any convex combination of $b_1,...,b_m$, which is justified by Algorithm \ref{al1} and Theorem \ref{t2}. As the first condition also implies that every bundle not included in $\{b_1,...,b_m\}$ is not very weakly dominated, the second condition further implies that every $b_i$ is not very weakly dominated. 

If there does not exist very weakly dominated bundles, by Proposition \ref{p4}, $\mathcal{B}$ is an optimal menu for problem \eqref{eq6}. Suppose that there exists $B \subsetneq \mathcal{B}$ such that $B$ is also an optimal menu for problem \eqref{eq6}, consider $b \in \mathcal{B} \backslash B$. Since $b$ is not very weakly dominated, by Proposition \ref{p1}, $b$ is a strict best response, i.e., $\exists t \in T$ such that $\phi(b,t)>\max \limits_{b' \neq b} \phi(b',t)$. In particular, $\phi(b,t)>\max \limits_{b' \in B} \phi(b',t)$, a contradiction to $B$ being an optimal menu. Hence, $\mathcal{B}$ is a minimal optimal menu.

If there exists very weakly dominated bundles:

    ``only if'': If $\{b_1,..,b_m\}$ is a minimal optimal menu, then by Proposition \ref{p4}, every weakly dominated bundle is very weakly dominated by a convex combination of $b_1,...,b_m$. In addition, if there exists $b_i \in \{b_1,...,b_m\}$ such that $b_i$ is very weakly dominated, then by Proposition \ref{p1}, it is a never strict best response, so that $\{b_1,...,b_{i-1},b_{i+1},...,b_m\}$ is also an optimal menu, a contradiction to $\{b_1,...,b_m\}$ being minimal. Hence, $\forall i \in \{1,...,m\}, b_i$ is not very weakly dominated.   

    ``if'': If every weakly dominated bundle is very weakly dominated by a convex combination of $b_1,...,b_m$, then by Proposition \ref{p4}, $\{b_1,...,b_m\}$ is an optimal menu for problem \eqref{eq6}. Suppose $\{b_1,...,b_m\}$ is not a minimal optimal menu, then there exists $B \subsetneq \{b_1,...,b_m\}$ such that $B$ is an optimal menu for problem \eqref{eq6}. Consider $b' \in B \backslash \{b_1,...,b_m\}$, as $b'$ is not very weakly dominated, by Proposition \ref{p1}, $b'$ is a strict best response, i.e., for any $b \in B$, there exists $t \in T$ (does not depend on $b$) such that $\phi(b',t) > \phi(b,t)$, a contradiction to $B$ being optimal. Thus, by contradiction, $\{b_1,...,b_m\}$ is a minimal optimal menu for problem \eqref{eq6}. 

\subsection{Assumption \ref{a1} and its Alternatives}  
\label{ae4}
I first discuss the connection between Assumption \ref{a1} and the following condition:
\begin{center}
    for any $b \in \mathcal{B}, a \in \Delta(\mathcal{B}), \phi(b,t)-\phi(a,t)$ is single-crossing in $t$
\end{center}

See below an example where for any $b \in \mathcal{B}, a \in \Delta(\mathcal{B}), \phi(b,t)-\phi(a,t)$ is single-crossing in $t$, but Assumption \ref{a1} fails to hold.\footnote{I thank Daniel Rappoport for a similar example.}

\begin{example} 
Consider the following example: $T=[-1,1],\mathcal{B}=\{b,b',b'',b'''\}, \phi(b,t)=5, \phi(b',t)=-5, \phi(b'',t)=-t$ \[ \phi(b''',t)= \begin{cases} -t^2 & \text{if } t \leq 0, \\ 0 & \text{if } t > 0. \end{cases} \] \

I first show that Assumption \ref{a1} does not hold in this example. Consider $a=\dfrac{49}{100}b+\dfrac{51}{100}b'$ and $a'=\dfrac{1}{3}b''+\dfrac{2}{3}b'''$. $\phi(a,-1)-\phi(a',-1)=-\dfrac{1}{10}-(-\dfrac{1}{3})>0; \phi(a,0)-\phi(a',0)=-\dfrac{1}{10}<0; \phi(a,1)-\phi(a',1)=-\dfrac{1}{10}-(-\dfrac{1}{3})>0$, a violation of Assumption \ref{a1}.

    \begin{figure}[H]
        \centering
        \includegraphics[width=0.4\linewidth]{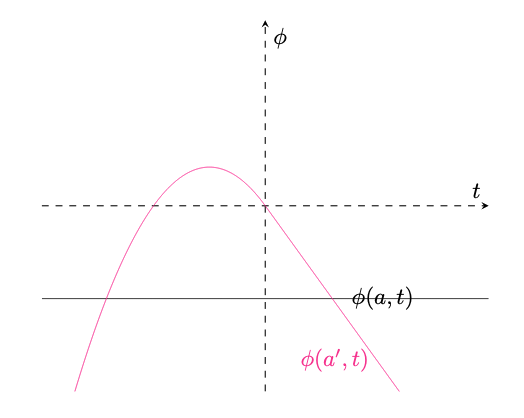}
        \caption{$\phi(a',t)-\phi(a,t)$ is not single-crossing in $t$}
    \end{figure}

I proceed to show that for any $b \in \mathcal{B}, a \in \Delta(\mathcal{B}), \phi(b,t)-\phi(a,t)$ is single-crossing in $t$ is indeed satisfied. Since for any $t \in T, a \in \Delta(\mathcal{B}), \phi(b',t) \leq \phi(a,t) \leq \phi(b,t)$, $\phi(b,t)-\phi(a,t)$ and $\phi(b',t)-\phi(a,t)$ are single-crossing in $t$ for any $a \in \Delta(\mathcal{B})$. Meanwhile, for any $t \in T,$ for any $ a \in \Delta(\mathcal{B}), \phi_t(b'',t) \leq \phi_t(a,t) \leq \phi_t(b''',t), $, indicating that $\phi(b'',t)-\phi(a,t)$ and $\phi(b''',t)-\phi(a,t)$ are single-crossing in $t$ for any $a \in \Delta(\mathcal{B})$.

This example also shows that even for any $b \in \mathcal{B}, a \in \Delta(\mathcal{B}), \phi(b,t)-\phi(a,t)$ is single-crossing in $t$ and for any $b,b' \in \mathcal{B}, \phi(b,t)-\phi(b',t)$ is monotonic in $t$ together is not a sufficient condition for Assumption \ref{a1}, for the monotonicity of the difference of virtual values among deterministic bundles is also satisfied in the example. 
    
\end{example}

However, the condition above is a sufficient condition of Proposition \ref{p1}. The result is formalized by the following proposition. 
\begin{proposition}
\label{p10}
    If for any $b \in \mathcal{B}, a \in \Delta(\mathcal{B}), \phi(b,t)-\phi(a,t)$ is single-crossing in $t$, then $b \in \mathcal{B}$ is a never strict best response if and only if $b$ is very weakly dominated.
\end{proposition}

\begin{proof}
    For any $b \in \mathcal{B}, a \in \Delta(\mathcal{B})$, let $u^b(a,t)=\phi(b,t)-\phi(a,t)$. I first claim that $\forall b \in \mathcal{B}, \min \limits_{a \in \Delta(\mathcal{B})} \max \limits_{t \in T} u^b(a,t) \leq 0$.\footnote{$\min$ and $\max$ are well-defined because $\Delta(\mathcal{B})$ and $T$ are compact, and $u^b$ is continuous in both components.} This is because for $a=b, \max \limits_{t \in T} u^b(b,t)=0$, so that $\min \limits_{a \in \Delta(\mathcal{B})} \max \limits_{t \in T} u^b(a,t) \leq 0$. In that case, for $a^* \in \argmin \limits_{a \in \Delta(\mathcal{B})} \max \limits_{t \in T} u^b(a,t), \phi(b,t) \leq \phi(a^*,t), \forall t$.

    The only thing left to show is if $b$ is a never strict best response, and for any $b \in \mathcal{B}, a \in \Delta(\mathcal{B}), \phi(b,t)-\phi(a,t)$ is single-crossing in $t$, then $\argmin \limits_{a \in \Delta(\mathcal{B})} \max \limits_{t \in T} u^b(a,t) \neq \{b\}$. This result follows from the proof of Proposition \ref{p1}.
\end{proof}

The two arguments above together imply that for any $b \in \mathcal{B}, a \in \Delta(\mathcal{B}), \phi(b,t)-\phi(a,t)$ is single-crossing in $t$ is a weaker sufficient condition for Proposition \ref{p1} than Assumption \ref{a1}, yet it is less structured in the sense that there does not exist a equivalent characterization on $\phi$. Moreover, this condition, unlike Assumption \ref{a1}, cannot be a good starting point to build connections between single-crossing differences and monotonic differences. 

In addition, it is unclear whether the equivalence between ``strictly dominated deterministic bundle'' and ``never best response deterministic bundle'' (analogous to Definition \ref{def1} and \ref{def2}) only requires the deterministic - lottery single-crossing difference but not Assumption \ref{a1}. The argument for this equivalence would be very similar to the proof of Lemma 3 of \cite{pearce1984}, which does need minmax theorem, thus Assumption \ref{a1}, to play a significant role. 

It is also worth mentioning that Example \ref{e2} can be used to demonstrate that for any $b,b' \in \mathcal{B}, \phi(b,t)-\phi(b',t)$ is single-crossing in $t$ does not imply that for any $b \in \mathcal{B}, a \in \Delta(\mathcal{B}), \phi(b,t)-\phi(a,t)$ is single-crossing in $t$.

Next I consider the following condition:
\begin{center}
    $\text{ for any }b \in \mathcal{B}, a \in \Delta(\mathcal{B}), \phi(b,t)-\phi(a,t)$ is monotonic in $t$
\end{center}

It is not hard to see that Assumption \ref{a1} does not imply the condition above. For example, when $|\mathcal{B}|=2$, if $\phi(b_1,t)-\phi(b_2,t)$ is single-crossing but not monotonic in $t$, then for any $a,a' \in \Delta(\{b_1,b_2\}), \phi(a,t)-\phi(a',t)$ is single-crossing in $t$, but for any $b \in \mathcal{B}, a \in \Delta(\mathcal{B}), \phi(b,t)-\phi(a,t)$ is only single-crossing yet not monotonic in $t$. Example \ref{e3} below provides a less trivial illustration. 

\begin{example}
\label{e3}
For any $a=a_1b+a_2b'+(1-a_1-a_2)b'' \in \Delta(\mathcal{B}), \phi(a,t)=5a_1+(3-t)a_2+(1-a_1-a_2)(1-\dfrac{1}{4}(1-t)^2)=a_1(4+\dfrac{1}{4}(1-t)^2)+a_2(2-t+\dfrac{1}{4}(1-t)^2)+(1-\dfrac{1}{4}(1-t)^2)$. Consider $\dfrac{4+\dfrac{1}{4}(1-t)^2}{2-t+\dfrac{1}{4}(1-t)^2}=\dfrac{t^2-2t+17}{(t-3)^2}$. It is straightforward that $t^2-2t+17>0, (t-3)^2>0$.\footnote{This also implies that both $4+\dfrac{1}{4}(1-t)^2$ and $2-t+\dfrac{1}{4}(1-t)^2$ are single-crossing.} Moreover, $(\dfrac{t^2-2t+17}{(t-3)^2})'=\dfrac{(2t-2)(t-3)^2-2(t-3)(t^2-2t+17)}{(t-3)^4}=\dfrac{-12t-28}{(t-3)^3}$, which is always positive on $[0,2]$. Hence, $4+\dfrac{1}{4}(1-t)^2$ and $2-t+\dfrac{1}{4}(1-t)^2$ are (strictly) ratio ordered on $[0,2]$.\footnote{See Definition 3 of \cite{kartik2023} for the definition of (strictly) ratio ordered functions.} By Theorem 2 of \cite{kartik2023}, for any $a,a' \in \Delta(\mathcal{B})$, $\phi(a,t)-\phi(a',t)$ is single-crossing in $t$ on $[0,2]$. However, $\phi(b,t)-\phi(b'',t)$ is not monotonic in $t$, a violation of the condition above.
\end{example}

The more interesting question is whether the above condition can imply Assumption \ref{a1}. The following proposition answers this question.

\begin{proposition}
\label{p2}
    If for any $b \in \mathcal{B}, a \in \Delta(\mathcal{B}), \phi(b,t)-\phi(a,t)$ is monotonic in $t$, then for any $a,a' \in \Delta(\mathcal{B}), \phi(a,t)-\phi(a',t)$ is monotonic in $t$.
\end{proposition}

\begin{proof}
        Let $\mathcal{B}=\{b_1,...,b_{2^n}\}$. Claim that the problem regarding $(\phi(b_1,t),...,\phi(b_{2^n},t))$ is equivalent to the problem regarding $(\tilde{\phi}(b_1,t),...,\tilde{\phi}(b_{2^n},t))$, in which $\tilde{\phi}(b_1,t)=0, \tilde{\phi}(b_i,t)=\phi(b_i,t)-\phi(b_1,t), \forall i \in \{2,...,2^n\}$. This is because for any $a=\sum \limits_{i=1}^{2^n} \lambda_i b_i$ and $a' = \sum \limits_{i=1}^{2^n} \lambda_i' b_i$, \begin{equation*} \begin{split} \phi(a,t)-\phi(a',t) &=\sum \limits_{i=2}^{2^n} \lambda_i [\phi(b_i,t)-\phi(b_1,t)]+\phi(b_1,t)-\sum \limits_{i=2}^{2^n} \lambda_i' [\phi(b_i,t)-\phi(b_1,t)]-\phi(b_1,t) \\ &= \sum \limits_{i=1}^{2^n} (\lambda_i-\lambda_i')[\phi(b_i,t)-\phi(b_1,t)]=\tilde{\phi}(a,t)-\tilde{\phi}(a',t). \end{split} \end{equation*} Consider the problem regarding $\tilde{\phi}$. Since $\forall i \in \{2,...,2^n\}, \tilde{\phi}(b_i,t)-\tilde{\phi}(b_1,t)$ is monotonic in $t$ and $\tilde{\phi}(b_1,t)=0$, $\tilde{\phi}(b_i,t)$ is monotonic in $t, \forall i \in \{2,...,2^n\}$.

    I first focus on the case where $\tilde{\phi}(b_i,t)$ is strictly monotonic, $\forall i \in \{2,...,2^n\}$. For this case, I claim that $\forall i,j \in \{2,...,2^n\}, \dfrac{\tilde{\phi}_t(b_i,t)}{\tilde{\phi}_t(b_j,t)}$ is a constant (which can depend on $i$ and $j$). I show this result by contradiction. Suppose this is not the case, i.e., there exists $i,j,t_1,t_2$ such that $\dfrac{\tilde{\phi}_t(b_i,t_1)}{\tilde{\phi}_t(b_j,t_1)} \neq \dfrac{\tilde{\phi}_t(b_i,t_2)}{\tilde{\phi}_t(b_j,t_2)}$. Without loss of generality, let $\dfrac{\tilde{\phi}_t(b_i,t_1)}{\tilde{\phi}_t(b_j,t_1)} < \dfrac{\tilde{\phi}_t(b_i,t_2)}{\tilde{\phi}_t(b_j,t_2)}$. If $\forall t, \tilde{\phi}_t(b_i,t) \tilde{\phi}_t(b_j,t) >0$, i.e., both $\tilde{\phi}(b_i,t)$ and $\tilde{\phi}(b_j,t)$ are either increasing or decreasing, then there exists $c>0$ such that $\dfrac{\tilde{\phi}_t(b_i,t_1)}{\tilde{\phi}_t(b_j,t_1)} < c < \dfrac{\tilde{\phi}_t(b_i,t_2)}{\tilde{\phi}_t(b_j,t_2)}$. It is further without loss to assume that $\forall t, \tilde{\phi}(b_i,t)>0, \tilde{\phi}(b_j,t)>0$ and $c \leq 1$.\footnote{If $c>1$, consider $\frac{\tilde{\phi}_t(b_j,t_1)}{\tilde{\phi}_t(b_i,t_1)} > \frac{1}{c} > \frac{\tilde{\phi}_t(b_j,t_2)}{\tilde{\phi}_t(b_i,t_2)}$.} The inequality above gives $\tilde{\phi}_t(b_i,t_1)-c \tilde{\phi}_t(b_j,t_1)<0; \tilde{\phi}_t(b_i,t_2)-c \tilde{\phi}_t(b_j,t_2)>0$. However, $\tilde{\phi}(b_i,t)-\tilde{\phi}(cb_j+(1-c)b_1,t)$ is monotonic in $t$, which gives $[\tilde{\phi}_t(b_i,t_1)-\tilde{\phi}_t(cb_j+(1-c)b_1,t_1)][\tilde{\phi}_t(b_i,t_2)-\tilde{\phi}_t(cb_j+(1-c)b_1,t_2)]=[\tilde{\phi}_t(b_i,t_1)-c \tilde{\phi}_t(b_j,t_1)][\tilde{\phi}_t(b_i,t_2)-c \tilde{\phi}_t(b_j,t_2)]\geq 0$, a contradiction.

    If $\forall t, \tilde{\phi}_t(b_i,t)\tilde{\phi}_t(b_j,t)<0$, i.e, one of $\tilde{\phi}(b_i,t)$ and $\tilde{\phi}(b_j,t)$ is increasing and the other is decreasing, then there exists $c<0$ such that $\dfrac{\tilde{\phi}_t(b_i,t_1)}{\tilde{\phi}_t(b_j,t_1)} < c < \dfrac{\tilde{\phi}_t(b_i,t_2)}{\tilde{\phi}_t(b_j,t_2)}$. It is still without loss to assume that $\forall t, \tilde{\phi}(b_i,t)>0$. The inequality above once again gives $\tilde{\phi}_t(b_i,t_1)-c \tilde{\phi}_t(b_j,t_1)<0; \tilde{\phi}_t(b_i,t_2)-c \tilde{\phi}_t(b_j,t_2)>0$. However, $\tilde{\phi}(\dfrac{1}{1-c}b_i+\dfrac{-c}{1-c}b_j,t)-\tilde{\phi}(b_1,t)$ is monotonic in $t$, which gives $\frac{1}{1-c}[\tilde{\phi}_t(b_i,t_1)-c\tilde{\phi}_t(b_j,t_1)][\tilde{\phi}_t(b_i,t_2)-c\tilde{\phi}_t(b_j,t_2)] \geq 0$, a contradiction.\footnote{$\frac{1}{1-c}b_i+\frac{-c}{1-c}b_j$ is a convex combination of $b_i$ and $b_j$ because $c<0$ implies that $0<\frac{1}{1-c}<1, 0<\frac{-c}{1-c}<1$.} 

    Hence, the initial claim holds. With the claim, it is not hard to see that $\forall i \in \{3,...,2^n\}, \exists c_i,c_i' \in \mathbb{R}$ such that $\tilde{\phi}(b_i,t)=c_i\tilde{\phi}(b_2,t)+c_i'$.\footnote{A way to show this is $\forall t, i \in \{3,...,n\}, \tilde{\phi}(b_i,t)-\tilde{\phi}(b_i,\underline{t})=\int_{\underline{t}}^t \tilde{\phi}_t(b_i,s)ds=c_i\int_{\underline{t}}^t \tilde{\phi}_t(b_2,s)ds=c_i[\tilde{\phi}(b_2,t)-\tilde{\phi}(b_2,\underline{t})]$, in which $c_i$ is the constant given by $\dfrac{\tilde{\phi}_t(b_i,t)}{\tilde{\phi}_t(b_2,t)}$.}As a result, $\forall a,a' \in \Delta(B), \tilde{\phi}(a,t)-\tilde{\phi}(a',t)$ is a affine transformation of $\tilde{\phi}(b_2,t)$, which is monotonic in $t$ because $\tilde{\phi}(b_2,t)$ is monotonic in $t$.

    Now I proceed to discuss the case where some of the $\tilde{\phi}(b_i,t)$ may not be strictly monotonic. I first show that there do not exist $i,j,t_1,t_2$ such that $\tilde{\phi}_t(b_i,t_1)=0, \tilde{\phi}_t(b_j,t_1) \neq 0; \tilde{\phi}_t(b_i,t_2) \neq 0, \tilde{\phi}_t(b_j,t_2) =0$. Suppose this is the case, Without loss of generality, let $\tilde{\phi}_t(b_j,t_1)>0$. For $\lambda \in (0,1), \tilde{\phi}(\lambda b_i+(1-\lambda)b_j, t)-\tilde{\phi}(b_1,t)$ is monotonic in $t$ implies that $\tilde{\phi}_t(b_i,t_2)>0$. Meanwhile, $\tilde{\phi}(b_i,t)-\tilde{\phi}(\lambda b_j+(1-\lambda) b_1,t)$ is monotonic in $t$ implies that $\tilde{\phi}_t(b_i,t_2)<0$, a contradiction. Next I claim that if there exist $i,j,t_1$ such that $\tilde{\phi}_t(b_i,t_1)=0, \tilde{\phi}_t(b_j,t_1) \neq 0$, then $\tilde{\phi}_t(b_i,t)=0, \forall t$, i.e., $\tilde{\phi}(b_i,t)$ is a constant function. Suppose this is not the case, then there exist $t_2$ such that $\tilde{\phi}_t(b_i,t_2) \neq 0$. By the result above, $\tilde{\phi}_t(b_j,t_2) \neq 0$. Without loss of generality, let $\tilde{\phi}_t(b_j,t_1)>0$, then by monotonicity and $\tilde{\phi}_t(b_j,t_2) \neq 0$, $\tilde{\phi}_t(b_j,t_2) > 0$. If $\tilde{\phi}_t (b_i,t_2)<0$, then there exists $\lambda \in (0,1)$ such that $\lambda \tilde{\phi}_t(b_j,t_2)+(1-\lambda) \tilde{\phi}_t (b_i,t_2)<0$. On the other hand, $\tilde{\phi}(\lambda b_j+(1-\lambda) b_i,t)-\tilde{\phi}(b_i,t)=\tilde{\phi}(\lambda b_j+(1-\lambda) b_i,t)$ is monotonic in $t$. $\lambda \tilde{\phi}_t(b_j,t_1)+(1-\lambda) \tilde{\phi}_t (b_i,t_1)=\lambda \tilde{\phi}_t(b_j,t_1) >0$ implies that $\lambda \tilde{\phi}_t(b_j,t_2)+(1-\lambda) \tilde{\phi}_t (b_i,t_2)>0$, a contradiction. If $\tilde{\phi}_t(b_i,t_2)>0$, then there exists $\lambda' \in (0,1)$ such that $\lambda' \tilde{\phi}_t (b_j,t_2)-\tilde{\phi}_t (b_i,t_2)<0$. On the other hand, $\tilde{\phi}(\lambda' b_j+(1-\lambda') b_1,t)-\tilde{\phi}(b_i,t)=\lambda' \tilde{\phi}(b_j,t)-\tilde{\phi}(b_i,t)$ is monotonic in $t$. $\lambda' \tilde{\phi}_t (b_j,t_1)-\tilde{\phi}_t(b_i,t_1)=\lambda' \tilde{\phi}_t (b_j,t_1)>0$ implies that $\lambda' \tilde{\phi}_t (b_j,t_2)-\tilde{\phi}_t (b_i,t_2)>0$, again a contradiction. Hence, if there exist $i,j,t_1$ such that $\tilde{\phi}_t(b_i,t_1)=0, \tilde{\phi}_t(b_j,t_1) \neq 0$, then $\tilde{\phi}_t(b_i,t)=0, \forall t$. It is not hard to see that for such a $b_i$, $\tilde{\phi}(b_i,t)$ can always be a affine transformation of $\tilde{\phi}(b_{i'},t)$ for $i' \neq i$.

    The only case left to discuss is when there exists $i,j \in \{2,...,2^n\}, i \neq j$ such that $\{ t| \tilde{\phi}_t(b_i,t)=0\}=\{ t| \tilde{\phi}_t(b_j,t)=0\}$, which in turn suggests that $\{ t| \tilde{\phi}_t(b_i,t) \neq 0\}=\{ t| \tilde{\phi}_t(b_j,t) \neq 0\}$. By the argument made regarding strict monotonicity, for any $t \in \{ t| \tilde{\phi}_t(b_i,t) \neq 0\}=\{ t| \tilde{\phi}_t(b_j,t) \neq 0\}$, $\dfrac{\tilde{\phi}_t(b_i,t)}{\tilde{\phi}_t(b_j,t)}$ is a constant, i.e., $\exists c_i,c_i' \in \mathbb{R}$ such that $\tilde{\phi}(b_i,t)=c_i \tilde{\phi}(b_j,t)+c_i'$. By continuity, for any $t \in \{ t| \tilde{\phi}_t(b_i,t) = 0\}=\{ t| \tilde{\phi}_t(b_j,t) = 0\}$, $\dfrac{\tilde{\phi}_t(b_i,t)}{\tilde{\phi}_t(b_j,t)}$, it is still the case that $\tilde{\phi}(b_i,t)=c_i \tilde{\phi}(b_j,t)+c_i'$, so that this affine transformation relationship holds for  all $t$.

    In conclusion, if for all $i \in \{2,...,2^n\}, \tilde{\phi}_t(b_i,t)=0, \forall t$, then the result holds trivially. Otherwise, there exists $\bar{i} \in \{2,...,n\}$ such that $\exists \bar{t}, \tilde{\phi}_t(b_{\bar{i}}, \bar{t}) \neq 0$, and $\forall i \neq \bar{i}, \tilde{\phi}(b_i,t)$ is a affine transformation of $\tilde{\phi}(b_{\bar{i}},t)$, which further implies the monotonicity of $\tilde{\phi}(a,t)-\tilde{\phi}(a',t)$ for all $a,a' \in \Delta(B)$ by the monotonicity of $\tilde{\phi}(b_{\bar{i}},t)$.
\end{proof}

Note that the proof is consistent with Lemma 2 of \cite{kartik2023}, which says any linear combinations of two monotonic functions is monotonic if and only if one function is an affine transformation of the other. 

Hence, the ``deterministic''-``stochastic'' monotonic difference condition is stronger than Assumption \ref{a1}, which makes it redundant for Proposition \ref{p1}, but the proposition above shows that it is exactly the condition needed for Theorem \ref{t2}.

\end{spacing}

\end{document}